\newcommand{\thetitle}{Maximizing Index Diversity in Committee Elections}
\def\HUBaffil{Humboldt-Universität zu Berlin, 
Department of Computer Science, Algorithm Engineering Group, Germany}
\def\TUCaffil{Technische Universität Clausthal, Institut für Informatik, Germany}

\def\thxPACSs{Deutsche Forschungsgemeinschaft (DFG, German
Research Foundation), project PACS (FL~1247/1-1, 522475669).}

\documentclass[10pt,twocolumn]{article}
\usepackage[a4paper,includefoot,margin=1.5cm]{geometry}
\usepackage{amsmath, amssymb, amsthm}

\usepackage{xcolor}
\usepackage{xcolor-material}
\usepackage[T1]{fontenc}

\usepackage[%
  breaklinks=true,
  pdfdisplaydoctitle,
  menucolor=orange!40!black,filecolor=magenta!40!black,urlcolor=blue!40!black,linkcolor=red!40!black,citecolor=green!40!black,
  colorlinks
  ]{hyperref} %
\DeclareUnicodeCharacter{1EF3}{\`y}

\usepackage{natbib}
\makeatletter
\def\@maketitle{%
\newpage\null\vskip 2em%
\begin{center}%
\let \footnote \thanks
  {\Large\bf \@title \par}\vskip 1.5em{\large\lineskip .5em\begin{tabular}[t]{c}\@author\end{tabular}\par}\vskip 1em{\large \@date}%
\end{center}\par\vskip 1.5em
}
\usepackage{url}
\def\url@leostyle{\@ifundefined{selectfont}{\def\UrlFont{\sf}}{\def\UrlFont{\small\ttfamily}}}
\makeatother
\usepackage{authblk}
\author[1]{Paula Böhm}
\author[1]{Robert Bredereck}
\author[2]{Till~Fluschnik}

\affil[1]{\footnotesize
  \TUCaffil\\
  \texttt{$\{$paula.boehm,robert.bredereck$\}$@tu-clausthal.de}
}
\affil[2]{\footnotesize
  \HUBaffil\\
  \texttt{till.fluschnik@hu-berlin.de}
}

\AtBeginEnvironment{tblr}{\footnotesize}

\usepackage{nameref}
\usepackage{mathtools}
\usepackage{xargs}
\usepackage[capitalise]{cleveref}

\usepackage{stmaryrd}
\usepackage{subcaption}
\usepackage{pgfplots}
\pgfplotsset{compat=1.18}
\usepackage{tabularray}
\usepackage{enumitem}
\UseTblrLibrary{amsmath,booktabs,tikz}
\usepackage{tikz}
\usepackage{twemojis}
\AtBeginDocument{
  \setlength{\twemojiDefaultHeight}{0.9em}
}

\title{\thetitle}
\newcommand{\theabstract}{%
  \looseness -1
  We introduce two models of multiwinner elections with approval preferences and labelled candidates that take the committee's diversity into account.
  One model aims to find a committee with maximal diversity given a scoring function (e.g. of a scoring-based voting rule) and a lower bound for the score to be respected.
  The second model seeks to maximize the diversity given a minimal satisfaction for each agent to be respected.
  To measure the diversity of a committee, we use multiple diversity indices used in ecology and introduce one new index.
  We define (desirable) properties of diversity indices, test the indices considered against these properties, and characterize the new index.
  We analyze the computational complexity of computing a committee for both models
  and scoring functions of well-known voting rules,
  and investigate the influence of weakening the score or satisfaction constraints on the diversity empirically.
}

\newcommandx{\mydefenv}[4][3=A]{%
  \newtheorem{#1}{#2}
  \ifstrequal{#3}{A}{\crefname{#1}{#2}{#2s}}{\crefname{#1}{#2}{#3}}%
  \Crefname{#1}{#4{.}}{#4s{.}}
}
\theoremstyle{remark}

\mydefenv{example}{Example}{Ex}
\mydefenv{remark}{Remark}{Rem}

\theoremstyle{plain}
\mydefenv{corollary}{Corollary}[Corollaries]{Cor}
\mydefenv{observation}{Observation}{Obs}
\mydefenv{fact}{Fact}{Fact}
\mydefenv{theorem}{Theorem}{Theorem}
\theoremstyle{definition}
\mydefenv{definition}{Definition}{Def}
\mydefenv{property}{Property}[Properties]{P}

\newcommand{\appsymb}{{\Large $\star$}}
\newcommand{\appref}[1]{{\hyperref[proof:#1]{\appsymb}}}
\newcommand{\apprefX}[1]{{\hyperref[#1]{\appsymb}}}

\newcommand{\appendixsection}[1]{%
  \gappto{\appendixProofText}{\section{Supplementary Material for Section~\ref{#1}}\label{app:#1}}
}

\newcommand{\toappendix}[1]{%
  \gappto{\appendixProofText}
  {{
    #1
  }}
}
\newcommand{\appendixproof}[2]{%
  \gappto{\appendixProofText}
  {
    \subsection{Proof of \cref{#1}}\label{proof:#1}
    #2
  }
}

\newcommand{\cocl}[1]{\ensuremath{\mathbf{#1}}}
\newcommand{\N}{\mathbb{N}}
\newcommand{\Nzero}{\mathbb{N}_0}
\newcommand{\calE}{\mathcal{E}}
\newcommand{\calR}{\mathcal{R}}
\newcommand{\Elec}{\calE}
\newcommand{\ceq}{\coloneqq}
\newcommand{\Pot}[1]{\ensuremath{\mathcal{P}(#1)}}
\newcommand{\distri}{\operatorname{distr}}
\newcommand{\rdistriv}{\operatorname{rdistr}}
\newcommand{\rdistrividx}{\rho}
\newcommand{\Rule}{\calR}
\newcommand{\Rulebase}{\Rule_{\rm vld}}
\newcommand{\rqed}{\hfill$\triangleleft$}
\newcommand{\eqed}{\hfill$\triangleleft$}%\sqcup
\newcommand{\substruct}[1]{\paragraph*{#1.}}
\NewDocumentCommand{\pref}{om}{\textit{\hyperref[{#2}]{\IfValueTF{#1}{#1}{\nameref*{#2}}}}}
\newcommand{\prob}[1]{\textnormal{\textsc{#1}}}
\newcommand{\NP}{\cocl{NP}}
\newcommand{\classP}{\cocl{P}}

\newcommandx{\tlog}[3][1=,3=]{\fun{\log_{#1}^{#3}}(#2)}
\newcommandx{\set}[2][1=1]{\ensuremath{\left\{#1,\ldots,#2\right\}}}
\newcommand{\Com}{\ensuremath{S}}
\newcommand{\ExCom}{\ensuremath{S'}}

\DeclareMathOperator{\argmax}{\arg\,\max}
\DeclareMathOperator{\argmin}{\arg\,\min}
\NewDocumentCommand\bcks{d[]}{\left[ #1 \right]}
\NewDocumentCommand\brcs{m}{\left\{ #1 \right\}}
\NewDocumentCommand\paren{d()}{\left( #1 \right)}
\NewDocumentCommand\setFP{d()}{\fun{C_{\mathit{label}}}( #1 )}
\NewDocumentCommand\muDO{d()}{\fun{\mu}( #1 )}
\NewDocumentCommand\fun{md()}{#1 \mathopen{}\left( #2 \right)}
\NewDocumentCommand\funbc{mm}{#1 \mathopen{}\left\{ #2 \right\}}
\NewDocumentCommand\vc{md()}{\fun{v^{\paren(#1)}}( #2 )}
\NewDocumentCommand\labof{}{\lambda}
\def\abs#1{\left\lvert#1\right\rvert}
\def\rich{\mathit{Ri}}
\def\shann{\mathit{Sh}}
\def\simps{\mathit{Si}}
\def\dns{\mathit{LC}}

\NewDocumentCommand{\labHealth}{}{\twemoji{face with thermometer}}%{\text{health}}
\NewDocumentCommand{\labEduc}{}{\twemoji{books}}%{\text{education}}
\NewDocumentCommand{\labSport}{}{\twemoji{muscle}}%{\text{sport}}
\let\iff\Leftrightarrow
\let\implies\Rightarrow

\newcommand{\BibTeX}{\rm B\kern-.05em{\sc i\kern-.025em b}\kern-.08em\TeX}
\newcommand{\mypar}[1]{\medskip\textbf{#1}\medspace}

\begin{document}
\maketitle 

\begin{abstract}
\theabstract{}
\end{abstract}
%%%%%%%%%%%%%%%%%%%%%%%%%%%%%%%%%%%%%%%%%%%%%%%%%%%%%%%%%%%%%%%%%%%%%%%%

\section{Introduction}
In the realm of decision making, where alternatives are chosen from a larger pool of options,
considerations of both quality and diversity become crucial.
This challenge presents itself in varied contexts, such as project funding
and academic committee selections.
For instance, consider the scenario faced by a city government engaging in participatory budgeting (PB).
Here, residents propose community projects characterized by target groups (e.g.~children, seniors)
and objectives (e.g.~environmental protection, education).\footnote{Indeed, PB instances from Pabulib used
in our experiments provide such characteristics.
Nevertheless, we focus on the plain multiwinner election scenario without prices and budget in this paper.}
Similarly, consider the formation of a university hiring committee, where the aim is to
select individuals who are not only seen as experts by the colleagues electing them,
but also contribute to a diverse range of perspectives (e.g.~scientists, non-scientific staff)
or disciplines (e.g.~math, physics).
In both cases, the task becomes to select $k$~alternatives with both
a high voter support (which we measure by a total \emph{score} or individual \emph{satisfaction})
and high level of diversity (which we measure by an \emph{index}).

Indeed, both a lower bound on the total score and individual satisfaction are meaningful depending on the application.
Consider, for example, an airline selecting $k$~films for its inflight catalogue: films are labelled by genre,
language, or region, and passenger groups (families, business travelers) supply an
aggregate preference score from historical click or usage data;
because these are past, aggregate preferences rather than current voters’ utilities,
it is natural to enforce a lower bound on the committee score (e.g.~total or average past utility)
and then maximize catalogue diversity by an index to preserve breadth and guard against sampling bias.
Conversely, when planning $k$~outreach programmes at a university, departments are actual voters
directly affected by the outcome; activities are labelled by target audience, topic or format,
and here one can require per‑department satisfaction guarantees
(for instance a minimum number of relevant offerings) while maximizing thematic diversity
to avoid overrepresentation and broaden impact.

Traditional models, as seen in many previous studies, address this by setting hard quotas or constraints,
ensuring a specific representation of specific labels.
While such an approach ensures a certain minimum diversity, it often disregards the fluid and multifaceted nature of diversity itself.
In particular, when many labels are to be respected, it seems impossible to come up with any meaningful constraints
without restricting the possible satisfaction of the voters unpredictably.
In contrast, our study introduces two models that incorporate diversity indices rather than strict constraints, while ensuring that either the total score of the committee or the satisfactions of the agents cannot be worsened arbitrarily. 
This approach, applied to the multiwinner election context, allows for a more nuanced and dynamic assessment of diversity.
By employing established ecological indices alongside a newly proposed one,
we introduce this elsewhere-established approach into the area of computational social choice.

\mypar{Our Contributions.}
We introduce two models for incorporating diversity indices into multiwinner elections with approval preferences and labelled candidates.
Our models aim to find a committee with maximal diversity given (1) a scoring function and a lower bound for the score,
or (2) a minimal satisfaction for each agent.
To measure diversity,
we adapt multiple diversity indices used in ecology and propose a new diversity index, the Lexicographic Counting Index,
which is designed to measure the diversity of a committee understandably.
We provide an analysis of the properties of the diversity indices,
including newly introduced properties,
capturing,
among others,
the ability of voters to easily understand why one solution is more diverse than another.
With these properties,
we characterize our new diversity index and are able to differentiate between any two diversity indices considered.

We also analyze the computational complexity of computing a committee for both models.
We show that, while computing maximally diverse committees (without any satisfaction goal) is polynomial-time
solvable for all indices, it is NP-hard to compute maximally diverse committees that provide some
minimum satisfaction for each voter.
The computational complexity of computing maximally diverse committees which provide some minimum score
depends
on the voting rule whose score we consider: E.g.~this problem is polynomial-solvable for each of the indices considered when using the score of Approval Voting or Satisfaction Approval Voting, but not when using the score of Chamberlin-Courant.

Our experimental results
provide insights into
the influence of weakening the score or satisfaction constraint on the diversity reached.
We find that the diversity of the committees determined by scoring-based approval rules can indeed be improved by using a diversity index:
For example, allowing the score to be reduced by $10\%$ of the optimal score increases the average percentage of the optimal diversity achieved by 12--19, depending on the index and score considered.
However, a reduction of the score by even $50\%$ does not lead to the optimal diversity in some cases.

Proofs marked with~\appsymb{} are deferred to the appendix.

\mypar{Related Work.}
Our index-based approach complements established fairness notions in
computational social choice.
For example, proportional representation axioms~\citep{AzizBCEFW17}
focus on representation guarantees,
whereas the literature on collective decision making with label-based diversity predominantly enforces diversity
via hard constraints (quotas, distributional requirements).
By contrast, this paper treats diversity as an optimizable objective via indices and studies maximizing
such indices under either an aggregate‑score lower bound or per‑agent satisfaction guarantees.
We nonetheless review and compare constraint‑based approaches below.

\looseness -1
\citet{celisijcaiFairness} and \citet{bredereck2018divconstr} introduce very similar models wherein candidates
possess (possibly structured) labels and diversity is reached through hard distributional constraints.
Their approach seeks an optimal committee that maximizes a performance score while meeting specified label occurrence requirements,
such as gender quotas.
In a similar vein, \citet{ianovski2022electingDomConstraints} explores a model that accommodates ``dominance constraints'',
requiring certain labels to occur at least as frequently as others, adding a layer of comparative label evaluation.
The work by \citet{aziz2019softconstr} provides a polynomial-time algorithm computing committees that satisfy two axioms,
one ensuring the given diversity constraints are satisfied as much as possible and the other one integrating candidate excellence.
\citet{2022Switzerland} present an innovative election process where voters initially decide
on distributional constraints for candidates' attributes before electing candidates under those constraints,
demonstrating this method in a Swiss primary election study.
In the domain of approval voting, \citet{straszak1993computer} proposed an integer linear programming (ILP) framework
to address diversity constraints across categorized candidate labels, offering computational tools and real-world data applications.
In a working paper, \citet{GassiDissMoyouwou2023} approach the integration of diversity
constraints within the framework of multiwinner elections, presenting a model where the
committee selection prioritizes both high scores and diversity metrics.
Their work extends
well-known committee scoring rules by tailoring them
to meet specified diversity requirements and introduces new axioms for
diverse committee selection under constraints.

Exploring applications beyond elections, \citet{gawron2022movies} utilize multiwinner voting to refine search systems
such as movie recommendations, balancing similarity with diversity
without relying on explicit diversity indices.
The model
from \citet{ijcai2022dire} includes
attributes for candidates and voters,
applying hard distributional constraints and ensuring population-based representation within elected committees,
enriching the voting model with demographic considerations.
Lastly, \citet{izsak2018synergies} present a framework where alternatives are classified, and inter- and intraclass relations
are modeled through synergy functions, aiming to maximize both score-based and relational metrics—a concept parallel to 
our work,
where diversity indices could be interpreted as synergy measures under score and satisfaction
constraints.
While synergy functions capture relational benefits between alternatives, our approach explicitly adapts ecological diversity indices.

Finding diverse solutions is also important in other contexts of collective decision making.
\citet{BenabbouCHSZ18diversityhousing} analyze diversity constraints in context of
(utilitarian) public housing allocation.
\citet{AB21} explore diversity constraints for Brazilian college
admissions through affirmative action policies, examining the strategic complexities of
multidimensional privileges and proposing a fair, strategy-proof mechanism to ensure
equitable student selection.
\citet{AZ25} analyze diversity in the context of school admissions by defining a rank-based
diversity concept, where maximal diversity is achieved by prioritizing student
matches to seats that fulfill the institution's most crucial diversity criteria, thereby
optimizing representation of key groups.
\citet{Biro2010collegequotas} analyze the computational complexity of stable-matching-based
college admissions and incorporate lower quotas for individual colleges and common quotas
for groups of colleges allowing to manage collective diversity targets.

\section{The Model}
Let $\N$ and~$\Nzero$ be the natural numbers excluding and including zero,
respectively,
$\bcks[t]$ the set $\brcs{1,\dots,t}$ for any integer $t$,
and $\Pot{X}$ the power set of any set~$X$.
$(x_i)_{i=a}^b$ denotes the vector~$(x_a,x_{a+1},\dots,x_b)$
and, for a given vector $y$, $\dim(y)$ is the number of entries of $y$, which we refer to as $y$'s dimension.

We consider committee elections with approval preferences
and labelled candidates,
i.e.~elections of the form $\Elec=(A,C,U,k,L,\labof)$ consisting of a set $A$ of agents,
a candidate set $C$ (e.g.~projects),
an approval profile $U\colon A \to \Pot{C}$,
and a desired committee size $0<k\leq \abs{C}$.
$L=\brcs{l_1,\dots, l_{m}}$ is a set of $m$ labels 
(e.g.~education or sport as the target of the projects)
and $\labof\colon C \to L$ assigns a label to each candidate.
In addition, 
for $i\in\bcks[m]$ and a committee $\Com\subseteq C$,
let
\begin{gather*}
  \setFP(\Elec,\Com, i) \ceq \brcs{c\in \Com: \fun{\labof}(c) = l_i},\\
  \fun{n_i}(\Elec,\Com) \ceq \abs{\setFP(\Elec,\Com, i)}\text{, and } \fun{p_i}(\Elec,\Com) \ceq \fun{n_i}(\Elec,\Com)/\abs{\Com}
\end{gather*}
be the set, number, and proportion of candidates in~$\Com$
with label~$l_i$, respectively.
Furthermore, let
\begin{align*}
  \fun{\distri}(\Elec,\Com) \ceq \paren(\abs{\brcs{i\in \bcks[m]: \fun{n_i}(\Elec,\Com) = j}})_{j=0}^{\abs{\Com}},
\end{align*}
where the $j$-th entry indicates the number of labels occurring $j-1$ times in $\Com$.

A rule $\Rule\colon \Elec \to \fun{\mathcal{P}}(\brcs{\Com \subseteq C: |\Com| = k})$
maps an election to at least one $k$-sized subset of~$C$.
We denote by $\Rulebase(\Elec) \ceq \brcs{\Com \subseteq C: |\Com| = k}$ the rule that outputs all committees of size $k$ and
by $\Rule^{s}(\Elec) = \argmax_{\Com \in \Rulebase(\Elec)} \fun{s}(\Elec,\Com)$ the rule that outputs all $S\in\Rulebase(\Elec)$ with maximal score, where $s$ is a scoring function mapping an election and a committee to $\N$.
We only consider scoring functions which take only $A,C,U$ and $k$ into account, i.e.~information that is part of a ``classical'' election.
For an $S\in\Rulebase(\Elec)$, we measure the satisfaction of an agent $a$ as $\fun{\mathrm{sat}}(\Elec, S, a) \ceq |S\cap \fun{U}(a)|$, i.e.~as the number of chosen candidates $a$ approves.
To measure the diversity, we look at diversity indices
which assign a real number to an election $\Elec$ and committee $\Com\in \Rulebase(\Elec)$.
In the following, we may omit arguments if they are clear from the context.

\section{The Diversity Indices}
\label{sec:indices}
\appendixsection{sec:indices}
In this section, we discuss the diversity indices we consider.
\begin{example}\label{exp:1}
 As a running example,
 consider an election $\Elec$ with projects as candidates, $m=3$ labels, $L=\brcs{\labHealth, \labEduc, \labSport}$ (\labHealth{} represents the label ``health'', \labEduc{} ``education'', and \labSport{} ``sport'')\footnote{The emoji graphics are taken from \texttt{twemojis} and licensed under \href{https://creativecommons.org/licenses/by/4.0/}{CC-BY 4.0}.
 Copyright 2019 Twitter, Inc and other contributors.}, $k=10$, and the following committees
 (for each candidate, we indicate its label):
 \begin{center}
  \setlength{\twemojiDefaultHeight}{0.84em}
  \begin{tblr}{colspec={ll},colsep=3pt}
    $\ExCom \in \Rulebase(\Elec)$ & $ \fun{\distri}(\Elec,\ExCom)$ \\\midrule
    $\ExCom_1 = \{\labHealth,\labHealth,\labHealth,\labEduc,\labEduc,\labEduc,\labSport,\labSport,\labSport,\labSport\}$ & $\paren(0,0,0,2,1,0,0,0,0,0,0) $\\
    $\ExCom_2 = \{\labHealth,\labEduc,\labSport,\labSport,\labSport,\labSport,\labSport,\labSport,\labSport,\labSport\}$ & $\paren(0,2,0,0,0,0,0,0,1,0,0) $\\
    $\ExCom_3 =\{\labEduc,\labEduc,\labEduc,\labEduc,\labEduc,\labSport,\labSport,\labSport,\labSport,\labSport\}$ & $\paren(1,0,0,0,0,2,0,0,0,0,0)$
  \end{tblr}
 \end{center}
 Clearly,
 $\ExCom_1$ appears as most diverse in the sense that 
 it contains all three different labels and the labels appear as evenly balanced as possible.
 Whether $\ExCom_2$ is more diverse than $\ExCom_3$ is in the eye of the beholder or, put differently, depends e.g.~on whether someone finds it more important that as many labels as possible are represented or that the labels represented appear as evenly as possible:
 $\ExCom_2$ contains three labels, but one occurs much more frequently than the others, while 
 $\ExCom_3$ contains only two labels, but these occur equally often.
 \eqed
\end{example}

\paragraph{Indices Used in Ecology.}
In the field of ecology, various indices have been defined to 
measure the diversity of a community of species%
---see e.g.~\citep{pielou1975ecological,leinster2021entropy} for an overview of diversity indices.
We directly adapt the following diversity indices often used in ecology
so that they receive an election $\Elec=(A,C,U,k,L,\labof)$
as well as a committee $S\in \Rulebase(\Elec)$ as inputs.
For each of the following indices,
a higher value indicates a higher diversity.

Richness \citep{whittaker1972evolution} is a simple diversity index
that takes
only the number of labels present
into account:
\begin{align*}
  \rich(\Elec,\Com) ={}& \abs{\brcs{i\in \bcks[m]: \fun{n_i}(\Elec,S) > 0}}\\
  ={}& \sum_{\ell=1}^k \distri(\Elec,\Com)_{\ell+1} = m - \distri(\Elec,\Com)_1.
\end{align*}

The Simpson index \citep{Simpson49}\footnote{
In the literature,
the Simpson index is usually stated unnegated.
We add the negation in order to maximize each index.
}
considers the probability that two candidates chosen independently and at random from the committee
have the same label, i.e.
\begin{align*}
  \simps(\Elec,\Com) &= -\sum_{i\in \bcks[m]} p_i(\Elec,\Com)^2\\
  &= -\sum_{\ell=1}^k \distri(\Elec,\Com)_{\ell+1} \cdot \paren(\frac{\ell}{k})^2.
\end{align*}

Another popular index \citep{leinster2012measuring}, derived from information theory,
is Shannon's entropy \citep{Shannon48}.
It represents the uncertainty in predicting the label of a randomly chosen candidate from the committee:
\begin{align*}
  \shann(\Elec,\Com) ={}& -\sum_{i \in \bcks[m]: \fun{p_i}(\Elec,S) > 0} p_i(\Elec,\Com)\cdot \tlog{p_i(\Elec,\Com)}\\
  ={}& -\sum_{\ell=1}^k \distri(\Elec,\Com)_{\ell+1} \cdot \frac{\ell}{k} \cdot \tlog{\frac{\ell}{k}}.
\end{align*}

\begin{remark}
  The indices rank the committees from \cref{exp:1} as follows:
  $\rich(\Elec,\ExCom_1) = 3 = \rich(\Elec,\ExCom_2) > \rich(\Elec,\ExCom_3)=2$, $\shann(\Elec,\ExCom_1) \approx 1.09 > \shann(\Elec,\ExCom_3)= 0.69 > \shann(\Elec,\ExCom_2) \approx 0.64$, and $\simps(\Elec,\ExCom_1) = -0.34 > \simps(\Elec,\ExCom_3)=-0.5 > \simps(\Elec,\ExCom_2) = -0.66$.
  Therefore, $\rich$ classifies $\ExCom_3$ as least diverse, while $\shann$ and $\simps$ both classify $\ExCom_2$ as the least diverse.
  \rqed
\end{remark}

\paragraph{New Index.}
We introduce a new diversity index, the \emph{Lexicographic Counting Index} ($\dns$).
While the idea of lexicographic ordering has been applied in various settings, it has not been used for defining
a diversity index before, to the best of our knowledge.
The idea behind $\dns$---which elevates the natural, but simple, diversity index $\rich$---is the following:
The primary goal is to maximize the number of labels occurring at least once (like $\rich$ does), the secondary goal is to maximize the number of labels occurring at least twice, and so on:
\begin{align*}
  \dns(\Elec,\Com) = {}& \sum_{i=1}^{k} \paren(\funbc{\min}{m, k} + 1)^{k+1-i} \cdot \abs{\sigma_i(\Elec,\Com)},\\
  \text{ with } \sigma_i(\Elec,\Com) = {}& \brcs{\ell\in\bcks[m]: n_\ell(\Elec,\Com) \geq i}.
\end{align*}
Note that the base is $\funbc{\min}{m, k} + 1$, as a committee consists of $k$ candidates
and each candidate introduces at most one new label, i.e.~$\sigma_i(\Elec,\Com) \leq \funbc{\min}{m, k}$ for all $i\in\bcks[k]$.
In Appendix~\ref{app:props}, we evaluate $\dns$ with respect to properties adapted from the literature:
These results provide some arguments why the new index can be called a diversity index.
\begin{remark}
  $\dns$ classifies $\ExCom_1$ as the most
  and $\ExCom_3$ as the least diverse committee: $\dns(\Elec,\ExCom_1)=4^{10}\cdot 3+4^{9}\cdot 3+4^{8}\cdot 3+4^7 = 4 145 152 > \dns(\Elec,\ExCom_2)= 4^{10}\cdot 3+ 4^9 + 4^8 + 4^7 + 4^6 + 4^5+ 4^4+ 4^3 = 3 495 232 > \dns(\Elec,\ExCom_3)= 4^{10}\cdot 2+ 4^9 \cdot 2 + 4^8 \cdot 2+ 4^7 \cdot 2+ 4^6 \cdot 2 = 2793472$.
  \rqed
\end{remark}

\toappendix{
\subsection{\boldmath Testing $\dns$ against Adapted Properties}\label{app:props}
In this section, we adapt properties defined in the literature as desirable for diversity indices 
and test our new diversity index $\dns$ against them.
Not all of these properties are necessary in our context, but we look at them to provide some reasons why we call the new index \emph{diversity} index.
We consider only the new diversity index here, but we want to stress that not each of the indices that are used in ecology and that we consider satisfies all following adapted properties.

\substruct{Set Monotonicity}{
This property defined by \citet{IZSAK2000151} states that the diversity should increase if a new species is added to a set of species.
We adapt this property by replacing the species with the labels and demanding the following:

A diversity index $D$ satisfies \textit{Set Monotonicity} if, for all elections $\Elec_1=(A,C,U,k,L,\labof)$ and $\Elec_2=(A,C,U,k+1,L,\labof)$ (i.e.~only the committee size differs) with $m>k$ and for all committees $S_1\in \fun{\Rulebase}(\Elec_1)$ and $S_2\in \fun{\Rulebase}(\Elec_2)$ with $S_2 = S_1\cup\brcs{c'}$ and $\fun{n_i}(\Elec_1,S_1)=0$ with $l_i=\fun{\labof}(c')$,
it holds that $\fun{D}(S_1) < \fun{D}(S_2)$.
}

$\dns{}$ satisfies this property, as
\begin{align*}
  {}&\fun{\dns}(\Elec_2,S_2) - \fun{\dns}(\Elec_1,S_1)\\
  ={}& \paren(\sum_{i=1}^{k} \paren(\paren(k + 2)^{k+2-i} - \paren(k+1)^{k+1-i}) \cdot \abs{\sigma_i})\\
  {}&+ \paren(k + 2)^{k+1}
  > 0.
\end{align*}
This property is not important in our context, as we search for a committee of a fixed size and, therefore, only need to compare committees with this size based on the same candidate pool.

\substruct{More Species do not Harm}{
This property mentioned by \citet{IZSAK2000151}, \citet{baczkowski1997properties}, and by \citet{izsak1996sensitivity} states that the diversity should not decrease if a new species is added to a set of species with equal frequencies in such a way that the frequencies of all species are equal.
We adapt this property by replacing the species with the labels and demanding the following:

A diversity index $D$ satisfies this property if for all elections $\Elec_1=(A,C,U,k,L,\labof)$ with $m$ labels for which $\exists i\in\N:k=m\cdot i$ and for all $\Elec_2=(A,C\cup C',U,k+i,L',\labof')$ with $m+1$ labels, $\abs{C'}=i$ and $L',\labof'$ leading to all candidates from $C$ having the same labels as in $\Elec_1$ and all candidates from $C'$ having the same label $l_{m+1}$ not present in $L$,
it holds that $\fun{D}(S_1) \leq \fun{D}(S_2)$ for all $S_1\in \fun{\Rulebase}(\Elec_1), S_2\in \fun{\Rulebase}(\Elec_2)$ with $\forall l\in\bcks[m]: \fun{n_l}(\Elec_1, S_1) = \fun{n_l}(\Elec_2, S_2) = i$ and $\fun{n_{m+1}}(\Elec_2, S_2)= i$.
}

$\dns{}$ satisfies this property, as
\begin{align*}
  \fun{\dns}(\Elec_1,S_1)
  ={}& \sum_{j=1}^i m \cdot (m+1)^{k+1-j} \\
  <{}& \sum_{j=1}^i \paren(m+1) \cdot (m+2)^{k+i+1-j}\\
  ={}& \fun{\dns}(\Elec_2,S_2).
\end{align*}
This property is not important in our context, either, for the same reasons as the previous property.

\substruct{Equal Frequencies are Optimal}\label{substruct:earo}
A property mentioned by \citet{baczkowski1997properties} and \citet{izsak1996sensitivity} demands of a diversity index (for a given number of species) that it is maximal if the frequencies of the species are equal.
We adapt this property by demanding the following:

A diversity index $D$ satisfies this property if, for all elections $\Elec=(A,C,U,k,L,\labof)$ with $m$ labels for which $\exists i\in\mathbb{N}:k=m\cdot i$ or for which $m \geq k$ (let $i=1$ in the latter case),
it holds that $\fun{D}(S_1)$ is optimal if $S_1\in \fun{\Rulebase}(\Elec)$ and $\exists L'\subseteq L$: $|L'|=\funbc{\min}{m, k}$ and $\forall l_j\in L': \fun{n_j}(\Elec, S_1)= i$.

$\dns{}$ satisfies this property:
Let $S_1\in \fun{\Rulebase}(\Elec)$ be a committee satisfying this condition
and $S_2\in \fun{\Rulebase}(\Elec)$ a committee violating this property, i.e.~$\exists l\in\bcks[m]: \fun{n_l}(\Elec, S_2)>i$.
$S_2$ can be transformed into $S_1$ iteratively by replacing a candidate with a label that occurs more than $i$ times with a candidate with a label that occurs less than $i$ times, which leads to a strictly higher diversity.
More formally:
\begin{enumerate}
  \item Let $l_1 \in\bcks[m]$ be a label with $i_1\ceq \fun{n_{l_1}}(\Elec, S_2) < i$ and $l_2 \in\bcks[m]$ be a label with $i_2\ceq \fun{n_{l_2}}(\Elec, S_2) > i$.
  \item Let $S_3\in \fun{\Rulebase}(\Elec)$ be a committee with $\fun{n_{l_1}}(\Elec, S_3) = i_1+1$, $\fun{n_{l_2}}(\Elec, S_3) =i_2-1$, and $\fun{n_{f}}(\Elec, S_3) = \fun{n_{f}}(\Elec, S_2)$ for $f\in \bcks[m] \setminus \brcs{l_1, l_2}$.
  \item Set $S_2 = S_3$.
  \item If there is a $L'\subseteq L$ with $|L'|=m$ so that $\forall l_j\in L': \fun{n_j}(\Elec, S_2)= i$, stop. Otherwise, continue with the first step.
\end{enumerate}
In the second step, the diversity of $S_3$ is strictly higher than that of $S_2$ according to $\dns$, with $\eta = \funbc{\min}{m, k}+1$:
\begin{align*}
  \fun{\dns}(S_3) - \fun{\dns}(\Elec,S_2)
  ={}& \eta^{k - i_1} - \eta^{k+1-i_2}\\
  ={}& \eta^{k - i_1} \paren(1-\eta^{1-\paren(i_2-i_1)})\\
  >{}&  0 \text{ because } i_2-i_1 > 1.
\end{align*}
Clearly, this property is desirable in our context:
If it is possible that each label occurs equally often, this should lead to the highest diversity.

\substruct{Symmetry}
This obviously desirable property mentioned in \citep{leinster2012measuring,izsak1996sensitivity} demands from a diversity index that its value remains the same regardless of the order of the species.
When adapting this property by replacing the species with the labels, $\dns$ fulfills this property trivially.

\substruct{Effective Number}
A property mentioned in \citep{leinster2012measuring,routledgeDivIdxAdmissable} demands from a diversity index that its value for a candidate set with $s$ species equals $s$ if each species has a frequency of $s^{-1}$.
We adapt his property by demanding the following of a diversity index $D$:
For all elections $\Elec=(A,C,U,k,L,\labof)$ with $m$ labels for which $\exists i\in\mathbb{N}:k=m\cdot i$ and $S\in \fun{\Rulebase}(\Elec)$ with $\forall l\in\bcks[m]: \fun{n_l}(\Elec, S)= i$, it holds that $\fun{D}(S) = m$.

This property is not fulfilled by $\dns$, 
as $\fun{\dns}(\Elec,S) = \sum_{j=1}^i m(m+1)^{k +1-j}$.

Because $\dns$ does not fulfill this property, it also does not fulfill a different property defined in \citep{routledgeDivIdxAdmissable}, which demands that the diversity is smaller than the number of species if the species do not have the same frequency.
Analogously, the diversity index does not satisfy a property defined in \citep{leinster2012measuring} which requires an index's value range to be $\brcs{1, \dots, s}$, where $s$ is the number of species (which is $m$ when replacing species with labels). 

Are these properties important in our context?
It seems that this depends on the situation.
It is useful, for example, if it is desirable to derive information about the present species directly from the value of the diversity index (i.e.~the effective number of species; for more information see \citep{jost2006entropy}, for example).
If, on the other hand, the $\distri$ vector is given, the number of represented species and which of two given sets is more diverse according to $\dns$ is \emph{easy} to see (see Section~\ref{sec:props}).
In such a situation, one could do without this property.
In addition, note that the other diversity indices we consider do not satisfy this property, either, $\rich$ being the only exception.

\substruct{Absent Species}
This property mentioned in \citep{leinster2012measuring} requires an index's value to remain the same if a new species is added that occurs zero times.
We adapt this property by replacing the species with the labels and demanding the following from a diversity index $D$:
For all elections $\Elec_1=(A,C,U,k,L,\labof)$ with $m$ labels and for all $\Elec_2=(A,C\cup\brcs{c'},U,k,L',\labof')$ with $c'\notin C$, $m+1$ labels and $L',\labof'$ leading to all candidates from $C$ having the same labels as in $\Elec_1$ and $c'$ having label $l_{m+1}$ not present in $L$,
it holds that $\fun{D}(S_1) = \fun{D}(S_2)$ with $S_1\in \fun{\Rulebase}(\Elec_1)$,$S_2\in \fun{\Rulebase}(\Elec_2)$, and $\forall l \in \bcks[m]: \fun{n_l}(\Elec_1, S_1)= \fun{n_l}(\Elec_2, S_2)$ and $\fun{n_{m+1}}(\Elec_2, S_2) = 0$.

$\dns$ does not fulfill this property:
Consider, for example, as $\Elec_1$ an election with $m=2$ and each label occurring twice (i.e.~$k=4$), then it holds that $\fun{\dns}(\Elec_2,S_2) - \fun{\dns}(\Elec_1,S_1)=424 > 0$.
This property is not important in our context for the same reasons as for \textit{Set Monotonicity}.

}

\section{Which diversity indices to use?}
\label{sec:props}
\appendixsection{sec:props}
Next, we want to distinguish formally between the aforementioned indices by defining properties and testing the indices against them.
Note that not all properties that we define should necessarily be satisfied by a diversity index:
The properties rather draw attention to differences between the indices, which should be taken into account when picking a diversity index to be used.
This is of interest when electing committees consisting of at least six candidates, because $\shann$, $\simps$, and $\dns$ behave the same for small committees when deciding which committee is more diverse:
\begin{observation}[\appref{obs:small_cands}]
  \label{obs:small_cands}
  For all elections $\Elec$ with $k\leq 5$, it holds that $\forall r,r'\in \brcs{\shann, \simps, \dns}, \diamond \in\brcs{<,>,=}, S_1,S_2\in\Rulebase(\Elec): \fun{r}(\Elec, S_1) \diamond\fun{r}(\Elec, S_2)  \iff \fun{r'}(\Elec, S_1) \diamond \fun{r'}(\Elec, S_2)$.
  For all elections $\Elec$ with $k\leq 7$, it holds that $\forall S_1,S_2\in\Rulebase(\Elec), \diamond \in\brcs{<,>,=} : \fun{\shann}(\Elec, S_1) \diamond\fun{\shann}(\Elec, S_2)  \iff \fun{\dns}(\Elec, S_1) \diamond \fun{\dns}(\Elec, S_2)$.
\end{observation}
\appendixproof{obs:small_cands}
{
We observed this programmatically: For each committee size $k\in\brcs{1,\dots,7}$, we iterated over each possible number $m$ of labels that can occur in the committee, i.e.~$m\in \bcks[k]$, as each candidate can introduce at most one label.
For each such combination of $k$ and $m$, we looked at all possible $\distri$ vectors and computed and compared the diversity indices of interest, exploiting the fact that the naming and ordering of the candidates and agents as well as their votes do not matter for the diversity indices at hand.
}
However, the indices can behave differently for larger $k$.
One reason for this is that only $\rich$ and $\dns$
consider the number of labels present to be more important
than the evenness of the distribution of the labels present,
while $\simps$ and $\shann$ do not---this can be seen e.g.~in \cref{exp:1}.
Thus, the first question that one has to answer is whether having as many labels in the committees as possible has the highest priority---this could be the case when electing a team working on an interdisciplinary project, where having an expert from a discipline that is not covered otherwise is very valuable.
We express this through the following property:
\begin{property}[Present Label Maximization]\label{propty:pfgpmax}%
  A diversity index $D$ satisfies \emph{Present Label Maximization} if, for all elections $\Elec$ and $S_1,S_2\in \fun{\Rulebase}(\Elec)$ for which $\fun{\distri}(\Elec, S_1)_1 < \fun{\distri}(\Elec,S_2)_1$
  (i.e.~$S_1$ contains more different labels than $S_2$),
  it holds that $\fun{D}(S_1) > \fun{D}(S_2)$.
\end{property}
\begin{observation}[\appref{obs:prop:pfgpmax}]
  \label{obs:prop:pfgpmax}
  $\rich$ and $\dns$ satisfy \pref{propty:pfgpmax}, $\simps$ and $\shann$ do not.
\end{observation}
\appendixproof{obs:prop:pfgpmax}{
  \textbf{Diversity indices $\shann$ and $\simps$:} Consider an election $\Elec$ with $k=10$ and $m=3$ labels, and two committees $S_1, S_2 \in \fun{\Rulebase}(\Elec)$ with
  \begin{align*}
    \fun{p_1}(\Elec, S_1) ={}& \fun{p_2}(\Elec, S_1)=0.1, \fun{p_3}(\Elec, S_1)=0.8\\
    \fun{p_1}(\Elec, S_2) ={}& 0, \fun{p_2}(\Elec, S_2)=\fun{p_3}(\Elec, S_2)=0.5.
  \end{align*}
  It holds that $\fun{\distri}(\Elec, S_1)_1 = 0 < 1 = \fun{\distri}(S_2)_1$,
  but $\fun{\shann}(S_1)\approx 0.64 < \fun{\shann}(S_2)\approx 0.69$, and $\fun{\simps}(S_1) = -0.66 < \fun{\simps}(S_2) = -0.5$.

  \textbf{Diversity index $\rich$:} This follows directly from the definition of $\rich$ as $m - \fun{\distri}(\Elec, S)_{1}$.

  \textbf{Diversity index $\dns$:}
  This follows directly from 
  \begin{gather*}
    e_{\dns}(\rdistriv^{(1)}, \rdistriv^{(2)}, \rdistrividx)\\
    = \begin{cases*} \mathit{less}, \text{ if }\dim(\rdistrividx)\geq 1 \text{ and }\rdistriv^{(1)}_1 > \rdistriv^{(2)}_1\\  \mathit{equal}, \text{ if } \dim(\rdistrividx) = 0 \\
  \mathit{more}, \text{ if } \dim(\rdistrividx)\geq 1 \text{ and }\rdistriv^{(1)}_1 < \rdistriv^{(2)}_1 \end{cases*}
  \end{gather*}
  being a $\dns$-explainable function (see \cref{proof:theorem:explain:bounds}) and $\fun{\rdistriv}(\Elec, S_1, S_2)_1  = \fun{\distri}(\Elec, S_1)_1  <  \fun{\rdistriv}(\Elec, S_2, S_1)_1 = \fun{\distri}(\Elec, S_2)_1$.
}

Another property which is quite natural and which should be satisfied by any diversity index is that 
increasing the occurrences of a label by decreasing those of a more frequent label by at least two 
leads to a higher (\pref{propty:balance}) or at least the same (\pref[Weak Occurrence Balancing]{propty:balance}) diversity:
\begin{property}[Occurrence Balancing]\label{propty:balance}%
  A diversity index $D$ satisfies (\emph{Weak}) \emph{Occurrence Balancing} if for all elections $\Elec$ and $S'\in \fun{\Rulebase}(\Elec)$ for which there are $i,j\in \bcks[m]$ with $n_i(\Elec,S') + 2 \leq n_j(\Elec,S')$,
  it holds that $\forall c_i \in \setFP(\Elec,C, i)\setminus S', c_j \in \setFP(\Elec,S', j): \fun{D}(S') < \fun{D}(S'')$ ($\fun{D}(S') \leq \fun{D}(S'')$) with $S''=S'\setminus \brcs{c_j} \cup \brcs{c_i}$.
\end{property}
The following result not only separates $\rich$ from the other three indices,
but is also the basis for showing in \cref{sec:compl} that
finding a committee with the highest diversity
is in \classP{}.
\begin{observation}[\appref{obs:prop:propBalance}]
  \label{obs:prop:propBalance}%
  All diversity indices considered satisfy \pref[Weak Occurrence Balancing]{propty:balance}.
  $\simps$, $\shann$, and $\dns$ satisfy \pref{propty:balance}.
\end{observation}
\appendixproof{obs:prop:propBalance}{
In the following, let $n'_l = \fun{n_l}(\Elec, S')$, $n_l'' = \fun{n_l}(\Elec, S'')$ for $l\in \bcks[m]$, $\distri'=\fun{\distri}(\Elec, S')$, $\distri''=\fun{\distri}(\Elec, S'')$.
For the meaning of $i,j,S'$ and $S''$ see the definition of \pref{propty:balance}.

\textbf{Diversity index $\rich$:} If $n_i' = 0$ and thus $\distri_1' > \distri_1''$, it holds that $\fun{\rich}(\Elec, S') = m -\distri_1' < m -\distri_1'' = \fun{\rich}(\Elec, S'')$.
Otherwise, i.e.~if $\distri'_1 = \distri''_1$, $\fun{\rich}(\Elec, S') = \fun{\rich}(\Elec, S'') =m - \distri'_1$.
Hence, $\fun{\rich}(S') \leq \fun{\rich}(S'')$.

\textbf{Diversity index $\shann$:} Let $M_R'=\brcs{i \in \bcks[m]: \fun{p_i}(\Elec,S') > 0}$ and $M_R''=\brcs{i \in \bcks[m]: \fun{p_i}(\Elec,S'') > 0}$. It holds that
\begin{align*}
  {}& \fun{\shann}(\Elec, S'')-\fun{\shann}(\Elec, S') > 0\\
  \iff {}& \sum_{p \in M_R'} \frac{n'_p}{k}\cdot \tlog{\frac{n'_p}{k}}-\sum_{p \in M_R''} \frac{n''_p}{k}\cdot \tlog{\frac{n''_p}{k}} > 0\\
  \iff{}& \sum_{p \in M_R} n'_p\cdot \tlog{n'_p}-\sum_{p \in M_R} n''_p\cdot \tlog{n''_p} > 0.
\end{align*}
If $n_i'=0$, this is equivalent to $n'_j\cdot \tlog{n'_j} - \paren(n'_j-1)\cdot \tlog{n'_j-1} > 0$, which is true because $\log$ is strictly monotonically increasing.
Otherwise, i.e.~if $n_i'>0$, it is equivalent to 
$n'_j\cdot \tlog{n'_j} - \paren(n'_j-1)\cdot \tlog{n'_j-1}
+n'_i\cdot \tlog{n'_i} - (n'_i+1)\cdot \tlog{n'_i +1} > 0$, which we prove in the following:

Let $\fun{f}(x) = x\fun{\log}(x)$ with $\fun{f'}(x) = 1+\fun{\log}(x)$ which is strictly monotonically increasing.
According to the mean value theorem, $\exists c_1 \in \paren(n'_j-1, n'_j), c_2 \in \paren(n'_i, n'_i+1)$ such that $\fun{f}(n_j')-\fun{f}(n'_j-1) = \fun{f'}(c_1) = 1+\fun{\log}(c_1)$ 
and $\fun{f}(n'_i+1) - \fun{f}(n'_i) = \fun{f'}(c_2) = 1 + \fun{\log}(c_2)$.
Thus, the former inequality is equivalent to $1+\fun{\log}(c_1)-1-\fun{\log}(c_2) = \fun{\log}(c_1)-\fun{\log}(c_2) > 0$, which is true due to $c_1 > c_2$.

\textbf{Diversity index $\simps$:} It holds that $\fun{\simps}(\Elec, S'')-\fun{\simps}(\Elec, S') > 0 \iff \sum_{p\in \bcks[m]} \paren(n''_p)^2 - \paren(n'_p)^2 < 0 \iff \paren(n'_j-1)^2- \paren(n'_j)^2 + \paren(n'_i+1)^2- \paren(n'_i)^2 < 0 \iff 2 +2 \paren(n_i' - n_j') < 0$.
The latter is true, as $n_i' + 1 < n_j'$ and thus $n_i' - n_j' \leq -2$.

\textbf{Diversity index $\dns$:} Let $\delta \ceq n'_j-n'_i$ and $\eta = \funbc{\min}{m, k}+1$. It holds that
  \begin{align*}
    {}&\fun{\dns}(S'') - \fun{\dns}(S') > 0\\
    \iff {}& \eta^{k+1-n'_i-1}-\eta^{k+1-n'_j} = \eta^{k-n'_i} - \eta^{k+1-n'_i-\delta} > 0\\
    \iff {}& 1 - \eta^{1-\delta} > 0
  \end{align*}
  which is true, as $1-\delta < 0$.
}

The properties introduced so far allow us to differentiate
between any pair of diversity indices considered apart from $\simps$ and $\shann$,
which both take the evenness of the distribution into account.
However, the indices disagree on whether balancing two labels is preferable for rare or for dominant labels.
We first define which pairs of labels whose occurrences differ by $d$ are \emph{balancable}: $\fun{\mathrm{balancable}}(\Elec, \Com, d)$ returns, for a given election $\Elec$ and a committee $\Com\in\Rulebase(\Elec)$, all pairs $(i,j)$ so that $n_i(\Elec,\Com) + d = n_j(\Elec,S)$ and $n_i(\Elec,S) + \lfloor \frac{d}{2}\rfloor \leq \abs{\setFP(\Elec,C, i)}$ (i.e.~the first label of the pair occurs $d$ fewer times than the second label, but its frequency could be increased by $\lfloor \frac{d}{2}\rfloor$).
The function $\fun{\mathrm{balance}}(\Elec, \Com, d, \paren(i,j))$ actually balances such labels by taking,
in addition to $\Elec$ and $\Com$,
a label pair $\paren(i,j) \in \fun{\mathrm{balancable}}(\Elec, \Com, d)$
as an argument
and returning a committee $\Com'\in\Rulebase(\Elec)$ so that
$n_i(\Elec,\Com') = n_i(\Elec,\Com) + \lfloor \frac{d}{2}\rfloor, n_j(\Elec,\Com') = n_j(\Elec,\Com)-\lfloor \frac{d}{2}\rfloor$
and $\forall e\in\bcks[m]\setminus\brcs{i,j}: n_e(\Elec,\Com') = n_e(\Elec,\Com)$,
i.e.~only the number of $l_i$ and $l_j$ are balanced.
Based on this, we can define the following property:
\begin{property}[Prioritization of Rare Label Balancing]\label{propty:Prio}
  A diversity index $D$ satisfies \emph{Prioritization of Rare Label Balancing} if, for all elections $\Elec$, $S\in \fun{\Rulebase}(\Elec)$ and $d\geq 2$ for which there are $(i,j),(k,l)\in \fun{\mathrm{balancable}}(\Elec, \Com, d)$ with $n_i(\Elec,\Com) < n_k(\Elec,\Com)$, it holds for $\Com_{\paren(i,j)} \ceq \fun{\mathrm{balance}}(\Elec, \Com, d, \paren(i,j))$ and $\Com_{\paren(k,l)} \ceq \fun{\mathrm{balance}}(\Elec, \Com, d, \paren(k,l))$ regarding the diversity that $\fun{D}(\Elec, \Com_{\paren(i,j)}) > \fun{D}(\Elec,\Com_{\paren(k,l)})$.
\end{property}
\begin{observation}[\appref{obs:prop:Prio}]
  \label{obs:prop:Prio}
  $\shann$ and $\dns$ satisfy \pref{propty:Prio}, $\simps$ and $\rich$ do not.
\end{observation}
\appendixproof{obs:prop:Prio}{
  In the following, let $\alpha \ceq n_i(\Elec,\Com)$ and $d_2 \ceq \lfloor \frac{d}{2}\rfloor$.

  \textbf{Diversity index $\rich$:} If $n_i(\Elec,\Com)\geq 1$, it holds that $\distri(\Elec,\Com)_1 = \distri(\Elec,\Com_{\paren(i,j)})_1 = \distri(\Elec,\Com_{\paren(k,l)})_1$ and therefore $\fun{\rich}(\Com_{\paren(i,j)}) =  m - \distri(\Elec,\Com)_1 = \fun{\rich}(\Com_{\paren(k,l)})$.

  \textbf{Diversity index $\simps$:} It holds that
  \begin{align*}
    &\fun{\simps}(\Com_{\paren(i,j)}) - \fun{\simps}(\Com)\\
    ={}& \paren(-\paren(\alpha+d_2)^2 -\paren(\alpha+d-d_2)^2 + \alpha^2 + \paren(\alpha+d)^2)/k^2 \\
    ={}& \paren(-d_2^2 - \paren(d-d_2)^2 +d^2)/k^2
  \end{align*}
  and analogously
  \begin{gather*}
    \fun{\simps}(\Com_{\paren(k,l)}) - \fun{\simps}(\Com) = \paren(-d_2^2 - \paren(d-d_2)^2 +d^2)/k^2.
  \end{gather*}
  Hence,
  \begin{align*}
    &\fun{\simps}(\Com_{\paren(i,j)}) - \fun{\simps}(\Com_{\paren(k,l)})\\
    =&{} \fun{\simps}(\Com_{\paren(i,j)}) - \fun{\simps}(\Com) - \paren(\fun{\simps}(\Com_{\paren(k,l)}) - \fun{\simps}(\Com)) = 0.
  \end{align*}

  \textbf{Diversity index $\dns$:} This follows from
  \begin{align*}
    &\fun{\rdistriv}(\Elec, \Com_{\paren(i,j)}, \Com_{\paren(k,l)})_1  = \fun{\distri}(\Elec, \Com_{\paren(i,j)})_{\alpha +1} \\
    <&{}  \fun{\rdistriv}(\Elec, \Com_{\paren(k,l)}, \Com_{\paren(i,j)})_1 = \fun{\distri}(\Elec, \Com_{\paren(k,l)})_{\alpha +1}
  \end{align*}
  and
  \begin{gather*}
    e_{\dns}(\rdistriv^{(1)}, \rdistriv^{(2)}, \rdistrividx)\\
    = \begin{cases*} \mathit{less}, \text{ if }\dim(\rdistrividx)\geq 1 \text{ and }\rdistriv^{(1)}_1 > \rdistriv^{(2)}_1\\  \mathit{equal}, \text{ if } \dim(\rdistrividx) = 0 \\
  \mathit{more}, \text{ if } \dim(\rdistrividx)\geq 1 \text{ and }\rdistriv^{(1)}_1 < \rdistriv^{(2)}_1 \end{cases*}
  \end{gather*}
  being a $\dns$-explainable function (see \cref{proof:theorem:explain:bounds}).

  \textbf{Diversity index $\shann$:}
  Consider the function
  \begin{align*}
    f(x) = \frac{1}{k}(&-\paren(x+d_2)\tlog{x+d_2}\\
    &-\paren(x+d-d_2)\tlog{x+d-d_2}\\
    &+x\tlog{x}+\paren(x+d)\tlog{x+d}).
  \end{align*}
  Hence, $f(\alpha) = \fun{\shann}(\Com_{\paren(i,j)}) - \fun{\shann}(\Com)$ and $f(n_k(\Elec,\Com)) = \fun{\shann}(\Com_{\paren(k,l)}) - \fun{\shann}(\Com)$.
  The derivative is
  \begin{align*}
    f'(x) ={}& (\tlog{x} + \tlog{d + x} - \tlog{d - d_2 + x}\\
    {}& - \tlog{d_2 + x})/k\\
    ={}& \tlog{\frac{x\paren(d+x)}{\paren(d-d_2+x)\paren(d_2+x)}}/k.
  \end{align*}
  As 
  \begin{align*}
    & \frac{x\paren(d+x)}{\paren(d-d_2+x)\paren(d_2+x)}\\
    ={}& \frac{xd+x^2}{dd_2 +dx -d_2^2 -d_2x + d_2x + x^2}\\
    ={}& \frac{xd+x^2}{xd + x^2 + dd_2 - d_2^2} < 1
  \end{align*}
  due to $d > d_2$, it holds that $f'(x)<0$ for $x\geq 0$.
  Hence, $f(x)$ is strictly monotonically decreasing for $x\geq 0$ and $f(\alpha) = \fun{\shann}(\Com_{\paren(i,j)}) - \fun{\shann}(\Com) > f(n_k(\Elec,\Com)) = \fun{\shann}(\Com_{\paren(k,l)}) - \fun{\shann}(\Com)\iff \fun{\shann}(\Com_{\paren(i,j)}) > {\shann}(\Com_{\paren(k,l)})$.
}
In some sense, this makes $\shann$ more similar to $\dns$ than $\simps$, as $\shann$ and $\dns$ both prefer to increase the frequency of the label that is rarest among the four labels.
In contrast, $\simps$ rates both options as equally good and thus does not differentiate whether you decrease the frequency of the most dominant label or increase the frequency of the rarest label among the four labels at hand.

Next, we want to look at another property that distinguishes between our indices and takes into account that the index is to be used in an election:
One important goal is to ensure that voters (or other stakeholders) are able to understand why a committee has been chosen over a different one,
which is likely to promote acceptance of the result or at least a more informed debate about it (e.g.~when maximizing diversity is incorporated into elections).
Hence, we want to focus on the \textit{explainability} of the diversity indices next, or, in other words, the effort required to decide which of two given committees is more diverse.

For this, let $\Elec$ be an election with $S_1,S_2\in\Rulebase(\Elec)$. 
To decide which of $S_1$ and $S_2$ is more diverse (or whether they are equally diverse),
the only information required by the diversity indices considered is how many labels occur how often,
which is provided by the $\distri$ vectors.
When comparing $\fun{\distri}(\Elec, S_1)$ and $\fun{\distri}(\Elec, S_2)$, e.g.~for $\ExCom_2$ and $\ExCom_3$ from \cref{exp:1}, it seems natural to look only at the entries at which the $\distri$ vectors differ.
Indeed, it is not difficult to see from the definitions of the considered diversity indices that
any $i\in\bcks[k+1]$ with $\fun{\distri}(\Elec, S_1)_i = \fun{\distri}(\Elec, S_2)_i$ is irrelevant.

Thus, we denote all relevant indices as $\fun{I_R}(\Elec, S_1, S_2)=\brcs{i\in\bcks[k+1]: \fun{\distri}(\Elec, S_1)_i \neq \fun{\distri}(\Elec, S_2)_i}$.
Based on this, we define the \emph{reduced distribution vector} $\fun{\rdistriv}(\Elec, S_1, S_2)$ as the $\distri$ vector of $S_1$ at the indices in $\fun{I_R}(\Elec, S_1, S_2)$,
with $\rdistrividx$ as the vector of the elements of $\fun{I_R}(\Elec, S_1, S_2)$ in ascending order:
\begin{align*}
  \fun{\rdistriv}(\Elec, S_1, S_2) &\ceq \paren(\fun{\distri_{\rdistrividx_i}}(\Elec, S_1))_{i=1}^{\abs{\fun{I_R}(\Elec, S_1, S_2)}}.
\end{align*}
Consequently, $\fun{\rdistriv}(\Elec, S_1, S_2)$ and $\fun{\rdistriv}(\Elec, S_2, S_1)$ correspond to $\distri(\Elec, S_1)$ and $\distri(\Elec, S_2)$, respectively, where the entries at indices at which the $\distri$ vectors are equal are removed.
Note that this implies $\dim(\fun{\rdistriv}(\Elec, S_1, S_2))=\dim(\fun{\rdistriv}(\Elec, S_2, S_1))=\dim(\rdistrividx)$.
\begin{remark}\label{remark:rdistr}
  For the committees $\ExCom_2$ and $\ExCom_3$ in \cref{exp:1}, we have
  \begin{gather*}
    \fun{I_R}(\Elec, \ExCom_2, \ExCom_3) = \fun{I_R}(\Elec, \ExCom_3, \ExCom_2) = \brcs{1,2,6,9}, \rdistrividx=(1,2,6,9)\\
    \fun{\rdistriv}(\Elec, \ExCom_2, \ExCom_3) = \paren(0,2,0,1), \fun{\rdistriv}(\Elec, \ExCom_3, \ExCom_2) = \paren(1,0,2,0)\tag*{\rqed}
  \end{gather*}
\end{remark}

Based on this, we call a diversity index \emph{$n$-explainable} if, roughly speaking, at most $n$ indices from $\fun{\rdistriv}(\Elec, S_1, S_2)$, $\fun{\rdistriv}(\Elec, S_2, S_1)$, or $\rdistrividx$ need to be used to decide which of $S_1, S_2$ is more diverse or whether they are equally diverse.
Clearly, if a diversity index uses only a few indices, the decision is easier to follow than when many are used, especially when the dimension of the vectors becomes larger.
Thus, $n$ represents a degree of explainability, with a small value indicating a high explainability.
We want to formalize this with the help of a function that makes this decision:
Let $T$ be all possible triples $(a,b,c)$ for which there are an election $\Elec$ and $S_1,S_2\in\Rulebase(\Elec)$ so that $a=\fun{\rdistriv}(\Elec, S_1, S_2), b=\fun{\rdistriv}(\Elec, S_2, S_1)$, and $c$ is the $\rdistrividx$ vector, i.e.~all the inputs such a function needs to be able to process.
For a given diversity index $D$, a function $e: T \to \brcs{\mathit{less},\mathit{more},\mathit{equal}}$ is called \emph{$D$-explainable using at most $n$ indices} if
\begin{enumerate}[nosep,labelindent=0pt]%
  \item $e$ decides whether the first argument describes a more or less diverse committee than the second, i.e.
    \begin{gather*}
    e(\fun{\rdistriv}(\Elec, S_1, S_2),\fun{\rdistriv}(\Elec, S_2, S_1),\rdistrividx) =\\
    \begin{cases}
      \mathit{less} &\text{if } D(\Elec, S_1) < D(\Elec, S_2)\\ 
      \mathit{equal}&\text{if } D(\Elec, S_1) = D(\Elec, S_2) \\
      \mathit{more} &\text{if } D(\Elec, S_1) > D(\Elec, S_2) \end{cases}
    \end{gather*}
  \item $e((),(),()) = \mathit{equal}$, which happens if two $\distri$ vectors are the same---a scenario in which each diversity index considered classifies the committees as equally diverse.
  \item For each $r\in\N$, there are two index sets $I_o,I_d\subseteq \bcks[r]$ so that
  \begin{itemize}
    \item for all $(\rdistriv^{(1)},\rdistriv^{(2)},\rdistrividx)\in T$ with $\dim(\rdistrividx) = r$, it holds that $\{i\in\bcks[r]: e\text{ uses }\rdistriv^{(1)}_i \text{ or }\rdistriv^{(2)}_i \text{ or both}\} \subseteq I_d$ and $\{i\in\bcks[r]: e\text{ uses }\rdistrividx_i\} \subseteq I_o$, i.e.~at most the indices from $I_o$ are accessed in $\rdistrividx$ and at most those from $I_d$ are accessed in the $\rdistriv$ vectors for any vectors of dimension $r$,
    \item $\abs{I_o}+\abs{I_d}\leq n$, i.e.~at most $n$ indices are used.
  \end{itemize}
  Thus, we count the use of both $\rdistriv^{(1)}_i$ and $\rdistriv^{(2)}_i$ only once, as it is natural to assume that we need to look at the same indices in both vectors to compare them. In addition, we allow the index sets to differ with varying $r$. This is necessary e.g.~when the diversity index needs to access nearly all indices, so that the index sets become larger with growing $r$.
\end{enumerate}
Based on this, we define the following properties:
\begin{property}[$n$-Explainability]\label{propty:obv}
  A diversity index $D$ is \emph{$n$-explainable} if 
  there is a $D$-explainable function~$e$ that uses at most $n$ indices for any argument $(a,b,c)\in T$ with $\dim(a)\geq 1$\footnote{For $\dim(a)=0$, each $D$-explainable function returns $\mathit{equal}$ and hence does not need to use any index.}.
\end{property}
As, for each of $e$'s arguments $t=(a,b,c)\in T$, the number of indices used is bounded upwards by $2\dim(a)=2\dim(b)=2\dim(c)$, we denote the dimension of $a$ for $t=(a,b,c)\in T$ as $r$.
By using all the available information of each argument $t\in T$, each of the considered diversity indices is clearly \pref[$2r$-explainable]{propty:obv}.
Next, we give the smallest possible $n$ so that the index is \pref[$n$-explainable]{propty:obv}, for each of the diversity indices considered.
\begin{theorem}[\appref{theorem:explain:bounds}]
  \label{theorem:explain:bounds}
  \begin{enumerate}[nosep]
    \item $\dns$ is \pref[$1$-explainable]{propty:obv},
    \item $\rich$ is \pref[$2$-explainable]{propty:obv}, but not \pref[$1$-explainable]{propty:obv},
    \item $\simps$ and $\shann$ are \pref[$2r-2$-explainable]{propty:obv}, but not \pref[$2r-3$-explainable]{propty:obv}.
  \end{enumerate}
\end{theorem}
While the proof can be found in \cref{proof:theorem:explain:bounds}, we want to give some further insights here:
For $\dns$ it holds---due to its lexicographic nature---that
$\fun{\rdistriv}(\Elec, S_1, S_2)_1 < \fun{\rdistriv}(\Elec, S_2, S_1)_1
\iff \fun{\dns}(\Elec, S_1)>\fun{\dns}(\Elec, S_2)$ so that
\begin{gather*}
  e_{\dns}(\rdistriv^{(1)}, \rdistriv^{(2)}, \rdistrividx)\\
  = \begin{cases}
    \mathit{less} &\text{if }\dim(\rdistrividx)\geq 1 \text{ and }\rdistriv^{(1)}_1 > \rdistriv^{(2)}_1\\
    \mathit{equal}&\text{if } \dim(\rdistrividx) = 0 \\
    \mathit{more} &\text{if } \dim(\rdistrividx)\geq 1 \text{ and }\rdistriv^{(1)}_1 < \rdistriv^{(2)}_1
  \end{cases}
\end{gather*}
can be chosen as an $\dns$-explainable function, which uses at most index $1$ in the $\rdistriv$ vectors.

For $\rich$ it follows directly from its definition that
\begin{gather*}
  e_{\rich}(\rdistriv^{(1)}, \rdistriv^{(2)}, \rdistrividx)\\
  = \begin{cases}
    \mathit{less} & \text{if }\dim(\rdistrividx)\geq 1 \text{ and }\rdistrividx_1 = 1\\
    &\text{and }\rdistriv^{(1)}_1 > \rdistriv^{(2)}_1\\ 
    \mathit{equal}&\text{if } \dim(\rdistrividx) = 0 \text{ or }  \rdistrividx_1 > 1\\
    \mathit{more}&\text{if }\dim(\rdistrividx)\geq 1 \text{ and }\rdistrividx_1 = 1\\
    &\text{and }\rdistriv^{(1)}_1 < \rdistriv^{(2)}_1 \end{cases}
\end{gather*}
can be chosen as a $\rich$-explainable function, using at most $2$ indices (index $1$ in $\rdistrividx$ and index $1$ in the $\rdistriv$ vectors).
However, it is not \pref[$1$-explainable]{propty:obv}.
One way to explain this is that a diversity index being \pref[$1$-explainable]{propty:obv} and fulfilling \pref{propty:pfgpmax} categorizes two committees as equally diverse if and only if their $\distri$ vectors are equal and hence the $\rdistriv$ vectors have the dimension zero, which is clearly not the case for $\rich$ 
(consider e.g.~$\ExCom_1$ and $\ExCom_2$ from \cref{exp:1}).

While $\dns$ and $\rich$ hence only need to use a small and constant number of indices and are therefore easy to comprehend, this is not the case for $\simps$ and $\shann$.
They are \pref[$2r-2$-explainable]{propty:obv}, because the last element in $\rdistrividx$ and in the $\rdistriv$ vectors simply do not need to be used because they can be reconstructed based on the other entries.
Thus, the result that there is no $\simps$- and $\shann$-explainable function using at most $2r-3$ indices is more interesting, which clearly distinguishes $\shann$ and $\simps$ from $\rich$ and $\dns$.

\appendixproof{theorem:explain:bounds}{
  (1) Let $t=(\fun{\rdistriv}(\Elec, S_1, S_2),\fun{\rdistriv}(\Elec, S_2, S_1),\rdistrividx)\in T$, $\distri^{(1)}=\fun{\distri}(\Elec, S_1)$, $\distri^{(2)}=\fun{\distri}(\Elec, S_2)$, $l = \rdistrividx_1$, $\eta = \funbc{\min}{m, k}+1$, $\rdistriv^{(1)}=\fun{\rdistriv}(\Elec, S_1, S_2)$, and $\rdistriv^{(2)}=\fun{\rdistriv}(\Elec, S_2, S_1)$.
  We first show that
  \begin{equation*}
    \fun{\dns}(\Elec, S_1)>\fun{\dns}(\Elec, S_2)\Rightarrow \rdistriv^{(1)}_1 < \rdistriv^{(2)}_1.
  \end{equation*}

  For this, suppose that $\fun{\dns}(\Elec, S_1)>\fun{\dns}(\Elec, S_2)$ and $\rdistriv^{(1)}_1 > \rdistriv^{(2)}_1$.
  Then it holds that
  \begin{equation*}
    \forall i \in\bcks[l-1]: \distri^{(1)}_{i}=\distri^{(2)}_{i}\land\fun{\sigma}(\Elec, S_1)_i=\fun{\sigma}(\Elec, S_2)_i
  \end{equation*}
  and thus
  $\fun{\sigma}(\Elec, S_1)_{l}<\fun{\sigma}(\Elec, S_2)_{l}$ because $\rdistriv^{(1)}_1 > \rdistriv^{(2)}_1\iff \distri^{(1)}_{l}>\distri^{(2)}_{l}$, which means that $S_1$ has more labels occurring $l-1$ times than $S_2$ and thus fewer labels occurring at least $l$ times than $S_2$.
  It holds that
  \begin{align*}
    & \fun{\dns}(\Elec, S_1) - \fun{\dns}(\Elec, S_2)\\
    ={} & \sum_{i=1}^{k} \eta^{k+1-i} \cdot \paren(\abs{\fun{\sigma}(\Elec,S_1)_i}-\abs{\fun{\sigma}(\Elec,S_2)_i})\\
    ={} & \eta^{k+1-l} \cdot \paren(\abs{\fun{\sigma}(\Elec,S_1)_{l}}-\abs{\fun{\sigma}(\Elec,S_2)_{l}})\\
    {}&+ \sum_{i=l+1}^{k} \eta^{k+1-i} \cdot \paren(\abs{\fun{\sigma}(\Elec,S_1)_i}-\abs{\fun{\sigma}(\Elec,S_2)_i})
  \end{align*}
  As it holds that $\fun{\sigma}(\Elec,S_1)_{l}<\fun{\sigma}(\Elec,S_2)_{l}$ and
  $\forall i\in\bcks[k]\setminus\bcks[l]: \fun{\sigma}(\Elec,S_1)_i \leq \eta -1, \fun{\sigma}(\Elec,S_2)_i \geq 0$,
  it follows (with the help of the geometric series formula) that
  \begin{align*}
    &\fun{\dns}(\Elec, S_1) - \fun{\dns}(\Elec, S_2)\\
    \leq &  -\eta^{k+1-l} + \sum_{i=l+1}^{k} \eta^{k+1-i}\cdot \paren(\eta -1)\\
    = & -\eta^{k+1-l} + \paren(\eta -1) \sum_{i=l+1}^{k} \eta^{k+1-i} \\
    = & -\eta^{k+1-l} + \paren(\eta -1) \frac{\eta^{k+1-l}-\eta}{\eta -1}\\
    = & -\eta^{k+1-l} + \eta^{k+1-l}-\eta =-\eta 
    < 0,
  \end{align*}
  a contradiction to the assumption that $\fun{\dns}(\Elec, S_1)>\fun{\dns}(\Elec, S_2)$.
  
  Next, we show $\rdistriv^{(1)}_1 < \rdistriv^{(2)}_1 \Rightarrow \fun{\dns}(\Elec, S_1)>\fun{\dns}(\Elec, S_2)$.
  For this, suppose that $\rdistriv^{(1)}_1 < \rdistriv^{(2)}_1$.
  Thus, $\forall i \in\brcs{1,\dots, l-1}: \distri^{(1)}_{i}=\distri^{(2)}_{i}\land\fun{\sigma}(\Elec,S_1)_i=\fun{\sigma}(\Elec,S_2)_i$ and hence
  $\fun{\sigma}(\Elec,S_1)_l>\fun{\sigma}(\Elec,S_2)_l$ because $\rdistriv^{(1)}_1 < \rdistriv^{(2)}_1\iff \distri^{(1)}_{l}<\distri^{(2)}_{l}$, which means that $S_2$ has more labels occurring $l-1$ times than $S_1$ and thus fewer labels occurring at least $l$ times than $S_1$.
  Based on this, it holds that
  \begin{align*}
    & \fun{\dns}(\Elec, S_1)-\fun{\dns}(\Elec, S_2)\\
    ={}& \sum_{i=1}^{k} \eta^{k+1-i} \cdot \paren(\abs{\fun{\sigma}(\Elec,S_1)_i}-\abs{\fun{\Elec,\sigma}(S_2)_i}) \\
    ={}& \eta^{k+1-l} \cdot \paren(\abs{\fun{\sigma}(\Elec,S_1)_l}-\abs{\fun{\sigma}(\Elec,S_2)_l})\\
    {}& + \sum_{i=l+1}^{k} \eta^{k+1-i} \cdot \paren(\abs{\fun{\sigma}(\Elec,S_1)_i}-\abs{\fun{\sigma}(\Elec,S_2)_i}).
  \end{align*}
  As $\fun{\sigma}(\Elec,S_1)_l>\fun{\sigma}(\Elec,S_2)_l$
  and $\forall i\in\bcks[k]: \eta-1\geq \fun{\sigma}(\Elec,S')_i\geq 0$ for $S'\in\brcs{S_1,S_2}$,
  it follows (with the help of the geometric series formula) that
  \begin{align*}
    & \fun{\dns}(\Elec, S_1)-\fun{\dns}(\Elec, S_2)\\
    \geq {}& \eta^{k+1-l} - \sum_{i=l+1}^{k} \paren(\eta-1) \cdot \eta^{k+1-i}\\
    >{}& \eta^{k+1-l} - \eta -\sum_{i=l+1}^{k} \paren(\eta-1) \cdot \eta^{k+1-i}\\
    ={}& \eta^{k+1-l} - \eta + \eta - \eta^{k+1-l} = 0.
  \end{align*}
  Thus, $\fun{\dns}(\Elec, S_2)<\fun{\dns}(\Elec, S_1)$.
  
  Thus, we have $\rdistriv^{(1)}_1 < \rdistriv^{(2)}_1 \iff \fun{\dns}(\Elec, S_1)>\fun{\dns}(\Elec, S_2)$ overall, so that
  \begin{gather*}
    e_{\dns}(\rdistriv^{(1)}, \rdistriv^{(2)}, \rdistrividx)\\
    = \begin{cases}
      \mathit{less} &\text{if }\dim(\rdistrividx)\geq 1 \text{ and }\rdistriv^{(1)}_1 > \rdistriv^{(2)}_1\\
      \mathit{equal}&\text{if } \dim(\rdistrividx) = 0 \\
      \mathit{more} &\text{if } \dim(\rdistrividx)\geq 1 \text{ and }\rdistriv^{(1)}_1 < \rdistriv^{(2)}_1
    \end{cases}
  \end{gather*}
  can be chosen as a $\dns$-explainable function, which uses at most $1$ index, namely index $1$ in the $\rdistriv$ vectors.

  (2) We first prove that $\rich$ is \pref[$2$-explainable]{propty:obv}:
  As
  \begin{equation*}
    \fun{\rich}(\Elec,S_1) = \sum_{i=1}^k \distri^{(1)}_{i+1} = m - \distri^{(1)}_{1},
  \end{equation*}
  we have
  \begin{gather*}
    e_{\rich}(\rdistriv^{(1)}, \rdistriv^{(2)}, \rdistrividx)\\
    = \begin{cases}
      \mathit{less} & \text{if }\dim(\rdistrividx)\geq 1 \text{ and }\rdistrividx_1 = 1\\
      &\text{and }\rdistriv^{(1)}_1 > \rdistriv^{(2)}_1\\
      \mathit{equal} & \text{if } \dim(\rdistrividx) = 0 \text{ or }  \rdistrividx_1 > 1\\
      \mathit{more} & \text{if }\dim(\rdistrividx)\geq 1 \text{ and }\rdistrividx_1 = 1\\
      &\text{and }\rdistriv^{(1)}_1 < \rdistriv^{(2)}_1
    \end{cases}
  \end{gather*}
  can be chosen as a $\rich$-explainable function, which uses at most $2$ indices, namely index $1$ in the $\rdistriv$ vectors and index $1$ in $\rdistrividx$.

  Next, we prove that $\rich$ is not \pref[$1$-explainable]{propty:obv}.
  We show in \cref{proof:theorem:charac} that each diversity index being \pref[$1$-explainable]{propty:obv} and fulfilling \pref{propty:pfgpmax} categorizes two committees as equally diverse if and only if their $\distri$ vectors are equal and hence the $\rdistriv$ vectors have the dimension zero.
  However, this is not the case for $\rich$: Consider e.g.~$\ExCom_1$ and $\ExCom_2$ from \cref{exp:1} which $\rich$ classifies as equally diverse although the $\distri$ vectors are not the same.

  (3) We first prove that $\simps$ and $\shann$ are \pref[$2r-2$-explainable]{propty:obv}. Let $t=(\fun{\rdistriv}(\Elec, S_1, S_2),\fun{\rdistriv}(\Elec, S_2, S_1),\rdistrividx)\in T$, $\rdistriv^{(1)}=\fun{\rdistriv}(\Elec, S_1, S_2)$, and $\rdistriv^{(2)}=\fun{\rdistriv}(\Elec, S_2, S_1)$. $\simps$ and $\shann$ are \pref[$2r-2$-explainable]{propty:obv} by using all indices but the last, i.e.~by using indices $\bcks[r]\setminus \brcs{r}$ in $\rdistrividx$ and indices $\bcks[r]\setminus \brcs{r}$ in $\fun{\rdistriv}(\Elec, S_1, S_2)$ and $\fun{\rdistriv}(\Elec, S_2, S_1)$.
  This is possible, as the last entries can be reconstructed based on the other entries:
  For this, we will use that the following holds:
  \begin{align}
    \sum_{i=1}^r \rdistriv^{(1)}_i &= \sum_{i=1}^r\rdistriv^{(2)}_i \label{helper:expla:one}\\
    \sum_{i=1}^r \rdistrividx_i\cdot \rdistriv^{(1)}_i &= \sum_{i=1}^r \rdistrividx_i\cdot \rdistriv^{(2)}_i.\label{helper:expla:two}
  \end{align}
  Let w.l.o.g.~$\sum_{i=1}^{r-1} \rdistriv^{(1)}_i < \sum_{i=1}^{r-1}\rdistriv^{(2)}_i$ (these sums cannot be equal because of \cref{helper:expla:one} and as $\forall i\in\bcks[r]: \rdistriv^{(1)}_i \neq \rdistriv^{(2)}_i$).
  Let
  \begin{equation*}
    x\ceq \sum_{i=1}^{r-1}\rdistriv^{(2)}_i - \rdistriv^{(1)}_i.
  \end{equation*}
  Then, it follows from \cref{helper:expla:one} that $\exists d\in\Nzero$ so that $\rdistriv^{(2)}_r = d$ and $\rdistriv^{(1)}_r = x + d$.
  It follows from \cref{helper:expla:two} that
  \begin{align*}
    \rdistrividx_r\paren(\rdistriv^{(1)}_r - \rdistriv^{(2)}_r) ={}& \rdistrividx_r\cdot x\\
    ={}& \sum_{i=1}^{r-1} \rdistrividx_i\paren( \rdistriv^{(2)}_i - \rdistriv^{(1)}_i)
  \end{align*}
  and hence
  \begin{equation*}
    \rdistrividx_r = \frac{\sum_{i=1}^{r-1} \rdistrividx_i\paren( \rdistriv^{(2)}_i - \rdistriv^{(1)}_i)}{x}\eqqcolon y.
  \end{equation*}
  Based on this, it holds for $\simps$ that
  \begin{align*}
    {}&\simps(\Elec, S_2) - \simps(\Elec, S_1)\\
    ={}& -\paren(\sum_{i=1}^{r}\paren(\rdistriv^{(2)}_i - \rdistriv^{(1)}_i)\paren(\frac{\rdistrividx_i-1}{k})^2)\\
    ={}& -\paren(\sum_{i=1}^{r-1}\paren(\rdistriv^{(2)}_i - \rdistriv^{(1)}_i)\paren(\frac{\rdistrividx_i-1}{k})^2)\\
    {}& + x\paren(\frac{y-1}{k})^2\\
    \eqqcolon{}& h_{\simps}(\rdistriv^{(1)}, \rdistriv^{(2)}, \rdistrividx).
  \end{align*}
  Then
  \begin{gather*}
    e_{\simps}(\rdistriv^{(1)}, \rdistriv^{(2)}, \rdistrividx)\\
    = \begin{cases}
      \mathit{less} &\text{if } h_{\simps}(\rdistriv^{(1)}, \rdistriv^{(2)}, \rdistrividx) > 0\\
      \mathit{equal}&\text{if } h_{\simps}(\rdistriv^{(1)}, \rdistriv^{(2)}, \rdistrividx) = 0 \\
      \mathit{more} &\text{if } h_{\simps}(\rdistriv^{(1)}, \rdistriv^{(2)}, \rdistrividx) < 0
    \end{cases}
  \end{gather*}
  can be chosen as a $\simps$-explainable function, which uses at most $2r-2$ indices.

  Analogously,
  it holds for $\shann$ that
  \begin{align*}
    {}&\shann(\Elec, S_2) - \shann(\Elec, S_1)\\
    ={}& -\paren(\sum_{i=1}^{r}\paren(\rdistriv^{(2)}_i - \rdistriv^{(1)}_i) \frac{\rdistrividx_i-1}{k} \tlog{\frac{\rdistrividx_i-1}{k}})\\
    ={}& -\paren(\sum_{i=1}^{r-1}\paren(\rdistriv^{(2)}_i - \rdistriv^{(1)}_i)\frac{\rdistrividx_i-1}{k} \tlog{\frac{\rdistrividx_i-1}{k}})\\
    {}&+ x \frac{y-1}{k} \tlog{\frac{y-1}{k}} \eqqcolon h_{\shann}(\rdistriv^{(1)}, \rdistriv^{(2)}, \rdistrividx).
  \end{align*}
  Then
  \begin{gather*}
    e_{\shann}(\rdistriv^{(1)}, \rdistriv^{(2)}, \rdistrividx)\\
    = \begin{cases}
      \mathit{less} &\text{if } h_{\shann}(\rdistriv^{(1)}, \rdistriv^{(2)}, \rdistrividx) > 0\\
      \mathit{equal}&\text{if } h_{\shann}(\rdistriv^{(1)}, \rdistriv^{(2)}, \rdistrividx) = 0 \\
      \mathit{more} &\text{if } h_{\shann}(\rdistriv^{(1)}, \rdistriv^{(2)}, \rdistrividx) < 0
    \end{cases}
  \end{gather*}
  can be chosen as a $\shann$-explainable function, which uses at most $2r-2$ indices.

  Next, we show that $\simps$ and $\shann$ are not \pref[$2r-3$-explainable]{propty:obv}:
  For a $t=(\fun{\rdistriv}(\Elec, S_1, S_2),\fun{\rdistriv}(\Elec, S_2, S_1),\rdistrividx)\in T$,
  let $I_o$ be the indices \emph{not} used by $e$ in $\rdistrividx$ and $I_d$ be the indices \emph{neither} used in $\fun{\rdistriv}(\Elec, S_1, S_2)$ \emph{nor} in $\fun{\rdistriv}(\Elec, S_2, S_1)$ in the following.
  We show that $\simps$ and $\shann$ are not \pref[$2r-3$-explainable]{propty:obv} for $r=4$ (i.e.~for $\rdistriv$ vectors and $\rdistrividx$ of dimension $4$) by doing the following:
  For each possible $I_o$ and $I_d$ with $\abs{I_o}+\abs{I_d} = 3$, we specify $t_1 = (\rdistriv^{(1)}, \rdistriv^{(2)}, \rdistrividx^{(1)})\in T$ and $t_2=(\rdistriv^{(3)}, \rdistriv^{(4)}, \rdistrividx^{(2)})\in T$ so that
  \begin{enumerate}
    \item $t_1$ and $t_2$ cannot be distinguished without using indices from $I_o$ or $I_d$, i.e.~$\forall i\in\bcks[4]\setminus I_o: \rdistrividx^{(1)}_i = \rdistrividx^{(2)}_i$ and $\forall i\in\bcks[4]\setminus I_d: \rdistriv^{(1)}_i = \rdistriv^{(3)}_i$ and $\rdistriv^{(2)}_i = \rdistriv^{(4)}_i$,
    \item and for all elections $\Elec_1$ and $\Elec_2$ and $S_1,S_2\in\Rulebase(\Elec_1), S_3,S_4\in \Rulebase(\Elec_2)$ so that
    \begin{align*}
      \rdistriv^{(1)}={}&\fun{\rdistriv}(\Elec_1, S_1, S_2),\\
      \rdistriv^{(2)}={}&\fun{\rdistriv}(\Elec_1, S_2, S_1)\\
      \rdistriv^{(3)}={}&\fun{\rdistriv}(\Elec_2, S_3, S_4),\\
      \rdistriv^{(4)}={}&\fun{\rdistriv}(\Elec_2, S_4, S_3),
    \end{align*}
  \end{enumerate}
  and $\rdistrividx^{(1)}$ and $\rdistrividx^{(2)}$ being the vector of elements in ascending order of $\fun{I_R}(\Elec_1, S_1, S_2)$ and $\fun{I_R}(\Elec_2, S_3, S_4)$, respectively, it holds, for each $D\in\brcs{\simps,\shann}$, that $D(\Elec, S_1)<D(\Elec, S_2)$ and $D(\Elec, S_3)>D(\Elec, S_4)$.
  Thus, $e$ will misclassify at least one of $t_1$ and $t_2$, as it cannot distinguish between them:
  \begin{itemize}[nosep]
    % |I_d| = 3, |I_rho| = 0
    \item $I_d=\brcs{1, 2, 3}, I_o=\brcs{}$:
    \begin{align*}
      \rdistrividx^{(1)} ={}& (3, 4, 10, 11){}&  \rdistrividx^{(2)} ={}& (3, 4, 10, 11) \\
      \rdistriv^{(1)} ={}& (2, 0, 0, 2){}&  \rdistriv^{(3)} ={}& (0, 5, 0, 2)\\
      \rdistriv^{(2)} ={}& (0, 2, 2, 0){}&  \rdistriv^{(4)} ={}& (4, 0, 3, 0)
    \end{align*}
    \item $I_d=\brcs{1, 2, 4}, I_o=\brcs{}$:
    \begin{align*}
      \rdistrividx^{(1)} ={}& (3, 10, 11, 13){}&  \rdistrividx^{(2)} ={}& (3, 10, 11, 13) \\
      \rdistriv^{(1)} ={}& (1, 0, 4, 1){}&  \rdistriv^{(3)} ={}& (0, 4, 4, 0)\\
      \rdistriv^{(2)} ={}& (0, 6, 0, 0){}&  \rdistriv^{(4)} ={}& (2, 0, 0, 6)
    \end{align*}
    \item $I_d=\brcs{1, 3, 4}, I_o=\brcs{}$:
    \begin{align*}
      \rdistrividx^{(1)} ={}& (3, 4, 10, 15){}&  \rdistrividx^{(2)} ={}& (3, 4, 10, 15) \\
      \rdistriv^{(1)} ={}& (0, 8, 0, 4){}&  \rdistriv^{(3)} ={}& (0, 8, 4, 0)\\
      \rdistriv^{(2)} ={}& (4, 0, 8, 0){}&  \rdistriv^{(4)} ={}& (9, 0, 0, 3)
    \end{align*}
    \item $I_d=\brcs{2, 3, 4}, I_o=\brcs{}$:
    \begin{align*}
      \rdistrividx^{(1)} ={}& (4, 5, 8, 11){}&  \rdistrividx^{(2)} ={}& (4, 5, 8, 11) \\
      \rdistriv^{(1)} ={}& (3, 1, 0, 5){}&  \rdistriv^{(3)} ={}& (3, 0, 7, 0)\\
      \rdistriv^{(2)} ={}& (0, 0, 9, 0){}&  \rdistriv^{(4)} ={}& (0, 7, 0, 3)
    \end{align*}
    %
    % |I_d| = 2, |I_rho| = 1
    %
    \item $I_d=\brcs{1,2}, I_o=\brcs{1}$:
    \begin{align*}
      \rdistrividx^{(1)} ={}& (3, 8, 14, 15){}&  \rdistrividx^{(2)} ={}& (7, 8, 14, 15) \\
      \rdistriv^{(1)} ={}& (1, 0, 9, 0){}&  \rdistriv^{(3)} ={}& (5, 0, 9, 0)\\
      \rdistriv^{(2)} ={}& (0, 3, 0, 7){}&  \rdistriv^{(4)} ={}& (0, 7, 0, 7)
    \end{align*}
    \item $I_d=\brcs{1,2}, I_o=\brcs{2}$:
    \begin{align*}
      \rdistrividx^{(1)} ={}& (7, 9, 11, 12){}&  \rdistrividx^{(2)} ={}& (7, 8, 11, 12) \\
      \rdistriv^{(1)} ={}& (0, 4, 0, 4){}&  \rdistriv^{(3)} ={}& (0, 8, 0, 4)\\
      \rdistriv^{(2)} ={}& (1, 0, 7, 0){}&  \rdistriv^{(4)} ={}& (5, 0, 7, 0)
    \end{align*}
    \item $I_d=\brcs{1,2}, I_o=\brcs{3}$:
    \begin{align*}
      \rdistrividx^{(1)} ={}& (3, 6, 9, 12){}&  \rdistrividx^{(2)} ={}& (3, 6, 11, 12) \\
      \rdistriv^{(1)} ={}& (1, 3, 0, 5){}&  \rdistriv^{(3)} ={}& (0, 9, 0, 5)\\
      \rdistriv^{(2)} ={}& (0, 0, 9, 0){}&  \rdistriv^{(4)} ={}& (5, 0, 9, 0)
    \end{align*}
    \item $I_d=\brcs{1,2}, I_o=\brcs{4}$:
    \begin{align*}
      \rdistrividx^{(1)} ={}& (5, 8, 12, 13){}&  \rdistrividx^{(2)} ={}& (5, 8, 12, 15) \\
      \rdistriv^{(1)} ={}& (2, 0, 9, 0){}&  \rdistriv^{(3)} ={}& (0, 0, 9, 0)\\
      \rdistriv^{(2)} ={}& (0, 5, 0, 6){}&  \rdistriv^{(4)} ={}& (2, 1, 0, 6)
    \end{align*}
    \item $I_d=\brcs{1,3}, I_o=\brcs{1}$:
    \begin{align*}
      \rdistrividx^{(1)} ={}& (6, 7, 12, 15){}&  \rdistrividx^{(2)} ={}& (3, 7, 12, 15) \\
      \rdistriv^{(1)} ={}& (0, 9, 0, 3){}&  \rdistriv^{(3)} ={}& (0, 9, 0, 3)\\
      \rdistriv^{(2)} ={}& (6, 0, 6, 0){}&  \rdistriv^{(4)} ={}& (4, 0, 8, 0)
    \end{align*}
    \item $I_d=\brcs{1,3}, I_o=\brcs{1}$:
    \begin{align*}
      \rdistrividx^{(1)} ={}& (6, 7, 12, 15){}&  \rdistrividx^{(2)} ={}& (3, 7, 12, 15) \\
      \rdistriv^{(1)} ={}& (0, 9, 0, 3){}&  \rdistriv^{(3)} ={}& (0, 9, 0, 3)\\
      \rdistriv^{(2)} ={}& (6, 0, 6, 0){}&  \rdistriv^{(4)} ={}& (4, 0, 8, 0)
    \end{align*}
    \item $I_d=\brcs{1,3}, I_o=\brcs{2}$:
    \begin{align*}
      \rdistrividx^{(1)} ={}& (3, 9, 13, 15){}&  \rdistrividx^{(2)} ={}& (3, 4, 13, 15) \\
      \rdistriv^{(1)} ={}& (1, 0, 6, 0){}&  \rdistriv^{(3)} ={}& (3, 0, 4, 0)\\
      \rdistriv^{(2)} ={}& (0, 4, 0, 3){}&  \rdistriv^{(4)} ={}& (0, 4, 0, 3)
    \end{align*}
    \item $I_d=\brcs{1,3}, I_o=\brcs{3}$:
    \begin{align*}
      \rdistrividx^{(1)} ={}& (5, 8, 10, 15){}&  \rdistrividx^{(2)} ={}& (5, 8, 14, 15) \\
      \rdistriv^{(1)} ={}& (1, 5, 0, 3){}&  \rdistriv^{(3)} ={}& (0, 5, 0, 3)\\
      \rdistriv^{(2)} ={}& (0, 0, 9, 0){}&  \rdistriv^{(4)} ={}& (3, 0, 5, 0)
    \end{align*}
    \item $I_d=\brcs{1,3}, I_o=\brcs{4}$:
    \begin{align*}
      \rdistrividx^{(1)} ={}& (3, 4, 6, 10){}&  \rdistrividx^{(2)} ={}& (3, 4, 6, 7) \\
      \rdistriv^{(1)} ={}& (0, 5, 0, 1){}&  \rdistriv^{(3)} ={}& (0, 5, 0, 1)\\
      \rdistriv^{(2)} ={}& (2, 0, 4, 0){}&  \rdistriv^{(4)} ={}& (3, 0, 3, 0)
    \end{align*}
    \item $I_d=\brcs{1,4}, I_o=\brcs{1}$:
    \begin{align*}
      \rdistrividx^{(1)} ={}& (6, 7, 11, 15){}&  \rdistrividx^{(2)} ={}& (3, 7, 11, 15) \\
      \rdistriv^{(1)} ={}& (0, 7, 0, 2){}&  \rdistriv^{(3)} ={}& (0, 7, 0, 1)\\
      \rdistriv^{(2)} ={}& (4, 0, 5, 0){}&  \rdistriv^{(4)} ={}& (3, 0, 5, 0)
    \end{align*}
    \item $I_d=\brcs{1,4}, I_o=\brcs{2}$:
    \begin{align*}
      \rdistrividx^{(1)} ={}& (3, 7, 8, 15){}&  \rdistrividx^{(2)} ={}& (3, 4, 8, 15) \\
      \rdistriv^{(1)} ={}& (3, 0, 4, 1){}&  \rdistriv^{(3)} ={}& (5, 0, 4, 0)\\
      \rdistriv^{(2)} ={}& (0, 8, 0, 0){}&  \rdistriv^{(4)} ={}& (0, 8, 0, 1)
    \end{align*}
    \item $I_d=\brcs{1,4}, I_o=\brcs{3}$:
    \begin{align*}
      \rdistrividx^{(1)} ={}& (3, 10, 14, 15){}&  \rdistrividx^{(2)} ={}& (3, 10, 11, 15) \\
      \rdistriv^{(1)} ={}& (1, 0, 8, 0){}&  \rdistriv^{(3)} ={}& (0, 0, 8, 0)\\
      \rdistriv^{(2)} ={}& (0, 4, 0, 5){}&  \rdistriv^{(4)} ={}& (1, 4, 0, 3)
    \end{align*}
    \item $I_d=\brcs{1,4}, I_o=\brcs{4}$:
    \begin{align*}
      \rdistrividx^{(1)} ={}& (3, 4, 8, 9){}&  \rdistrividx^{(2)} ={}& (3, 4, 8, 15) \\
      \rdistriv^{(1)} ={}& (6, 0, 4, 0){}&  \rdistriv^{(3)} ={}& (5, 0, 4, 0)\\
      \rdistriv^{(2)} ={}& (0, 8, 0, 2){}&  \rdistriv^{(4)} ={}& (0, 8, 0, 1)
    \end{align*}
    \item $I_d=\brcs{2,3}, I_o=\brcs{1}$:
    \begin{align*}
      \rdistrividx^{(1)} ={}& (5, 6, 9, 11){}&  \rdistrividx^{(2)} ={}& (3, 6, 9, 11) \\
      \rdistriv^{(1)} ={}& (0, 6, 0, 3){}&  \rdistriv^{(3)} ={}& (0, 8, 0, 3)\\
      \rdistriv^{(2)} ={}& (3, 0, 6, 0){}&  \rdistriv^{(4)} ={}& (3, 0, 8, 0)
    \end{align*}
    \item $I_d=\brcs{2,3}, I_o=\brcs{2}$:
    \begin{align*}
      \rdistrividx^{(1)} ={}& (3, 4, 10, 14){}&  \rdistrividx^{(2)} ={}& (3, 8, 10, 14) \\
      \rdistriv^{(1)} ={}& (0, 3, 0, 1){}&  \rdistriv^{(3)} ={}& (0, 9, 0, 1)\\
      \rdistriv^{(2)} ={}& (2, 0, 2, 0){}&  \rdistriv^{(4)} ={}& (2, 0, 8, 0)
    \end{align*}
    \item $I_d=\brcs{2,3}, I_o=\brcs{3}$:
    \begin{align*}
      \rdistrividx^{(1)} ={}& (4, 8, 10, 15){}&  \rdistrividx^{(2)} ={}& (4, 8, 14, 15) \\
      \rdistriv^{(1)} ={}& (0, 8, 0, 2){}&  \rdistriv^{(3)} ={}& (0, 2, 0, 2)\\
      \rdistriv^{(2)} ={}& (1, 0, 9, 0){}&  \rdistriv^{(4)} ={}& (1, 0, 3, 0)
    \end{align*}
    \item $I_d=\brcs{2,3}, I_o=\brcs{4}$:
    \begin{align*}
      \rdistrividx^{(1)} ={}& (3, 6, 9, 12){}&  \rdistrividx^{(2)} ={}& (3, 6, 9, 10) \\
      \rdistriv^{(1)} ={}& (0, 5, 0, 3){}&  \rdistriv^{(3)} ={}& (0, 3, 0, 3)\\
      \rdistriv^{(2)} ={}& (1, 0, 7, 0){}&  \rdistriv^{(4)} ={}& (1, 0, 5, 0)
    \end{align*}
    \item $I_d=\brcs{2,4}, I_o=\brcs{1}$:
    \begin{align*}
      \rdistrividx^{(1)} ={}& (10, 11, 12, 13){}&  \rdistrividx^{(2)} ={}& (8, 11, 12, 13) \\
      \rdistriv^{(1)} ={}& (0, 5, 0, 3){}&  \rdistriv^{(3)} ={}& (0, 6, 0, 2)\\
      \rdistriv^{(2)} ={}& (1, 0, 7, 0){}&  \rdistriv^{(4)} ={}& (1, 0, 7, 0)
    \end{align*}
    \item $I_d=\brcs{2,4}, I_o=\brcs{2}$:
    \begin{align*}
      \rdistrividx^{(1)} ={}& (3, 8, 10, 12){}&  \rdistrividx^{(2)} ={}& (3, 4, 10, 12) \\
      \rdistriv^{(1)} ={}& (2, 0, 7, 0){}&  \rdistriv^{(3)} ={}& (2, 0, 7, 0)\\
      \rdistriv^{(2)} ={}& (0, 8, 0, 1){}&  \rdistriv^{(4)} ={}& (0, 4, 0, 5)
    \end{align*}
    \item $I_d=\brcs{2,4}, I_o=\brcs{3}$:
    \begin{align*}
      \rdistrividx^{(1)} ={}& (4, 5, 14, 15){}&  \rdistrividx^{(2)} ={}& (4, 5, 10, 15) \\
      \rdistriv^{(1)} ={}& (5, 0, 5, 0){}&  \rdistriv^{(3)} ={}& (5, 0, 5, 0)\\
      \rdistriv^{(2)} ={}& (0, 6, 0, 4){}&  \rdistriv^{(4)} ={}& (0, 8, 0, 2)
    \end{align*}
    \item $I_d=\brcs{2,4}, I_o=\brcs{4}$:
    \begin{align*}
      \rdistrividx^{(1)} ={}& (3, 6, 8, 9){}&  \rdistrividx^{(2)} ={}& (3, 6, 8, 15) \\
      \rdistriv^{(1)} ={}& (1, 0, 6, 0){}&  \rdistriv^{(3)} ={}& (1, 0, 6, 0)\\
      \rdistriv^{(2)} ={}& (0, 4, 0, 3){}&  \rdistriv^{(4)} ={}& (0, 6, 0, 1)
    \end{align*}
    \item $I_d=\brcs{3,4}, I_o=\brcs{1}$:
    \begin{align*}
      \rdistrividx^{(1)} ={}& (5, 7, 11, 15){}&  \rdistrividx^{(2)} ={}& (6, 7, 11, 15) \\
      \rdistriv^{(1)} ={}& (4, 0, 4, 0){}&  \rdistriv^{(3)} ={}& (4, 0, 5, 0)\\
      \rdistriv^{(2)} ={}& (0, 7, 0, 1){}&  \rdistriv^{(4)} ={}& (0, 7, 0, 2)
    \end{align*}
    \item $I_d=\brcs{3,4}, I_o=\brcs{2}$:
    \begin{align*}
      \rdistrividx^{(1)} ={}& (3, 4, 7, 11){}&  \rdistrividx^{(2)} ={}& (3, 6, 7, 11) \\
      \rdistriv^{(1)} ={}& (0, 4, 0, 1){}&  \rdistriv^{(3)} ={}& (0, 4, 0, 0)\\
      \rdistriv^{(2)} ={}& (2, 0, 3, 0){}&  \rdistriv^{(4)} ={}& (2, 0, 1, 1)
    \end{align*}
    \item $I_d=\brcs{3,4}, I_o=\brcs{3}$:
    \begin{align*}
      \rdistrividx^{(1)} ={}& (6, 9, 11, 15){}&  \rdistrividx^{(2)} ={}& (6, 9, 12, 15) \\
      \rdistriv^{(1)} ={}& (2, 0, 6, 0){}&  \rdistriv^{(3)} ={}& (2, 0, 8, 0)\\
      \rdistriv^{(2)} ={}& (0, 7, 0, 1){}&  \rdistriv^{(4)} ={}& (0, 7, 0, 3)
    \end{align*}
    \item $I_d=\brcs{3,4}, I_o=\brcs{4}$:
    \begin{align*}
      \rdistrividx^{(1)} ={}& (6, 8, 10, 14){}&  \rdistrividx^{(2)} ={}& (6, 8, 10, 11) \\
      \rdistriv^{(1)} ={}& (0, 6, 0, 1){}&  \rdistriv^{(3)} ={}& (0, 6, 0, 4)\\
      \rdistriv^{(2)} ={}& (2, 0, 5, 0){}&  \rdistriv^{(4)} ={}& (2, 0, 8, 0)
    \end{align*}
    %
    % |I_d| = 1, |I_rho| = 2
    %
    \item $I_d=\brcs{1}$ or $I_d=\brcs{2}$ or $I_d=\brcs{3}$ or $I_d=\brcs{4}$, $I_o=\brcs{1,2}$:
    \begin{align*}
      \rdistrividx^{(1)} ={}& (6, 7, 12, 15) & \rdistrividx^{(2)} ={}& (3, 5, 12, 15) \\
      \rdistriv^{(1)} ={}& (0, 6, 0, 2) & \rdistriv^{(3)} ={}& (0, 6, 0, 2) \\
      \rdistriv^{(2)} ={}& (4, 0, 4, 0) & \rdistriv^{(4)} ={}& (4, 0, 4, 0)
    \end{align*}
    \item $I_d=\brcs{1}$ or $I_d=\brcs{2}$ or $I_d=\brcs{3}$ or $I_d=\brcs{4}$, $I_o=\brcs{1,3}$:
    \begin{align*}
      \rdistrividx^{(1)} ={}& (5, 6, 10, 13) & \rdistrividx^{(2)} ={}& (3, 6, 11, 13) \\
      \rdistriv^{(1)} ={}& (0, 8, 0, 4) & \rdistriv^{(3)} ={}& (0, 8, 0, 4) \\
      \rdistriv^{(2)} ={}& (4, 0, 8, 0) & \rdistriv^{(4)} ={}& (4, 0, 8, 0) 
    \end{align*}
    \item $I_d=\brcs{1}$ or $I_d=\brcs{2}$ or $I_d=\brcs{3}$ or $I_d=\brcs{4}$, $I_o=\brcs{1,4}$:
    \begin{align*}
      \rdistrividx^{(1)} ={}& (3, 6, 8, 9) & \rdistrividx^{(2)} ={}& (5, 6, 8, 11) \\
      \rdistriv^{(1)} ={}& (2, 0, 6, 0) & \rdistriv^{(3)} ={}& (2, 0, 6, 0) \\
      \rdistriv^{(2)} ={}& (0, 6, 0, 2) & \rdistriv^{(4)} ={}& (0, 6, 0, 2)
    \end{align*}
    \item $I_d=\brcs{1}$ or $I_d=\brcs{2}$ or $I_d=\brcs{3}$ or $I_d=\brcs{4}$, $I_o=\brcs{2,3}$:
    \begin{align*}
      \rdistrividx^{(1)} ={}& (3, 4, 10, 13) & \rdistrividx^{(2)} ={}& (3, 6, 12, 13) \\
      \rdistriv^{(1)} ={}& (0, 5, 0, 3) & \rdistriv^{(3)} ={}& (0, 5, 0, 3) \\
      \rdistriv^{(2)} ={}& (3, 0, 5, 0) & \rdistriv^{(4)} ={}& (3, 0, 5, 0) 
    \end{align*}
    \item $I_d=\brcs{1}$ or $I_d=\brcs{2}$ or $I_d=\brcs{3}$ or $I_d=\brcs{4}$, $I_o=\brcs{2,4}$:
    \begin{align*}
      \rdistrividx^{(1)} ={}& (3, 6, 10, 11) & \rdistrividx^{(2)} ={}& (3, 4, 10, 15) \\
      \rdistriv^{(1)} ={}& (2, 0, 4, 0) & \rdistriv^{(3)} ={}& (2, 0, 4, 0) \\
      \rdistriv^{(2)} ={}& (0, 4, 0, 2) & \rdistriv^{(4)} ={}& (0, 4, 0, 2)  
    \end{align*}
    \item $I_d=\brcs{1}$ or $I_d=\brcs{2}$ or $I_d=\brcs{3}$ or $I_d=\brcs{4}$, $I_o=\brcs{3,4}$:
    \begin{align*}
      \rdistrividx^{(1)} ={}& (3, 4, 6, 12) & \rdistrividx^{(2)} ={}& (3, 4, 5, 6) \\
      \rdistriv^{(1)} ={}& (0, 9, 0, 1) & \rdistriv^{(3)} ={}& (0, 9, 0, 1) \\
      \rdistriv^{(2)} ={}& (4, 0, 6, 0) & \rdistriv^{(4)} ={}& (4, 0, 6, 0)  
    \end{align*}
    %
    % |I_d| = 0, |I_rho| = 3
    % 
    \item $I_d=\brcs{}$, $I_o=\brcs{1,2,3}$:
    \begin{align*}
      \rdistrividx^{(1)} ={}& (6, 9, 14, 15) & \rdistrividx^{(2)} ={}& (4, 5, 12, 15) \\
      \rdistriv^{(1)} ={}& (3, 0, 9, 0) & \rdistriv^{(3)} ={}& (3, 0, 9, 0) \\
      \rdistriv^{(2)} ={}& (0, 6, 0, 6) & \rdistriv^{(4)} ={}& (0, 6, 0, 6) 
    \end{align*}
    \item $I_d=\brcs{}$, $I_o=\brcs{1,2,4}$:
    \begin{align*}
      \rdistrividx^{(1)} ={}& (3, 6, 7, 9) & \rdistrividx^{(2)} ={}& (3, 5, 7, 15) \\
      \rdistriv^{(1)} ={}& (1, 0, 6, 0) & \rdistriv^{(3)} ={}& (1, 0, 6, 0) \\
      \rdistriv^{(2)} ={}& (0, 6, 0, 1) & \rdistriv^{(4)} ={}& (0, 6, 0, 1) 
    \end{align*}
    \item $I_d=\brcs{}$, $I_o=\brcs{1,3,4}$:
    \begin{align*}
      \rdistrividx^{(1)} ={}& (3, 5, 8, 10) & \rdistrividx^{(2)} ={}& (4, 5, 8, 15) \\
      \rdistriv^{(1)} ={}& (5, 0, 5, 0) & \rdistriv^{(3)} ={}& (5, 0, 5, 0) \\
      \rdistriv^{(2)} ={}& (0, 9, 0, 1) & \rdistriv^{(4)} ={}& (0, 9, 0, 1) 
    \end{align*}
    \item $I_d=\brcs{}$, $I_o=\brcs{2,3,4}$:
    \begin{align*}
      \rdistrividx^{(1)} ={}& (3, 5, 9, 14) & \rdistrividx^{(2)} ={}& (3, 10, 13, 14) \\
      \rdistriv^{(1)} ={}& (0, 4, 0, 2) & \rdistriv^{(3)} ={}& (0, 4, 0, 2) \\
      \rdistriv^{(2)} ={}& (1, 0, 5, 0) & \rdistriv^{(4)} ={}& (1, 0, 5, 0) 
    \end{align*}
  \end{itemize}
}

To conclude this section, we answer whether the behavior of $\dns$---which is the only index fulfilling all the properties introduced---can be characterized by (some of) the properties introduced so far:
\begin{theorem}[\appref{theorem:charac}]
  \label{theorem:charac}
  $\dns$ is characterized by \pref[$1$-Explainability]{propty:obv} and \pref{propty:pfgpmax}.
\end{theorem}
\appendixproof{theorem:charac}{
  In the following, we will use that the following equations hold for each $(\rdistriv^{(1)},\rdistriv^{(2)},\rdistrividx)\in T$:
  \begin{align}
    \sum_{i=1}^r \rdistriv^{(1)}_i &= \sum_{i=1}^r\rdistriv^{(2)}_i \label{helper:charac:one}\\
    \sum_{i=1}^r \rdistrividx_i\cdot \rdistriv^{(1)}_i &= \sum_{i=1}^r \rdistrividx_i\cdot \rdistriv^{(2)}_i \label{helper:charac:two}.
  \end{align}
  First, note that $r \in\brcs{1,2}$ is not possible:
  Assume that $r=1$ is possible (i.e.~$\dim(\rdistriv^{(1)})=\dim(\rdistriv^{(2)}) = 1$) and let w.l.o.g. $\rdistriv^{(1)}_1>\rdistriv^{(2)}_1$.
  However, there must be another index $r_2$ so that $\rdistrividx_2 = r_2$ and $\rdistriv^{(1)}_2<\rdistriv^{(2)}_2$ because of \cref{helper:charac:one}, i.e.~$r$ needs to be at least two.
  W.l.o.g., let $\rdistrividx_1 < \rdistrividx_2$.
  Assume that $r=2$ is possible and thus $d\ceq \rdistriv^{(1)}_1-\rdistriv^{(2)}_1 = \rdistriv^{(2)}_2-\rdistriv^{(1)}_2>0$ because of \cref{helper:charac:one}.
  Then, it holds that
  \begin{align*}
    {}&\sum_{i=1}^r \rdistrividx_i\cdot \paren(\rdistriv^{(1)}_i-\rdistriv^{(2)}_i) \\
    ={}& \rdistrividx_1 \cdot d - \rdistrividx_2 \cdot d = d \cdot \paren(\rdistrividx_1-\rdistrividx_2) \neq 0,
  \end{align*}
  a contradiction to \cref*{helper:charac:two}.
  Thus, $r\in\brcs{1,2}$ is not possible.

  Secondly, note that a function explaining a diversity index that fulfills \pref[$1$-Explainability]{propty:obv} and \pref{propty:pfgpmax} clearly does not use one index from the third argument, i.e.~from $\rdistrividx$:
  If it would, it could not distinguish between $(a,b,c),(b,a,c) \in T$ and would hence classify everything as equally diverse.

  Next, for $r\geq 4$, we show that a diversity index fulfilling \pref[$1$-Explainability]{propty:obv} and \pref{propty:pfgpmax} uses only index $1$ of the $\rdistriv$ vectors and classifies the committee with a smaller value at index $1$ of the $\rdistriv$ vector as more diverse---note that this means that two committees are classified as equally diverse if and only if the $\rdistriv$ vectors are empty.

  To prove this, we first give, for each $l\in\bcks[r]\setminus\brcs{1}$,
  \begin{align*}
    t_1 ={}& (\rdistriv^{(1)}, \rdistriv^{(2)}, \rdistrividx^{(1)})\in T\\
    t_2 = {}& (\rdistriv^{(3)}, \rdistriv^{(4)}, \rdistrividx^{(2)})\in T
  \end{align*}
  so that
  \begin{enumerate}
    \item $t_1$ and $t_2$ cannot be distinguished by looking only at index $l$ in the $\rdistriv$ vectors, i.e.~$\rdistriv^{(1)}_l = \rdistriv^{(3)}_l$ and $\rdistriv^{(2)}_l = \rdistriv^{(4)}_l$,
    \item and for all elections $\Elec_1$ and $\Elec_2$ and $S_1,S_2\in\Rulebase(\Elec_1), S_3,S_4\in \Rulebase(\Elec_2)$ so that
    \begin{align*}
      \rdistriv^{(1)}={}&\fun{\rdistriv}(\Elec_1, S_1, S_2),\\
      \rdistriv^{(2)}={}&\fun{\rdistriv}(\Elec_1, S_2, S_1),\\
      \rdistriv^{(3)}={}&\fun{\rdistriv}(\Elec_2, S_3, S_4),\\
      \rdistriv^{(4)}={}&\fun{\rdistriv}(\Elec_2, S_4, S_3),
    \end{align*}
    \item and so that $D(\Elec, S_1)>D(\Elec, S_2)$ and $D(\Elec, S_3)<D(\Elec, S_4)$ for an index $D$ satisfying \pref{propty:pfgpmax}.
  \end{enumerate}
  Thus, at least one of $t_1$ and $t_2$ would be misclassified.
  For this, we consider for each $l$
  \begin{align*}
    \rdistriv^{(1)}_1 &= 0 &\rdistriv^{(2)}_1 &= r-3\\
    \forall j\in\brcs{2,\dots, r-1}: \rdistriv^{(1)}_j &= 1 &\rdistriv^{(2)}_j &= 0\\
    \rdistriv^{(1)}_r &= 0 &\rdistriv^{(2)}_r &= 1
  \end{align*}
  with $\forall j \in\bcks[r-1]: \rdistrividx_j = j$, and $\rdistrividx_r = 1+\sum_{j=2}^{r-1} j-1$.
  Thus, $\fun{D}(\Elec_1, S_1) > \fun{D}(\Elec_1, S_2)$.
  \begin{itemize}[nosep]
    \item For $l\in\bcks[r]\setminus\brcs{1,2,r-1}$, let
      \begin{align*}
        \rdistriv^{(3)}_1 &= 1 & \rdistriv^{(4)}_1 &= 0  \\
        \rdistriv^{(3)}_2 &= 0 & \rdistriv^{(4)}_2 &= r-1  \\
        \forall j\in\brcs{3,\dots, r-2}: \rdistriv^{(3)}_j &= 1 & \rdistriv^{(4)}_j &= 0 \\
        \rdistriv^{(3)}_{r-1} &= 3 & \rdistriv^{(4)}_{r-1} &= 0 \\
        \rdistriv^{(3)}_r &= 0 & \rdistriv^{(4)}_r &= 1
      \end{align*}
      with $\forall j \in\bcks[r-1]: \rdistrividx_j = j$ and $\rdistrividx_r = 1+\paren(\sum_{j=3}^{r-2} j-1) + 3\paren(r-2) -\paren(r-1)$.
      Thus, $\fun{D}(\Elec_2, S_4) > \fun{D}(\Elec_2, S_3)$.
      However, $\forall l \in\bcks[r]\setminus\brcs{1,2,r-1,r}: \rdistriv^{(1)}_l = \rdistriv^{(3)}_l = 1 > \rdistriv^{(2)}_l = \rdistriv^{(4)}_l=0$ and $\rdistriv^{(1)}_r = \rdistriv^{(3)}_r = 0 < \rdistriv^{(2)}_r = \rdistriv^{(4)}_r=1$.
    \item For $l=2$:
      \begin{itemize}
        \item If $r=4$: Let
        \begin{align*}
          \rdistriv^{(3)}_1 &= 1 & \rdistriv^{(4)}_1 &= 0  \\
          \rdistriv^{(3)}_2 &= 1 & \rdistriv^{(4)}_2 &= 0  \\
          \rdistriv^{(3)}_3 &= 0 & \rdistriv^{(4)}_3 &= 3  \\
          \rdistriv^{(3)}_4 &= 1 & \rdistriv^{(4)}_4 &= 0
        \end{align*}
        with $\forall j \in\bcks[3]: \rdistrividx_j = j$ and $\rdistrividx_4 = 6$.
        Thus, $\fun{D}(\Elec_2, S_4) > \fun{D}(\Elec_2, S_3)$.
        However, $\rdistriv^{(1)}_2 = \rdistriv^{(3)}_2 = 1 > \rdistriv^{(2)}_2 = \rdistriv^{(4)}_2=0$.
        \item If $r=5$: Let
        \begin{align*}
          \rdistriv^{(3)}_1 &= 1 & \rdistriv^{(4)}_1 &= 0  \\
          \rdistriv^{(3)}_2 &= 1 & \rdistriv^{(4)}_2 &= 0  \\
          \rdistriv^{(3)}_3 &= 0 & \rdistriv^{(4)}_3 &= 2  \\
          \rdistriv^{(3)}_4 &= 0 & \rdistriv^{(4)}_4 &= 1  \\
          \rdistriv^{(3)}_5 &= 1 & \rdistriv^{(4)}_5 &= 0
        \end{align*}
        with $\forall j \in\bcks[4]: \rdistrividx_j = j$ and $\rdistrividx_5 = 7$.
        Thus, $\fun{D}(\Elec_2, S_4) > \fun{D}(\Elec_2, S_3)$.
        However, $\rdistriv^{(1)}_2 = \rdistriv^{(3)}_2 = 1 > \rdistriv^{(2)}_2 = \rdistriv^{(4)}_2=0$.
        \item If $r\geq 6$: Let $R=\brcs{3,\dots, r-1}$ and
        \begin{align*}
          \rdistriv^{(3)}_1 &= r-5 & \rdistriv^{(4)}_1 &= 0  \\
          \rdistriv^{(3)}_2 &= 1 & \rdistriv^{(4)}_2 &= 0  \\
          \forall j\in R: \rdistriv^{(3)}_j &= 0 & \rdistriv^{(4)}_j &= 1  \\
          \rdistriv^{(3)}_r &= 1 & \rdistriv^{(4)}_r &= 0
        \end{align*}
        with $\forall j \in\bcks[r-1]: \rdistrividx_j = j$ and $\rdistrividx_r = \sum_{j=3}^{r-1} j-1$.
        Thus, $\fun{D}(\Elec_2, S_4) > \fun{D}(\Elec_2, S_3)$.
        However, $\rdistriv^{(1)}_2 = \rdistriv^{(3)}_2 = 1 > \rdistriv^{(2)}_2 = \rdistriv^{(4)}_2=0$.
      \end{itemize}
    \item For $l=r-1$:
      \begin{itemize}
        \item If $r=4$: Let
        \begin{align*}
          \rdistriv^{(3)}_1 &= 1 & \rdistriv^{(4)}_1 &= 0 & \rdistrividx_1 &= 1 \\
          \rdistriv^{(3)}_2 &= 0 & \rdistriv^{(4)}_2 &= 3 & \rdistrividx_2 &= 5\\
          \rdistriv^{(3)}_3 &= 1 & \rdistriv^{(4)}_3 &= 0 & \rdistrividx_3 &= 6\\
          \rdistriv^{(3)}_4 &= 1 & \rdistriv^{(4)}_4 &= 0 & \rdistrividx_4 &= 8
        \end{align*}
        Thus, $\fun{D}(\Elec_2, S_4) > \fun{D}(\Elec_2, S_3)$.
        However, $\rdistriv^{(1)}_3 = \rdistriv^{(3)}_3 = 1 > \rdistriv^{(2)}_3 = \rdistriv^{(4)}_3=0$.
        \item If $r\geq 5$: Let $R=\brcs{3,\dots, r-1}$ and
        \begin{align*}
          \rdistriv^{(3)}_1 &= 1 & \rdistriv^{(4)}_1 &= 0  \\
          \rdistriv^{(3)}_2 &= 0 & \rdistriv^{(4)}_2 &= r-3  \\
          \forall j\in R: \rdistriv^{(3)}_j &= 1 & \rdistriv^{(4)}_j &= 0  \\
          \rdistriv^{(3)}_r &= 0 & \rdistriv^{(4)}_r &= 1
        \end{align*}
        with $\forall j \in\bcks[2]: \rdistrividx_j = j$, $\forall j \in\bcks[r-1]\setminus\bcks[2]: \rdistrividx_j = r + j - 4$ and $\rdistrividx_r = -r + 4 + \sum_{j=3}^{r-1} r - 5 + j$.
        Thus, $\fun{D}(\Elec_2, S_4) > \fun{D}(\Elec_2, S_3)$.
        However, $\rdistriv^{(1)}_{r-1} = \rdistriv^{(3)}_{r-1} = 1 > \rdistriv^{(2)}_{r-1} = \rdistriv^{(4)}_{r-1}=0$.
      \end{itemize}
  \end{itemize}
  Thus, a $D$-explainable function $e_D$ for a diversity index $D$ that is \pref[$1$-explainable]{propty:obv} and that fulfills \pref{propty:pfgpmax} uses only at most $\rdistriv^{(1)}_1$ and $\rdistriv^{(2)}_1$ for $t=(\rdistriv^{(1)},\rdistriv^{(2)},\rdistrividx)\in T$.
  As $D$ fulfills \pref{propty:pfgpmax}, it needs to hold that, if $\rdistrividx_1 =1$, the committee described by $\rdistriv^{(1)}$ is more diverse than the committee described by $\rdistriv^{(2)}$ if $\rdistriv^{(1)}_1 < \rdistriv^{(2)}$ and otherwise as less diverse.
  As $e_D$ does not use $\rdistrividx_1$, it needs to hold that
  \begin{gather*}
    e_{D}(\rdistriv^{(1)}, \rdistriv^{(2)}, \rdistrividx)\\
    = \begin{cases}
      \mathit{less} &\text{if }\dim(\rdistrividx)\geq 1 \text{ and }\rdistriv^{(1)}_1 > \rdistriv^{(2)}_1\\
      \mathit{equal}&\text{if } \dim(\rdistrividx) = 0 \\
      \mathit{more} &\text{if } \dim(\rdistrividx)\geq 1 \text{ and }\rdistriv^{(1)}_1 < \rdistriv^{(2)}_1
    \end{cases}
  \\ = e_{\dns}(\rdistriv^{(1)}, \rdistriv^{(2)}, \rdistrividx),
  \end{gather*}
  as long as there is, for each $x\in\Nzero, d\in\N$, a $(\rdistriv^{(1)},\rdistriv^{(2)},\rdistrividx)\in T$ with $\rdistriv^{(1)}_1 = x, \rdistriv^{(2)}_2 = x+d$ and $\rdistrividx_1 = 1$, i.e.~an election with two committees so that one of them has fewer labels occurring zero times.
  Then, $e_{D}$ cannot know for a given argument whether $\rdistrividx_1 = 1$ or not and needs to classify them as if $\rdistrividx_1 = 1$.
  This is actually the case, which we prove next (still for $r\geq 4$), by making a case distinction on how much $d$ and $r-3$ differ:
  \begin{itemize}
    \item If $d=r-3$: Let $R=\brcs{2,\dots, r-1}$ and
    \begin{align*}
      \rdistriv^{(1)}_1 &= x & \rdistriv^{(2)}_1 &= x+d  \\
      \forall j\in R: \rdistriv^{(1)}_j &= 1 & \rdistriv^{(2)}_j &= 0  \\
      \rdistriv^{(1)}_r &= 0 & \rdistriv^{(2)}_r &= 1,
    \end{align*}
    with $\forall j \in\bcks[r-1]: \rdistrividx_j = j$ and $\rdistrividx_r = 1+\sum_{j=2}^{r-1} j-1$.
    \item If $d>r-3$: Let $R=\brcs{3,\dots, r-1}$ and
    \begin{align*}
      \rdistriv^{(1)}_1 &= x & \rdistriv^{(2)}_1 &= x+d  \\
      \rdistriv^{(1)}_2 &= d-r+4 & \rdistriv^{(2)}_2 &= 0 \\
      \forall j\in R: \rdistriv^{(1)}_j &= 1 & \rdistriv^{(2)}_j &= 0  \\
      \rdistriv^{(1)}_r &= 0 & \rdistriv^{(2)}_r &= 1,
    \end{align*}
    with $\forall j \in\bcks[r-1]: \rdistrividx_j = j$ and $\rdistrividx_r = d-r +5 +\sum_{j=3}^{r-1} j-1$.
    \item If $d<r-3$: Let $R=\brcs{4,\dots, r-1}$ and
    \begin{align*}
      \rdistriv^{(1)}_1 &= x & \rdistriv^{(2)}_1 &= x+d  \\
      \rdistriv^{(1)}_2 &= 0 & \rdistriv^{(2)}_2 &= r-3-d \\
      \rdistriv^{(1)}_3 &= 2 & \rdistriv^{(2)}_3 &= 0 \\
      \forall j\in R: \rdistriv^{(1)}_j &= 1 & \rdistriv^{(2)}_j &= 0  \\
      \rdistriv^{(1)}_r &= 0 & \rdistriv^{(2)}_r &= 1,
    \end{align*}
    with $\forall j \in\bcks[2]: \rdistrividx_j = j$, $\forall j \in\bcks[r-1]\setminus\brcs{1,2}: \rdistrividx_j = r + j - d - 4$ and $\rdistrividx_r = 1 - \paren(r-3-d) + 2\paren(r + 3 - d - 5) + \sum_{j=4}^{r-1} \paren(r + j - d - 5)$.
  \end{itemize}
  This concludes the proof for $r\geq 4$.

  With this, only $r=3$ is left to consider.
  In the following, let $D$ be a diversity index that is \pref[$1$-explainable]{propty:obv} and that fulfills \pref{propty:pfgpmax}.
  First, we show for each
  $t = (\rdistriv^{(1)}, \rdistriv^{(2)}, \rdistrividx^{(1)})\in T$
  that
  $\rdistriv^{(1)}_1 < \rdistriv^{(2)}_1 \iff \rdistriv^{(1)}_2 > \rdistriv^{(2)}_2 \iff \rdistriv^{(1)}_3 < \rdistriv^{(2)}_3$ (note that, for each $i\in \bcks[r]$, it holds that $\rdistriv^{(1)}_i \neq \rdistriv^{(2)}_i$ and hence $\rdistriv^{(1)}_i \geq \rdistriv^{(2)}_i \implies \rdistriv^{(1)}_i > \rdistriv^{(2)}_i$ and analogously $\rdistriv^{(2)}_i \geq \rdistriv^{(1)}_i \implies \rdistriv^{(2)}_i > \rdistriv^{(1)}_i$):
  \begin{itemize}[nosep]
    \item $\rdistriv^{(1)}_1 < \rdistriv^{(2)}_1 \Rightarrow \rdistriv^{(1)}_2 > \rdistriv^{(2)}_2$: Suppose $\rdistriv^{(1)}_2 < \rdistriv^{(2)}_2$ holds and let $d_1,d_2\in\N$ so that $\rdistriv^{(2)}_1=\rdistriv^{(1)}_1 + d_1$ and $\rdistriv^{(2)}_2 = \rdistriv^{(1)}_2 + d_2$. It follows from \cref{helper:charac:one} that $\rdistriv^{(1)}_3 = \rdistriv^{(2)}_3+d_1+d_2$, leading to
    \begin{align*}
      {}&\sum_{i=1}^r \rdistrividx_i \paren(\rdistriv^{(2)}_i-\rdistriv^{(1)}_i)
      \\ ={}& \rdistrividx_1 \cdot d_1 + \rdistrividx_2\cdot d_2 - \rdistrividx_3 \cdot (d_1 + d_2)\\
      <{}& \rdistrividx_3 \cdot d_1 + \rdistrividx_3 \cdot d_2 - \rdistrividx_3 \cdot (d_1 + d_2) = 0,
    \end{align*}
    a contradiction to \cref{helper:charac:two}.
    \item $\rdistriv^{(1)}_2 > \rdistriv^{(2)}_2 \Rightarrow \rdistriv^{(1)}_3 < \rdistriv^{(2)}_3$: Suppose $\rdistriv^{(1)}_3 > \rdistriv^{(2)}_3$ holds and let $d_2,d_3\in\N$ so that $\rdistriv^{(1)}_3=\rdistriv^{(2)}_3 + d_3$ and $\rdistriv^{(1)}_2 = \rdistriv^{(2)}_2 + d_2$. It follows from \cref{helper:charac:one} that $\rdistriv^{(2)}_1 = \rdistriv^{(1)}_1+d_3+d_2$, leading to
    \begin{align*}
      {}&\sum_{i=1}^r \rdistrividx_i \paren(\rdistriv^{(2)}_i-\rdistriv^{(1)}_i)
      \\ ={}& \rdistrividx_1 \paren(d_2 +d_3) - \rdistrividx_2\cdot d_2 - \rdistrividx_3 \cdot d_3\\
      <{}& \rdistrividx_1 \paren(d_2 +d_3) - \rdistrividx_1 \cdot d_2 - \rdistrividx_1 \cdot d_3 = 0,
    \end{align*}
    a contradiction to \cref{helper:charac:two}.
    \item $\rdistriv^{(1)}_3 < \rdistriv^{(2)}_3 \Rightarrow \rdistriv^{(1)}_2 > \rdistriv^{(2)}_2$: Suppose $\rdistriv^{(1)}_2 < \rdistriv^{(2)}_2$ holds and let $d_2,d_3\in\N$
    so that $\rdistriv^{(2)}_3=\rdistriv^{(1)}_3 + d_3$ and $\rdistriv^{(2)}_2 = \rdistriv^{(1)}_2 + d_2$. It follows from \cref{helper:charac:one} that $\rdistriv^{(1)}_1 = \rdistriv^{(2)}_1+d_3+d_2$, leading to
    \begin{align*}
      {}&\sum_{i=1}^r \rdistrividx_i \paren(\rdistriv^{(2)}_i-\rdistriv^{(1)}_i)
      \\ ={}& - \rdistrividx_1 \paren(d_2 +d_3) + \rdistrividx_2\cdot d_2 + \rdistrividx_3 \cdot d_3\\
      >{}& - \rdistrividx_1 \paren(d_2 +d_3) + \rdistrividx_1 \cdot d_2 + \rdistrividx_1 \cdot d_3 = 0,
    \end{align*}
    a contradiction to \cref{helper:charac:two}.
    \item $\rdistriv^{(1)}_2 > \rdistriv^{(2)}_2 \Rightarrow \rdistriv^{(1)}_1 < \rdistriv^{(2)}_1$: Suppose $\rdistriv^{(1)}_1 > \rdistriv^{(2)}_1$ holds and let $d_1,d_2\in\N$
    so that $\rdistriv^{(1)}_2=\rdistriv^{(2)}_2 + d_2$ and $\rdistriv^{(1)}_1 = \rdistriv^{(2)}_1 + d_1$. It follows from \cref{helper:charac:one} that $\rdistriv^{(2)}_3 = \rdistriv^{(1)}_3+d_1+d_2$, leading to
    \begin{align*}
      {}&\sum_{i=1}^r \rdistrividx_i \paren(\rdistriv^{(2)}_i-\rdistriv^{(1)}_i)
      \\ ={}& - \rdistrividx_1 \cdot d_1 - \rdistrividx_2\cdot d_2 + \rdistrividx_3 \paren(d_1+d_2)\\
      > {}& - \rdistrividx_3 \cdot d_1 - \rdistrividx_3 \cdot d_2 + \rdistrividx_3 \paren(d_1+d_2) = 0,
    \end{align*}
    a contradiction to \cref{helper:charac:two}.
  \end{itemize}
  
  Next, we make a case distinction on which index from $\rdistriv^{(1)}$ and $\rdistriv^{(2)}$ is used for a $(\rdistriv^{(1)}, \rdistriv^{(2)}, \rdistrividx)\in T$ with $\dim(\rdistrividx)=3$.
  \begin{itemize}
    \item If only index $3$ is used, it needs to hold for the $D$-explainable function $e_D$ that
    \begin{gather*}
      e_D(\rdistriv^{(1)},\rdistriv^{(2)},\rdistrividx)\\
      = \begin{cases}
        \mathit{less} & \text{if } \rdistriv^{(1)}_3 > \rdistriv^{(2)}_3 \\
        \mathit{equal}& \text{if } \dim(\rdistrividx)=0 \\
        \mathit{more} & \text{if } \rdistriv^{(1)}_3 < \rdistriv^{(2)}_3
      \end{cases}
    \end{gather*}
    Assume that is not the case. One possibility is that there is a $(\rdistriv^{(1)},\rdistriv^{(2)},\rdistrividx)\in T$ with $\dim(\rdistrividx)=3$, $\rdistriv^{(1)}_3 > \rdistriv^{(2)}_3$, and $e_D(\rdistriv^{(1)},\rdistriv^{(2)},\rdistrividx)\neq \mathit{less}$.
    Let $x\in\Nzero, d\in \N$ so that $\rdistriv^{(1)}_3 = x+d$ and $\rdistriv^{(2)}_3 =x$.
    However, for
    \begin{align*}
      \rdistriv^{(1)}_1 &= d & \rdistriv^{(2)}_1 &= 0 & \rdistrividx_1 &= 1 \\
      \rdistriv^{(1)}_2 &= 0 & \rdistriv^{(2)}_2 &= 2d & \rdistrividx_2 &= 2\\
      \rdistriv^{(1)}_3 &= x+d & \rdistriv^{(2)}_3 &= x & \rdistrividx_3 &= 3,
    \end{align*}
    $\rdistriv^{(1)}$ describes a less diverse committee than $\rdistriv^{(2)}$, as $D$ satisfies \pref{propty:pfgpmax}, a contradiction to $e_D(\rdistriv^{(1)},\rdistriv^{(2)},\rdistrividx)\neq \mathit{less}$.
    The same $\rdistriv^{(1)},\rdistriv^{(2)}$, and $\rdistrividx$ can be used to get a contradiction for the case that there is a $(\rdistriv^{(1)},\rdistriv^{(2)},\rdistrividx)\in T$ with $\dim(\rdistrividx)=3$, $\rdistriv^{(1)}_3 < \rdistriv^{(2)}_3$, and $e_D(\rdistriv^{(1)},\rdistriv^{(2)},\rdistrividx)\neq \mathit{more}$.

    Thus, the above definition of $e_D$ is correct and, in addition, equivalent to $e_{\dns}$ as $\rdistriv^{(1)}_1 < \rdistriv^{(2)}_1 \iff \rdistriv^{(1)}_3 < \rdistriv^{(2)}_3$.
    
    \item If only index $2$ is used, it needs to hold for the $D$-explainable function $e_D$ that
    \begin{gather*}
      e_D(\rdistriv^{(1)},\rdistriv^{(2)},\rdistrividx)\\
      = \begin{cases}
        \mathit{less} & \text{if } \rdistriv^{(1)}_2 < \rdistriv^{(2)}_2 \\
        \mathit{equal}& \text{if } \dim(\rdistrividx)=0 \\
        \mathit{more} & \text{if } \rdistriv^{(1)}_2 > \rdistriv^{(2)}_2
      \end{cases}
    \end{gather*}
    Assume that is not the case. One possibility is that there is a $(\rdistriv^{(1)},\rdistriv^{(2)},\rdistrividx)\in T$ with $\dim(\rdistrividx)=3$, $\rdistriv^{(1)}_2 < \rdistriv^{(2)}_2$, and $e_D(\rdistriv^{(1)},\rdistriv^{(2)},\rdistrividx)\neq \mathit{less}$.
    Let $x\in\Nzero, d\in \N$ so that $\rdistriv^{(2)}_2 = x+d$ and $\rdistriv^{(1)}_2 =x$.
    Note that $d\geq 2$: Assume $d=1$ is possible. It follows from $\rdistriv^{(1)}_1 < \rdistriv^{(2)}_1 \iff \rdistriv^{(1)}_2 > \rdistriv^{(2)}_2 \iff \rdistriv^{(1)}_3 < \rdistriv^{(2)}_3$ that $\rdistriv^{(1)}_1 > \rdistriv^{(2)}_1$ and $\rdistriv^{(1)}_3 > \rdistriv^{(2)}_3$.
    Let $d_1,d_3\in\N$ so that $\rdistriv^{(1)}_3 = \rdistriv^{(2)}_3+d_1$ and $\rdistriv^{(1)}_1 = \rdistriv^{(2)}_1+d_1$.
    Because of \cref{helper:charac:one}, it needs to hold that $d_1+d_3 = d=1$, a contradiction to $d_1,d_2\in \N$.

    For
    \begin{align*}
      \rdistriv^{(1)}_1 &= d-1 & \rdistriv^{(2)}_1 &= 0 & \rdistrividx_1 &= 1 \\
      \rdistriv^{(1)}_2 &= x & \rdistriv^{(2)}_2 &= x+d& \rdistrividx_2 &= 2\\
      \rdistriv^{(1)}_3 &= 1 & \rdistriv^{(2)}_3 &= 0 & \rdistrividx_3 &= d+1,
    \end{align*}
    $\rdistriv^{(1)}$ describes a less diverse committee than $\rdistriv^{(2)}$, as $D$ satisfies \pref{propty:pfgpmax}, a contradiction to $e_D(\rdistriv^{(1)},\rdistriv^{(2)},\rdistrividx)\neq \mathit{less}$.
    The same $\rdistriv^{(1)},\rdistriv^{(2)}$, and $\rdistrividx$ can be used to get a contradiction for the case that there is a $(\rdistriv^{(1)},\rdistriv^{(2)},\rdistrividx)\in T$ with $\dim(\rdistrividx)=3$, $\rdistriv^{(1)}_2 > \rdistriv^{(2)}_2$, and $e_D(\rdistriv^{(1)},\rdistriv^{(2)},\rdistrividx)\neq \mathit{more}$.

    Thus, the above definition of $e_D$ is correct and, in addition, equivalent to $e_{\dns}$, as $\rdistriv^{(1)}_1 < \rdistriv^{(2)}_1 \iff \rdistriv^{(1)}_2 > \rdistriv^{(2)}_2$.

    \item If only index $1$ is used, it needs to hold for the $D$-explainable function $e_D$ that
    \begin{gather*}
      e_D(\rdistriv^{(1)},\rdistriv^{(2)},\rdistrividx)\\
      = \begin{cases}
        \mathit{less} & \text{if } \rdistriv^{(1)}_1 > \rdistriv^{(2)}_1 \\
        \mathit{equal}& \text{if } \dim(\rdistrividx)=0 \\
        \mathit{more} & \text{if } \rdistriv^{(1)}_1 < \rdistriv^{(2)}_1
      \end{cases},
    \end{gather*}
    Assume that is not the case. One possibility is that there is a $(\rdistriv^{(1)},\rdistriv^{(2)},\rdistrividx)\in T$ with $\dim(\rdistrividx)=3$, $\rdistriv^{(1)}_1 > \rdistriv^{(2)}_1$, and $e_D(\rdistriv^{(1)},\rdistriv^{(2)},\rdistrividx)\neq \mathit{less}$.
    Let $x\in\Nzero, d\in \N$ so that $\rdistriv^{(1)}_1 = x+d$ and $\rdistriv^{(2)}_1 =x$.
    However, for
    \begin{align*}
      \rdistriv^{(1)}_1 &= x+d & \rdistriv^{(2)}_1 &= x    & \rdistrividx_1 &= 1 \\
      \rdistriv^{(1)}_2 &= 0 & \rdistriv^{(2)}_2 &= d+1   & \rdistrividx_2 &= 2\\
      \rdistriv^{(1)}_3 &= 1 & \rdistriv^{(2)}_3 &= 0  & \rdistrividx_3 &= d+2,
    \end{align*}
    $\rdistriv^{(1)}$ describes a less diverse committee than $\rdistriv^{(2)}$, as $D$ satisfies \pref{propty:pfgpmax}, a contradiction to $e_D(\rdistriv^{(1)},\rdistriv^{(2)},\rdistrividx)\neq \mathit{less}$.
    The same $\rdistriv^{(1)},\rdistriv^{(2)}$, and $\rdistrividx$ can be used to get a contradiction for the case that there is a $(\rdistriv^{(1)},\rdistriv^{(2)},\rdistrividx)\in T$ with $\dim(\rdistrividx)=3$, $\rdistriv^{(1)}_1 < \rdistriv^{(2)}_1$, and $e_D(\rdistriv^{(1)},\rdistriv^{(2)},\rdistrividx)\neq \mathit{more}$.

    Thus, the above definition of $e_D$ is correct and, in addition, equivalent to $e_{\dns}$.
  \end{itemize}
}

\section{Incorporating Diversity Indices into Elections}
\label{sec:approaches}
In this section, we propose two 
ways 
to incorporate diversity (indices) into elections.
One way is to maximize the diversity given a minimal satisfaction for each agent:
\begin{definition}[\prob{MAX-$D$-DSAT}]
  \label{def:maxdsa}
  Given an $\Elec=(A,C,U,k,L,\labof)$ and a function $h : A \to \Nzero$, find a committee with maximal diversity with respect to the diversity index $D$ among all committees $S\in\fun{\Rulebase}(\Elec)$ for which $\forall a\in A: 
  \fun{\mathrm{sat}}(\Elec, S, a)
  \geq \fun{h}(a)$.
\end{definition}
On the one hand, this provides a certain amount of freedom in the search for a diverse committee, but on the other hand, it gives the voters the certainty that their satisfaction cannot be worsened arbitrarily.
One possibility for defining $\fun{h}(a)$ is to compute a committee~$S$ with a well-known rule and to define $\fun{h}(a)$ as the satisfaction of agent $a$ with $S$ minus one (which we will consider in our numerical experiments) or to ensure the same minimum satisfaction for each agent.

A different way to incorporate diversity is to maximize the diversity given a scoring function (e.g.~of a scoring-based voting rule) and a minimum committee score:
\begin{definition}[\prob{MAX-$\paren(D, s)$-DSCR}]
  \label{def:maxdss}
  Given an election $\Elec$ and a bound $\beta \in \Nzero$, find a committee with maximal diversity with respect to diversity index~$D$ among all committees $S\in\fun{\Rulebase}(\Elec)$ for which $\fun{s}(\Elec, S)\geq \beta$, where $s$ is a scoring function.
\end{definition}
We will show results for the
following, well-known scoring-based approval voting rules in the following sections:
Multi-Winner Approval Voting (AV), Satisfaction Approval Voting (SAV), Proportional Approval Voting (PAV), and Approval Chamberlin-Courant (CC) (see e.g.~\citep{lackner2023multi} for a definition of these rules).

We will refer to the score of the scoring-based voting rule $\Rule$ as $\mathrm{score}_{\Rule}$, and to the decision variant of \prob{MAX-$D$-DSAT} and \prob{MAX-$\paren(D, s)$-DSCR}, in which the goal is to find a committee with a diversity which is at least a given value $\delta$, as \prob{$D$-DSAT} and \prob{$\paren(D, s)$-DSCR}, respectively.

\subsection{Complexity Results}
\label{sec:compl}
\appendixsection{sec:compl}
\NewDocumentCommand\wC{d()}{\fun{w}( #1 )}
First, we consider the computational complexity of finding a ``diversity optimal'' committee without additional constraints.
For this, we can exploit the fact that each index satisfies \pref[Weak Occurrence Balancing]{propty:balance}.
Therefore, 
the highest diversity is reached when starting with an empty committee and iteratively adding a candidate with a label that, among the labels with unselected candidates,
occurs least often in the current committee:

\begin{observation}[\appref{obs:indicespoly}]
\label{obs:indicespoly}
Given an election $\Elec$ and one of the diversity indices considered, choosing a committee $S\in\fun{\Rulebase}(\Elec)$ with the highest diversity is polynomial-time solvable.
\end{observation}
\appendixproof{obs:indicespoly}{
The following algorithm yields a most diverse committee:
\begin{enumerate}[nosep]
  \item Start with an empty candidate set $C'=\emptyset$.
  \item Pick an $l' \in \argmin_{l\in L'} \fun{n_{l}}(\Elec,C')$ with $L'=\{l\in \bcks[m]: \fun{n_{l}}(\Elec,C\setminus C') > 0\}$ (i.e.~a label is in $L'$ if it is assigned to at least one candidate not chosen yet).
  \item Add a $c\in \setFP(\Elec,C\setminus C', l')$ to $C'$.
  \item If $|C'| = k$, stop. Otherwise, go to step 2.
\end{enumerate}

Clearly, this returns a committee $C'$ satisfying the following property:
$\forall \paren(l, l')\in \bcks[m]\times \bcks[m]: \fun{n_{l'}}(\Elec,C') + 1 <\fun{n_{l}}(\Elec,C') \Rightarrow \fun{n_{l'}}(\Elec,C\setminus C') = 0$
(otherwise it would contradict step 2 and 3).

In the following, let $\distri'=\fun{\distri}(\Elec, C')$.
First, we show that all committees from $\fun{\Rulebase}(\Elec)$ fulfilling the above property have the same $\distri$ vector and therefore the same diversity according to each of the indices at hand.
Assume this is not the case, i.e.~there exists a $C''\in\fun{\Rulebase}(\Elec)$ with $\distri''\ceq \fun{\distri}(\Elec, C'')\neq \distri'$ that fulfills the property.
Thus, there are $l\in \bcks[m]$, $l'\in \bcks[m]\setminus \brcs{l}$ and $d, d'\in \N$ such that $\fun{n_{l}}(\Elec,C') = \fun{n_{l}}(\Elec,C'') + d$ and 
$\fun{n_{l'}}(\Elec,C') + d' = \fun{n_{l'}}(\Elec,C'')$ and, hence, $\fun{n_{l}}(\Elec,C\setminus C'') > 0$.
We make a case distinction as to how much the number of occurrences of $l$ and $l'$ differ in $C'$---note that $\fun{n_{l}}(\Elec,C') > \fun{n_{l'}}(\Elec,C') +1$ is not possible because $C'$ satisfies the property---, each of which leads to a contradiction:
\begin{itemize}[nosep]
  \item $\fun{n_{l}}(\Elec,C') \leq \fun{n_{l'}}(\Elec,C')$:
  As the number of $l$ is smaller in $C''$ than in $C'$, $\fun{n_{l}}(\Elec,C\setminus C'') > 0$ and
  $\fun{n_{l'}}(\Elec,C'')-\fun{n_{l}}(\Elec,C'') > 1$, because
  $\fun{n_{l'}}(\Elec,C'') > \fun{n_{l'}}(\Elec,C') \geq \fun{n_{l}}(\Elec,C') > \fun{n_{l}}(\Elec,C'')$.
  Therefore, $C''$ violates the above property, a contradiction. 
  \item $\fun{n_{l}}(\Elec,C') = \fun{n_{l'}}(\Elec,C') +1$:
  If $d=d'=1$, this pair of labels does not lead to the $\distri$ vectors being different, i.e.~there needs to be a different pair of labels for which the number of occurrences differ between $C'$ and $C''$ as described above (continue the case distinction for a different pair of labels).
  Otherwise, $\fun{n_{l'}}(\Elec,C'')-\fun{n_{l}}(\Elec,C'') > 1$, because
  $\fun{n_{l'}}(\Elec,C'') = \fun{n_{l'}}(\Elec,C')+d' = \fun{n_{l}}(\Elec,C')+d'-1$ and $\fun{n_{l}}(\Elec,C') = \fun{n_{l}}(\Elec,C'') + d$ and $d,d'\geq 1$ and $d+d'>2$.
  Therefore, $C''$ violates the above property, a contradiction. 
\end{itemize}
Thus, each committee satisfying the property has the same diversity according to the indices at hand.
In addition, as these diversity indices satisfy \pref[Weak Occurrence Balancing]{propty:balance}, each committee $C''\in\fun{\Rulebase}(\Elec)$ that does not fulfill the property is at most as diverse as $C'$, making $C'$ one of the most diverse committees.
}

Yet, \prob{$D$-DSAT} is \NP-hard for these indices.
The same holds for \prob{$\paren(D, s)$-DSCR} if finding winning committees is \NP-hard for $\Rule^s$.
We show this result for a broader class of diversity indices satisfying the following, clearly desirable property:
\begin{property}[Uniqueness Optimality]\label{propty:UniqueOptimal}
  A diversity index $D$ satisfies \emph{Uniqueness Optimality} 
  if for all elections with $m\geq k$, it holds that $S\in\fun{\Rulebase}(\Elec)$ has optimal diversity if and only if no two candidates in~$S$ have the same label.
\end{property}
Clearly, each of $\dns, \simps$, and $\shann$ satisfy this property because they satisfy \pref{propty:balance}, and $\rich$ satisfies it because it is defined as $m - \distri(\Elec,\Com)_1$ and $\distri(\Elec,\Com)_1$ gets smaller if labels occur more than once.
For such indices, we have the following:
\begin{observation}
  \label{observation:ApproaGenNPhard}
  If $D$ is a diversity index that satisfies \pref{propty:UniqueOptimal},
  (1) \prob{$\paren(D, s)$-DSCR} is \NP-hard if the decision problem of $\Rule^s$ is \NP-hard and (2) \prob{$D$-DSAT} is \NP-hard.
\end{observation}
\begin{proof}
  (1) The reduction from the \NP-hard decision problem of $\Rule^s$, i.e.~where given an election $\Elec=(A,C,U,k)$, the task is to decide whether there is $\Com\in\fun{\Rulebase}(\Elec)$ with $\fun{s}(\Elec, \Com)\geq s^*$, works as follows:
  Choose the election $\Elec'=(A,C,U,k,L',\labof')$ with $L' = \brcs{l_c: c\in C}$ as the possible labels and $\fun{\labof'}(c) = l_c$
  (i.e.~each candidate has a different label).
  Set $\delta = \fun{D}(\Elec, \Com^*)$,
  where $\Com^*$ is a committee for~$\Elec$ in which $k$ labels occur exactly once, 
  and $\beta = s^*$.
  Thus, each committee of size $k$ for~$\Elec'$ has the diversity $\fun{D}(\Elec, \Com^*)$.
  Hence, $\paren(D, s)$-DSCR has a solution if and only if 
  there is a committee of score at least~$s^*$.\\
  (2) The reduction from the problem of finding a committee in which each agent has a satisfaction of at least one, which is \NP-hard \citep{procaccia2008CChard}, works analogously by giving each candidate a different label, setting $\delta$ to $\fun{D}(\Elec, \Com^*)$ with $\Com^*$ being a committee with $k$ labels, and $\fun{h}(a) = 1$ for all agents~$a\in A$.
\end{proof}
As the decision problems of CC and PAV are \NP-hard \citep{procaccia2008CChard,aziz2015PAChard}, we have that
\prob{$\paren(D, \mathrm{score}_{\mathrm{CC}})$-DSCR} and \prob{$\paren(D, \mathrm{score}_{\mathrm{PAV}})$-DSCR} are \NP-hard, where $D$ is one of the indices considered.

In the remainder of this section, 
we want to focus on separable scoring functions:
\begin{definition}
  A scoring function~$s$ is called \emph{separable} if there is a polynomial-time computable function~$w$ mapping an election and a candidate to $\mathbb{N}$ such that 
  for every election~$\Elec$ and
  for every~$S\subseteq C$,
  we have that $s(\Elec,S)=\sum_{c\in S} \wC(\Elec, c)$.
\end{definition}
Clearly, 
a committee of $\Rule^{s}$ can be computed in polynomial time if $s$ is separable.
The scoring function of
AV is such a separable function with $\wC(\Elec,c) = |\brcs{a \in A: c\in\fun{U}(a)}|\leq \abs{A}$, and the scoring function of SAV can be transformed into a separable function with $\wC(\Elec,c) = \sum_{a \in A: c\in\fun{U}(a)} \paren(\ell/\abs{\fun{U}(a)})$, where $\ell$ is the least common multiple of $\bigcup_{a\in A} \brcs{\abs{\fun{U}(a)}}$, which can be computed in polynomial time using binary representation.
For separable functions,
we have:

\begin{theorem}[\appref{theorem:DssDnf}]
  \label{theorem:DssDnf}
  \prob{MAX-$\paren(D, s)$-DSCR} is in \classP{} if $s$ is a separable function and $D\in\brcs{\rich,\dns}$.
\end{theorem}
Therefore, for $D\in \brcs{\rich,\dns}$, it holds that \prob{MAX-$\paren(D, \mathrm{score}_{\mathrm{AV}})$-DSCR} and \prob{MAX-$\paren(D, \mathrm{score}_{\mathrm{SAV}})$-DSCR} are in \classP{}.%
\appendixproof{theorem:DssDnf}{
  First, we show the result for \textbf{$\dns$}:
  In the following, let $\eta=\funbc{\min}{m, n}$ and
  \begin{align*}
    \muDO(i,\Gamma) \coloneqq {}& \sum_{j=1}^{i} \paren(\eta + 1))^{n+1-j} \cdot \abs{\fun{\sigma_j}(\Gamma)}\\
    \text{with } \fun{\sigma_j}(\Gamma) = {}& \brcs{l\in\bcks[m]: \fun{n_l}(\Elec,\Gamma) \geq j},\\
    \fun{L_j}(\Gamma) \ceq {}& \brcs{l\in \bcks[m]: \fun{n_l}(\Elec,\Gamma)=j}
  \end{align*}
  and therefore $\fun{\dns}(\Elec,K) = \muDO(k, K)$.
  Note that for two committees $S_1$ and $S_2$ it holds, if $\exists j\in \bcks[i]$ with $\fun{\distri_{j}}(S_1)\neq \fun{\distri_{j}}(S_2)$, that
  \begin{equation*}
    \muDO(i,S_1)>\muDO(i,K_2) \iff \fun{\distri_{m'}}(S_1)< \fun{\distri_{m'}}(S_2)
  \end{equation*}
  with $m'=\min\brcs{j\in \bcks[i]: \fun{\distri_{j}}(S_1)\neq \fun{\distri_{j}}(S_2)}$, for an analogous reason to why 
  \begin{gather*}
    e_{\dns}(\rdistriv^{(1)}, \rdistriv^{(2)}, \rdistrividx)\\
    = \begin{cases*} \mathit{less}, \text{ if }\dim(\rdistrividx)\geq 1 \text{ and }\rdistriv^{(1)}_1 > \rdistriv^{(2)}_1\\  \mathit{equal}, \text{ if } \dim(\rdistrividx) = 0 \\
  \mathit{more}, \text{ if } \dim(\rdistrividx)\geq 1 \text{ and }\rdistriv^{(1)}_1 < \rdistriv^{(2)}_1 \end{cases*}
  \end{gather*}
  is a $\dns$-explainable function (see \cref{proof:theorem:explain:bounds}).

  Our polynomial-time algorithm works as follows:
  \begin{enumerate}[nosep]
    \item Start with a committee $K^*$ which the rule $\Rule^s$ determines and which therefore has the highest score, set $K=K^*$ and $j=0$.
    If $\fun{s}(\Elec,K)<\beta$, stop, as there is no committee fulfilling the condition regarding the score.
    \item If $j=k$, return $K$.
    Otherwise, let 
    \begin{equation*}
      I_j = \brcs{i \in \bcks[m]: \fun{n_i}(\Elec,K) = j \land \fun{n_i}(\Elec,C) > j}
    \end{equation*}
    and, for $i\in I_j$, let $w_i = \max_{c \in \setFP(\Elec,C, i)\setminus K} \wC(\Elec,c)$,
    $c_i$ a candidate with label $l_i$ that would contribute the most to score among the candidates with this label that are not part of $K$, i.e.~$c_i \in \brcs{c \in\setFP(\Elec,C, i)\setminus K: \wC(\Elec,c) = w_i}$, $X_s = \brcs{c_i: i \in I_j}$, and $X_e = \brcs{}$.
    \item Pick a $c_a\in \argmax_{c_i \in X_s\setminus X_e} w_i$.
    Let
    \begin{equation*}
      C_p = \brcs{c \in K: \fun{n_i}(\Elec,K) > j+1 \text{ with } l_i = \fun{\labof}(c)}.
    \end{equation*}
    Pick a $c_r \in \argmin_{c\in C_p} \wC(\Elec,c)$, i.e.~a candidate which is currently part of $K$, has a label which occurs more than $j+1$ times in $K$, and which contributes the least to the score among such candidates.
    Let $K'=K\cup\brcs{c_a}\setminus\brcs{c_r}$.
    If $\fun{s}(\Elec,K')<\beta$, set $j=j+1$ and go to step 2 (as we cannot increase the number of occurrences of labels occurring $j$ times (further) without decreasing the number of labels occurring at most $j+1$ times).
    Otherwise, set $K=K'$, $X_s = X_s \setminus \brcs{c_a}$ and $X_e = X_e \cup \brcs{c_a}$.
    If $|X_s|=0$ (there are no labels occurring $j$ times left), set $j=j+1$ and go to step 2.
    Otherwise, start with step 3 again.
\end{enumerate}
We show by induction that, when visiting the second step for the $i$-th time, with $i\in \brcs{2,\dots k+1}$ and $K_i$ being the committee when reaching this step for the $i$-th time, there is
\begin{enumerate}[nosep]
  \item no other committee $K'\in\fun{\Rulebase}(\Elec)$ with $\fun{s}(\Elec,K') \geq\beta$ and $\muDO(i-1,K')>\muDO(i-1,K_i)$ (later referred to as the first condition).
  \item no other committee $K''\in\fun{\Rulebase}(\Elec)$ with $\fun{s}(\Elec,K'')> \fun{s}(\Elec,K_i)$ and $\muDO(i-1,K'')=\muDO(i-1,K_i)$ (later referred to as the second condition).
\end{enumerate}
Note that these conditions imply the following (later referred to as the third condition):
\begin{gather*}
  \forall l\in \bcks[m], l'\in \brcs{\phi\in \bcks[m]: \fun{n_{\phi}}(\Elec,K_i) < \fun{n_{l}}(\Elec,K_i)},\\
  c \in \setFP(\Elec,K_i, l), c'\in \setFP(\Elec,C\setminus K_i, l'):\\
  \wC(\Elec,c')\leq \wC(\Elec,c).
\end{gather*} 
Assume this is not the case, i.e.~there are $l,l',c,c'$ so that $\wC(\Elec,c') > \wC(\Elec,c)$.
Let $K_d = K_i\setminus\brcs{c}\cup\brcs{c'}$.
Therefore, $\fun{s}(\Elec,K_d)> \fun{s}(\Elec,K_i)\geq \beta$.
In addition, if $\fun{n_{l'}}(\Elec,K_i)\geq i-1$ or $\fun{n_{l}}(\Elec,K_i) = \fun{n_{l'}}(\Elec,K_i)+1$, it follows that $\muDO(i-1,K_d)=\muDO(i-1,K_i)$, which contradicts the second condition.
On the other hand, if $\fun{n_{l'}}(\Elec,K_i) < i-1$ and $\fun{n_{l}}(\Elec,K_i) > \fun{n_{l'}}(\Elec,K_i)+1$, it follows that $\muDO(i-1,K_d)>\muDO(i-1,K_i)$, which contradicts the first condition.

We start with $i=2$, i.e.~with visiting the second step for the second time:
If $K_1=K^*$ has either no label occurring $0$ times, or $l_p$ labels occurring $0$ times and all other committees of size $k$ which satisfy the bound $\beta$ have at least $l_p$ labels occurring $0$ times, $\fun{s}(\Elec,K_1)=\fun{s}(\Elec,K_2)\geq\fun{s}(\Elec,K')$ for all other committees $K'$ of size $k$, as $K_1=K^*=K_2$.
Otherwise, let $K'\in\fun{\Rulebase}(\Elec)$ be a committee respecting $\beta$ in which $l_n<l_p$ labels occur $0$ times and, therefore, in which more labels occur at least once. We show that $K_2$ has at most $l_n$ labels occurring $0$ times and, if it has exactly $l_n$ many labels occurring $0$ times, a score which is at least as good.

First, we transform $K'$ so that the score does not decrease and the number of labels occurring $0$ times does not increase.
As long as $\exists l'\in \bcks[m]: \fun{n_{l'}}(\Elec,K') = 0 <\fun{n_{l'}}(\Elec,K_1)$
and therefore $\exists l\in \bcks[m]: \fun{n_{l}}(\Elec,K_1) = 0 < \fun{n_{l}}(\Elec,K')$,
we set $K' = K' \setminus \brcs{c} \cup \brcs{c'}$ with $c'\in \setFP(\Elec,K_1, l')$ and $c\in \setFP(\Elec,K', l)$, which does not reduce the score (as $K_1$ maximizes the score).
Let $l_n$ be the number of labels occurring $0$ times in this updated $K'$.

Pick the $L_a\subseteq \fun{L_0}(K_1)\setminus \fun{L_0}(K')$ with $\abs{L_a}=l_p-l_n$ and for each $l\in L_a$ a candidate $c_a^{\paren(l)} \in \setFP(\Elec,K',l)$ and let $C_a = \bigcup_{l\in L_a}\brcs{c_a^{\paren(l)}}$.
There needs to be a $C_r\subseteq C$ with $|C_r|=|C_a|=l_p-l_n$ such that $\forall c\in C_r: c\in K_1\setminus K'$ and $\fun{n_{\theta}}(\Elec,K_1) - \fun{n_{\theta}}(\Elec,C_r) \geq \fun{n_{\theta}}(\Elec,K') >0$ with $l_{\theta} = \fun{\labof}(c)$, as it needs to hold that
$k=\sum_{i\in \bcks[m]}\fun{n_i}(K_1) = \sum_{i \in \bcks[m]\setminus \fun{L_0}(K_1)}\fun{n_i}(K_1) = \sum_{i\in \bcks[m]\setminus \fun{L_0}(K_1)}\fun{n_i}(K') + \sum_{i\in L_a}\fun{n_i}(K')$ with $\sum_{i\in L_a}\fun{n_i}(K') \geq l_p-l_n$.

It holds that $\fun{s}(\Elec,K' \cup C_r \setminus C_a)\leq \fun{s}(\Elec,K_1)$ and therefore it holds that $\fun{s}(\Elec,K_1\setminus C_r \cup C_a)\geq \fun{s}(\Elec,K')\geq\beta$.
Consequently, the algorithm will repeat the third step at least $l_p-l_n$ times, in each step removing a candidate from $C_r$ or, alternatively, a candidate $c'$ of a label occurring at least $2$ times with a lower $\wC(\Elec,c')$, and in each step adding a candidate from $C_a$ or a candidate $c''$ of a label occurring $0$ times with a larger $\wC(\Elec,c'')$.
Thus, $K_{2}$ has at most $l_n$ labels occurring $0$ times and, if exactly $l_n$ labels occur $0$ times, a score at least as good as the score of $K'$.

Now, the inductive step $i \leadsto i+1$ follows:
Let $l_p$ be the number of labels occurring $i-1$ times in $K_i$.
Note that for a committee $K'\in\fun{\Rulebase}(\Elec)$ with $\muDO(i,K')\geq\muDO(i,K_i)$ and $\fun{\mathrm{score}_{\mathrm{AV}}}(K')\geq \beta$, it has to hold that $\muDO(i-1,K')=\muDO(i-1,K_i)$ and thus $\fun{\sigma_j}(K')=\fun{\sigma_j}(K_i)$ for $j\in\bcks[i-1]$ and $\abs{\fun{L_j}(K')}=\abs{\fun{L_j}(K_i)}$ for $j\in\brcs{0,\dots, i-2}$.

If $l_p=0$ (and therefore the number of labels occurring at least $i$ times is optimal) or all other committees $K'\in\fun{\Rulebase}(\Elec)$ which satisfy the bound $\beta$ and for which $\muDO(i-1,K')=\muDO(i-1,K_i)$ have at least $l_p$ many labels occurring $i-1$ times (and thus at most as many labels occurring at least $i$ times as $K_i$), $K_{i+1}=K_i$ and $\fun{s}(\Elec,K_i)\geq\fun{s}(\Elec,K')$ due to the induction hypothesis.

Otherwise, let $K'\in\fun{\Rulebase}(\Elec)$ be a committee respecting $\beta$ in which $l_n<l_p$ labels occur $i-1$ times and $\muDO(i-1,K')=\muDO(i-1,K_i)$ and thus more labels occur at least $i$ times.
We show that $K_{i+1}$ has at most $l_n$ labels occurring $i-1$ times and, if exactly $l_n$ labels occur $i-1$ times, a score which is at least as good.

First, we transform $K'$ so that $\muDO(i-1,K')$ remains unchanged, the score does not decrease, and the number of labels occurring exactly $i-1$ times does not increase.
First, note that it is not possible that
$\exists l, l' \in \bcks[m]: \fun{n_{l}}(\Elec,K') \leq i-2$, $\fun{n_{l}}(\Elec,K') = \fun{n_{l'}}(\Elec,K_{i}) < \fun{n_{l}}(\Elec,K_i)$,
and $\fun{n_{l'}}(\Elec,K') \geq \fun{n_{l}}(\Elec,K') + 2$:
If this were possible, $\fun{s}(\Elec,K'\setminus \brcs{c'} \cup \brcs{c}) \geq \fun{s}(\Elec,K') \geq \beta$ with $c'\in \setFP(\Elec,K', l')\setminus K_{i}, c\in \setFP(\Elec,K_{i}, l)\setminus K'$ because of the third condition and,
therefore, $\muDO(i-1,K'\setminus \brcs{c'} \cup \brcs{c})>\muDO(i-1,K_i)$, a contradiction to the induction hypothesis.

For $j\in\paren(0,\dots,i-2)$, we do the following:
As long as $\exists l\in \bcks[m]: \fun{n_{l}}(\Elec,K') = j < \fun{n_{l}}(\Elec,K_i)$
and therefore $\exists l'\in \bcks[m]:\fun{n_{l'}}(\Elec,K_i) = j < \fun{n_{l'}}(\Elec,K') = j+1$ (see the previous note),
we set $K' = K' \setminus \brcs{c'} \cup \brcs{c}$ with $c'\in \setFP(\Elec,K', l')\setminus K_i$ and $c\in \setFP(\Elec,K_i, l)\setminus K'$,
which leaves $\muDO(i-1,K')$ unchanged and does not reduce the score because of the third condition.
This results in 
\begin{align*}
 L_d\ceq{}&\brcs{l\in \bcks[m]: \fun{n_l}(\Elec,K_i)=\fun{n_l}(\Elec,K')\leq i-2}\\
  ={}& \brcs{l\in \bcks[m]: \fun{n_l}(\Elec,K_i)\leq i-2 \lor \fun{n_l}(\Elec,K')\leq i-2}.
\end{align*}
Then, for $j=i-1$, we do the following: As long as 
\begin{itemize}[nosep]
  \item $\exists l'\in \bcks[m]:\fun{n_{l'}}(\Elec,K_i) > \fun{n_{l'}}(\Elec,K') = j$
  \item and therefore $\exists l\in \bcks[m]: \fun{n_{l}}(\Elec,K') > \fun{n_{l}}(\Elec,K_i) = j$ (because $\abs{\fun{L_j}(K_i)}>\abs{\fun{L_j}(K')}$)
  \item and $\exists C_a\subseteq K_i\setminus K': |C_a| = \fun{n_{l}}(\Elec,K') -\fun{n_{l}}(\Elec,K_i)$ such that
  $\forall c\in C_a:$
  $\fun{n_{\theta}}(\Elec,K') + \fun{n_{\theta}}(\Elec,C_a) \leq \fun{n_{\theta}}(\Elec,K_i)$ with $l_{\theta} = \fun{\labof}(c)$,
  and $\exists c_{l'} \in C_a: c_{l'} \in \setFP(\Elec,K_i, l')$
  (such a $C_a$ exists, as it holds that
  $\sum_{\phi\in \bcks[m]\setminus L_d\setminus\brcs{l}} \fun{n_{\phi}}(K_i) + \fun{n_l}(\Elec,K_i)
  = \sum_{\phi\in \bcks[m]\setminus L_d\setminus\brcs{l}} \fun{n_{\phi}}(K') + \fun{n_l}(\Elec,K')$ and $\fun{n_l}(\Elec,K_i)<\fun{n_l}(\Elec,K')$),
\end{itemize}
set $K' = K' \setminus C_r \cup C_a$ with $C_r\subseteq \setFP(\Elec,K', l)\setminus K_i$ and $|C_r| = |C_a|$
so that $l'$ occurs more often than $j$ times in $K'$, but now $l$ occurs $j$ times.
This does not reduce the score either because of the third condition, as we decrease the number of occurrences of label $l'$ in $K'$ (by removing candidates with this label that are not part of $K_i$) and as we increase the number of labels which occur more than $j$ times in $K_i$ (by adding candidates with this label that are part of $K_i$ which have a higher score $s$ than the removed candidates because of the third condition).

For the resulting $K'$, let $l_n$ be the number of labels occurring $i-1$ times in this
updated $K'$, for which $\forall l\in \fun{L_{i-1}}(K'): \fun{n_l}(\Elec, K')=\fun{n_l}(\Elec, K_i)$.
Next, take $L_a\subseteq \fun{L_{i-1}}(K_i)\setminus \fun{L_{i-1}}(K')$ for which it holds that $|L_a| = l_p-l_n$
such that $\forall l_a \in L_a: \setFP(\Elec,K', l_a)>i-1$, pick a $c_{l_a}\in \setFP(\Elec,K', l_a) \setminus K_i$ for all $l_a \in L_a$,
and let $C_a = \bigcup_{l_a\in L_a} c_{l_a}$.
Additionally, there needs to be a $C_r\subseteq K_i\setminus K'$ with $|C_r|=|C_a|$ such that $\forall c\in C_r:
\fun{n_{\theta}}(\Elec,K_i) > \fun{n_{\theta}}(\Elec,K') > i-1 \land \fun{n_{\theta}}(\Elec,K_i) - \fun{n_{\theta}}(\Elec,C_r) \geq \fun{n_{\theta}}(\Elec,K')$ with $l_{\theta}=\fun{\labof}(c)$:
Such a $C_r$ exists, as it holds, with $L^{\geq}_{i} = \bcks[m]\setminus L_d\setminus \fun{L_{i-1}}(K_i)$ as the indices of labels which occur at least $i$ times in $K_i$, that
\begin{align*}
  &\sum_{l\in L^{\geq}_{i}} \fun{n_l}(\Elec,K_i)
  +\sum_{\fun{L_{i-1}}(K_i)} \fun{n_l}(\Elec,K_i)\\
  &= \sum_{l\in L^{\geq}_{i}} \fun{n_l}(\Elec,K') +\sum_{\fun{L_{i-1}}(K_i)} \fun{n_l}(\Elec,K')\\
  \iff& \sum_{l\in L^{\geq}_{i}} \fun{n_l}(\Elec,K_i)\\
  &= \sum_{l\in L^{\geq}_{i}} \fun{n_l}(\Elec,K') +\sum_{\fun{L_{i-1}}(K_i)} \fun{n_l}(\Elec,K') - l_p\paren(i-1)\\
  &\geq \sum_{l\in L^{\geq}_{i}} \fun{n_l}(\Elec,K') + l_p-l_n.
\end{align*}
The last inequality holds because
\begin{gather*}
  \sum_{l\in \fun{L_{i-1}}(K_i)} \fun{n_l}(\Elec,K') \geq l_n\paren(i-1) + \paren(l_p-l_n) \text{ and}\\
  {}l_n\paren(i-1) + \paren(l_p-l_n)i  - l_p\paren(i-1)\\
  ={} \paren(l_p-l_n)i - \paren(l_p-l_n)(i-1)= l_p-l_n.
\end{gather*}

$\fun{s}(\Elec,K' \cup C_r \setminus C_a)\leq \fun{s}(\Elec,K_i)$ holds (because of the induction hypothesis and the third condition) and hence $\fun{s}(\Elec,K_i\setminus C_r \cup C_a)\geq \fun{s}(\Elec,K')\geq\beta$.
Consequently, the algorithm will repeat the third step at least $l_p-l_n$ times, in each step removing a candidate from $C_r$ or, alternatively, a candidate $c'$ of a label occurring more than $i$ times with a lower $\wC(\Elec,c')$, and in each step adding a candidate from $C_a$ or a candidate $c''$ of a label occurring $i-1$ times with a larger $\wC(\Elec,c'')$.
Thus, $K_{i+1}$ has at most $l_n$ labels occurring $i-1$ times and, if exactly $l_n$ labels occur $i-1$ times, a score at least as good as the score of $K'$.

The result for \textbf{$\rich$} follows directly from an optimal solution of 
$\dns$ being an optimal solution of $\rich$:
As $\dns$ fulfills \pref{propty:pfgpmax}, a solution $S\in\Rulebase(\Elec)$ that is optimal for $\dns$ will always have the smallest possible value for $\distri(\Elec, S)_1$.
Therefore, $m-\distri(\Elec, S)_1 =\rich(\Elec,S)$ will always be as large as possible.
}%
While the algorithm for \cref{theorem:DssDnf} utilizes the lexicographic nature of $\dns$,
we show a result for a class of diversity indices which includes $\simps$ and $\shann$ next:
\begin{theorem}
  \label{theorem:DSSGP}
  \prob{MAX-$\paren(D, s)$-DSCR} is in \classP{} if 
  \begin{itemize}[nosep]
    \item
%     (i)
    $s$ is a separable function and there is an $\alpha \in\mathbb{N}_0$ polynomial in the input size so that, for each $c\in C$, $\wC(\Elec,c)\leq\alpha$,
    \item
%     (ii)
    the index $D$ can be expressed as $\sum_{l\in \bcks[m]}\sum_{i=1}^{n_l}\fun{t}(i)$ and, for all $i \in\bcks[k]$, it holds that $0<\fun{t}(i)$ and $\fun{t}(i)$ can be computed in polynomial time, and $t$ is strictly monotonically decreasing.
  \end{itemize}
\end{theorem}
\begin{proof}
Given an instance $I_D$ of MAX-$\paren(D, s)$-DSCR with~$n$ candidates,
we construct an instance~$I_K$ of the 0-1 Knapsack problem in polynomial time.
We assume that~$\beta \leq k\cdot \alpha$ 
($\beta$ being the lower bound for the score, see \cref{def:maxdss}), 
since there is no solution for~$I_D$ otherwise.
For each candidate~$c_i$, we add an item~$x_i$
with the weight~$w_K(x_i) \ceq n\alpha + 1 - w(c_i)$
and the value~$v(x_i) = t(\pi(c_i)) + \eta$,
where~$\eta\ceq k\cdot t(1)+1$ and~$\pi$ outputs~$c_i$'s position in a descending ordering of the candidates with the same label as $c_i$ based on $w$.
The knapsack's bound is~$B\ceq k(n\alpha + 1) - \beta$.
Let, for a solution $X$ of $I_K$, $S(X) = \{c_i \mid x_i\in X\}$, $v(X)$ and $w_K(X)$ the value and weight of $X$, and, for a solution $S$ of $I_D$, $X(S) = \{x_i: c_i \in S\}$. 
We claim that $I_K$ has a solution~$X$ with value at least~$k\cdot \eta$
if and only if 
$S(X)$ is a solution to~$I_D$
and that $I_D$ is infeasible otherwise:

Let~$X$ be a solution to~$I_K$ with value at least~$k\cdot \eta$.
It follows that~$|X|\geq k$.
Assume that~$|X|>k$,
then~$w_K(X)\geq (k+1)(n\alpha + 1) - w(S(X)) \geq (k+1)(n\alpha + 1) - n\alpha = k(n\alpha+1)+1 > B$,
a contradiction.
Thus, $|X|=k$ and $w_K(X)\leq B\iff k(n\alpha + 1) - w(S(X)) \leq k(n\alpha + 1) - \beta \iff \beta \leq w(S(X))$.
Thus,
$S(X)$ is of size exactly~$k$ and respects the score constraint.
If $X$ is a solution to~$I_K$ with value less than~$k\cdot \eta$,
then~$|X|<k$. 
Since~$X$ is maximal,
there is no size-$k$ solution respecting the capacity constraints
and hence no solution to~$I_D$ that respects the score constraint.
Analogously,
$I_K$ has no solution with value at least~$k\cdot \eta$ if~$I_D$ has none.

\newcommand{\OSK}{\ensuremath{X_K^*}}
\newcommand{\OSD}{\ensuremath{S_D^*}}
\newcommand{\OSB}{\ensuremath{S^*}}
Finally, we show that, for an optimal solution $\OSK$ of $I_K$ with value at least~$k\cdot \eta$, $S(\OSK)$ is an optimal solution of $I_D$.
Note that, if, for a label $l_j$, $n_j$ many items $x_i$ with $\fun{\labof}(c_i) = l_j$ are chosen in $\OSK$, the $n_j$ items which come first in the order $\pi$ are always chosen and thus $v(\OSK)=\fun{D}(\Elec, S(\OSK)) + k\cdot \eta$.
Next, assume that $I_D$ has an optimal solution $\OSD\neq S(\OSK)$ with $\fun{D}(\OSD) > \fun{D}(S(\OSK))$.
Consider the committee $\OSB$ with $\fun{n_l}(\Elec, \OSB) = \fun{n_l}(\Elec, \OSD)$ for $l\in \bcks[m]$ and $c\in \OSB \iff \fun{\pi}(c)\leq \fun{n_j}(\Elec, \OSD)$ with $l_j=\fun{\labof}(c)$.
Thus, $\fun{D}(\Elec,\OSB) = \fun{D}(\Elec,\OSD) = \fun{v}(X(\OSB)) - k \cdot \eta > \fun{D}(\Elec, S(\OSK)) = \fun{v}(\OSK) - k \cdot \eta$, a contradiction.

Therefore, the problem can be solved, e.g.~by the solving the 0-1 Knapsack instance with dynamic programming in $\fun{\mathcal{O}}(nB)=\fun{\mathcal{O}}(n \paren(k(n\alpha + 1) - \beta))$ time.
\end{proof}

$\simps$ fulfills the condition in \cref{theorem:DSSGP} by using $\fun{t}(i)=2k+1-2i$, 
$\shann$ through $\fun{t}(i)=-i\fun{\log}(i)+\paren(i-1)\fun{\log}(i-1)+\fun{\log}(k)+2$ with $\fun{t}(1) =\fun{\log}(k) + 2$.
Thus, we have:

\begin{corollary}[\appref{corl:dssP}]
  \label{corl:dssP}
  \prob{MAX-$\paren(\simps, \mathrm{score}_{\mathrm{AV}})$-DSCR} is in \classP{} and, if considering a computational model in which the logarithm of natural numbers and addition and multiplication including a logarithm can be computed in polynomial time, \prob{MAX-$\paren(\shann, \mathrm{score}_{\mathrm{AV}})$-DSCR} is in \classP.
\end{corollary}
\appendixproof{corl:dssP}{
  For $\mathrm{score}_{\mathrm{AV}}$ it holds that
  $\wC(\Elec,c) = |\brcs{a \in A: c\in\fun{U}(a)}|\leq \abs{A}$, for each $c\in C$.

  For $\simps$, it holds that maximizing $-\sum_{i\in \bcks[m]} p_i^2$ yields the
  same solutions (with a different objective value) as maximizing $\paren(\sum_{i\in \bcks[m]} -n_i^2)+2k^2=\sum_{i\in \bcks[m]} -n_i^2+2\cdot n_i \cdot k$.
  Thus, $\fun{t}(i)=-2i+1+2k$ can be chosen (which is strictly monotonically decreasing),
  as it holds that $1\leq \fun{t}(i)\leq 2k+1$ for $i\in\bcks[k]$ and $\sum_{i=1}^n 2k -2i +1 = 2kn-n^2$.

  For $\shann$, let $M_r = \brcs{i \in \bcks[m]: \fun{p_i}(\Elec,S) > 0}$.
  It holds that maximizing
  \begin{align*}
    &-\sum_{i \in M_r} p_i\cdot \tlog{p_i}\\
    ={}&  -\sum_{i \in M_r} \frac{n_i}{k}\cdot \tlog{\frac{n_i}{k}}\\
    = {}& -\frac{1}{k }\paren(\sum_{i \in M_r} n_i\cdot \paren(\tlog{n_i}-\tlog{k}))\\
    = {}& -\frac{1}{k } \paren( - k \cdot \tlog{k} + \sum_{i \in M_r} n_i\cdot \tlog{n_i})
  \end{align*}
  yields the same results as maximizing $- \sum_{i \in M_r} n_i\cdot \tlog{n_i}$ which, in turn, yields the same results as maximizing
  \begin{align*}
    & k\cdot\paren(\tlog{k} +2) -\sum_{i \in M_r} n_i \cdot \tlog{n_i}\\
    ={}&  -\sum_{i \in M_r} n_i \paren(\tlog{n_i} - \tlog{k} - 2) 
  \end{align*}
  Thus, $\fun{t}(i)=-i\fun{\log}(i)+\paren(i-1)\fun{\log}(i-1)+\fun{\log}(k)+2$ with $\fun{t}(1) =\fun{\log}(k) + 2$ can be chosen,
  which is strictly monotonically decreasing, and $\fun{t}(i) > 0$ for $i\in\bcks[k]$:

  Let $\fun{f}(x) = x\fun{\log}(x)$ with $\fun{f'}(x) = 1+\fun{\log}(x)$.
  According to the mean value theorem, $\exists c \in \paren(i-1, i)$ such that $\fun{f}(i)-\fun{f}(i-1) = \fun{f'}(c) = 1+\fun{\log}(c)$.
  Thus, it holds for all $ i \in \bcks[k]$ that
  \begin{align*}
    \fun{t}(i) ={}& -i\fun{\log}(i)+\paren(i-1)\fun{\log}(i-1)+\fun{\log}(k)+2 \\
    = {}& - 1 - \fun{\log}(c) + \fun{\log}(k)+2 > 0,
  \end{align*}
  as $\fun{\log}(c) \leq \fun{\log}(k)$ for $i \in \bcks[k]$.
}

However, we cannot apply \cref{theorem:DSSGP} for SAV with the previously mentioned approach to transform SAV's score into a separable function,
as it leads to weights that could not be bounded by a value polynomial in the input size.
Therefore, we show the following, which is also applicable to SAV:
\begin{theorem}[\appref{theorem:DSSWP}]
  \label{theorem:DSSWP}
  \prob{$\paren(D, s)$-DSCR} is in \classP{} if
    $s$ is a separable function and
    $D$ has the same properties as those in \cref{theorem:DSSGP} and there is a $\zeta\in\mathbb{N}$ polynomial in the input size such that $0<\fun{t}(i)\leq\zeta$ for all $i \in\bcks[k]$.
\end{theorem}
\appendixproof{theorem:DSSWP}{
\newcommand{\OSK}{\ensuremath{X_K^*}}
\newcommand{\OSD}{\ensuremath{S_D^*}}
\newcommand{\OSB}{\ensuremath{S^*}}
Given an instance $I_D$ of $\paren(D, s)$-DSCR with~$n$ candidates, $\beta$ as the score bound, and $\delta$ as the diversity bound,
we construct an instance~$I_K$ of the 0-1 Knapsack problem in polynomial time.
We assume that $\beta \geq 0$---otherwise, we only need to optimize the diversity, which is in \classP{} for such diversity indices, as the same approach as in the proof of \cref{obs:indicespoly} can be used because $\fun{t}(i)$ is strictly monotonically decreasing.
In addition, we assume that $\delta \geq 0$---otherwise, the problem is in \classP{}, as we only need to find a committee of size $k$ fulfilling the score constraint which, if the problem is feasible, can be achieved by choosing $k$ many candidates with the highest weights.
Furthermore, we assume that $\delta\leq k \zeta$ (which can be checked easily), as the problem is infeasible otherwise.

For each candidate~$c_i$, we add an item~$x_i$
with the weight
\begin{equation*}
  w_K(x_i) \ceq -t(\pi(c_i)) + \eta \text{ with } \eta \ceq k\zeta + \zeta + 1,
\end{equation*}
where $\pi$ outputs~$c_i$'s position in a descending ordering of the candidates with the same label as $c_i$ based on $w$,
and the value
\begin{equation*}
  v(x_i) \ceq w(c_i)+ n \alpha +1.
\end{equation*}
Thus, between two candidates $c_i, c_j$ with the same label of which $c_i$ contributes more to the score---i.e.~$w(c_i)>w(c_j)$ and $\pi(x_i) < \pi(x_j)$---, $v(x_i)>v(x_j)$, $t(\pi(x_i))>t(\pi(x_j))$, and hence $w_{K}(x_i)<w_{K}(x_j)$.
Therefore, replacing an item $x_j$ in a solution to $I_K$ by an item $x_i$ with the same label but with a smaller value of $\pi$ will lead to a solution to $I_K$ with a higher value and a lower weight.
Furthermore, we set the knapsack's bound to
\begin{equation*}
  B\ceq -\delta + k \eta.
\end{equation*}
Let, for a solution $X$ of $I_K$, $S(X) = \{c_i \mid x_i\in X\}$, $v(X)$ and $w_K(X)$ be the value and weight of $X$, and, for a solution $S$ of $I_D$, $X(S) = \{x_i: c_i \in S\}$.
We claim the following:
\begin{enumerate}[nosep]
  \item If $X$ is a solution to $I_K$ with value at least $u_v\ceq \beta + k\paren(n \alpha +1)$, then $S(X)$ is a solution to $I_D$.
  \item If $S$ is a solution to $I_D$, there is a (possibly different) solution $S'$ to $I_D$ such that $D(S) = D(S')$ and $X(S')$ is a solution to $I_K$ with value at least $u_v$.
\end{enumerate}
1. Let~$X$ be a solution to~$I_K$ with value at least~$u_v$.
It follows that~$|X|\geq k$, otherwise
\begin{align*}
  v(X) \leq {}& \paren(k-1)\paren(n\alpha + 1 +\alpha)\\
  ={}& k\paren(n\alpha +1 )+k\alpha -n\alpha - 1 -\alpha\\
  \mathrel{\overset{k\leq n}{<}} {}& k\paren(n\alpha +1)\leq u_v,
\end{align*}
a contradiction.
Next, assume that~$|X| = k + i>k$,
then
\begin{align*}
  w_K(X) \geq{}& \paren(k+1)\paren(-\zeta + \eta) =  \paren(k+1)\paren(k \zeta + 1)\\
  ={}& k^2\zeta + k + k\zeta + 1\\
  >{}& k^2\zeta + k\zeta + k = k\eta \geq B,
\end{align*}
a contradiction.
Thus, $|X|=k$.
In addition, $-\delta + k\eta = B \geq w_K(X) \geq k\eta - D(S(X)) \iff \delta \leq D(S(X))$ and $v(X) = k\paren(n\alpha + 1) + w(S(X)) \geq u_v = \beta + k\paren(n \alpha +1) \iff w(S(X)) \geq \beta$.
Thus, $S(X)$ fulfills the score and diversity constraints and is therefore a solution to $I_D$.

2. Let~$\OSD$ be a solution to~$I_D$.
Consider the committee $\OSB$ with $\fun{n_l}(\Elec, \OSB) = \fun{n_l}(\Elec, \OSD)$ for $l\in \bcks[m]$ and $c\in \OSB \iff \fun{\pi}(c)\leq \fun{n_j}(\Elec, \OSD)$ with $l_j=\fun{\labof}(c)$.
Thus,
\begin{gather*}
  \paren(\fun{D}(\Elec,\OSB) = -\fun{w}(X(\OSB)) + k \eta = \fun{D}(\Elec,\OSD) \geq \delta)\\
  \Rightarrow \paren(w(X(\OSB)) \leq -\delta + k\eta),\\
  \paren(\fun{w}(\OSB) \geq \fun{w}(\OSD)\geq \beta)\\
  \Rightarrow (v(X(\OSB)) = \fun{w}(\OSB) + k\paren(n\alpha + 1)\\
  \geq \beta + k\paren(n\alpha + 1) = u_v).
\end{gather*}
Therefore, $X(\OSB)$ is a solution to $I_K$.

Thus, the problem can be solved by the solving the 0-1 Knapsack instance with dynamic programming in $\fun{\mathcal{O}}(nB)=\fun{\mathcal{O}}(n \paren(k \eta-\delta))$ time.
}
The proof of \cref{theorem:DSSWP} is similar in nature to the proof of \cref{theorem:DSSGP} by constructing an instance of the 0-1 knapsack problem, but the diversity constraint is expressed with the help of the weights and the score constraint with the help of the values of the items.
While $\shann$ does not fulfill the imposed conditions of \cref{theorem:DSSWP}%
---leaving the question open whether \prob{$\paren(\shann, \mathrm{score}_{\mathrm{SAV}})$-DSCR} is in \classP{}---%
the previously mentioned choice of $\fun{t}(i)$ for $\simps$ fulfills them. Therefore:
\begin{corollary}
  \prob{$\paren(\simps, \mathrm{score}_{\mathrm{SAV}})$-DSCR} is in \classP{}.
\end{corollary}

\subsection{Experiments}
\label{sec:exps}
\appendixsection{sec:exps}
\NewDocumentCommand{\scrRule}{m}{\Rule_{\mathrm{scr}}^{#1}}
\NewDocumentCommand{\scrRuleX}{mm}{\text{#1}_{\mathrm{scr}}^{#2}}
\NewDocumentCommand{\satRule}{m}{#1_{\mathrm{sat}}^{-1}}

To evaluate the problems,
we use datasets with approval preferences from Pabulib \citep{ijcai2023p297}---a collection of participatory budgeting data, a scenario in which incorporating diversity may be desirable---, in which categories (e.g.~urban greenery) and/or targets (e.g.~adults) are assigned to the candidates:
For each such instance, we create up to three new instances of our model by assigning to a candidate as the label (1) the categories, (2) the targets, or (3) the union of the categories and targets (e.g.~$\brcs{\text{urban greenery}, \text{adults}}$ as a label and different sets forming different labels).
We also transformed two datasets \citep{LaSt08,Laslier_2004} with approval preferences from PrefLib \citep{mattei2017apreflib} about the French presidential election in 2002, consisting of seven instances overall, by assigning the combination of gender and political leaning as the label to each candidate.
The dimensions of the experimental data can be seen in \cref{fig:dims}.
We removed instances with $\abs{C}=m$ because each committee leads to the optimal diversity for them, as each diversity index considered satisfies \pref{propty:UniqueOptimal}.
In the following, we show results for $k=10$.
For this, we discarded instances with $\abs{C}\leq k$, which results in $687$ instances.
We conducted the same experiments for $k\in\brcs{6, 8}$ as well, shown in \cref{app:sec:exps}, which support the main takeaways presented here.

The code for creating the instances and running the experiments is published in \citep{Our_Repo}.

\begin{figure}\centering%
  \makebox{\includegraphics{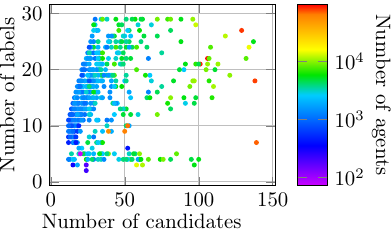}}%
  %\quad%
  \caption{The dimensions of the experimental data, where the color of each point represents the average number of agents of all instances with the given number of labels and candidates.}
  \label{fig:dims}
\end{figure}

\mypar{Experimental Setup.}
To investigate \prob{MAX-$\paren(D, s)$-DSCR}, i.e.~the influence of weakening the score constraint on the diversity reached,
we consider $\mathrm{score}_{\mathrm{AV}}$ and $\mathrm{score}_{\mathrm{SAV}}$ (for which solving \prob{MAX-$\paren(D, s)$-DSCR} is in \classP{}). For a given diversity index, let $\scrRule{p}$ be the rule returning the committees with the highest diversity among the committees reaching at least $p\%$ of the highest value of $\mathrm{score}_{\Rule}$.

To examine \prob{MAX-$D$-DSAT}, i.e.~the influence of the satisfaction constraints, we additionally consider CC, 
PAV, the Method of Equal Shares (Rule X), and Phragmén's sequential rule (seq-Phragmén) (see \citep{lackner2023multi} for definitions).
For each rule $\Rule$ considered, we first compute one committee $S$ using the Python library \emph{abcvoting} \citep{joss-abcvoting} with default parameters and refer to the rule returning this committee as~$\Rule$.
Based on the satisfactions of the agents with $S$,
we look at the rule that returns the committees
with the highest diversity reachable
when the satisfaction of each agent can be decreased by at most one,
which we denote as $\satRule{\Rule}$.

We also compute all winning committees for these rules with \emph{abcvoting} to investigate whether the diversity index could serve as a tiebreaker between them.

As the results for SAV,
PAV, Rule X, and seq-Phragmén
are very similar to the results for AV,
we focus on AV and CC here and show the results for the other rules
in \cref{app:sec:exps}.

\mypar{Experimental Results.}
\begin{figure*}[ht]
  \centering%
  \includegraphics{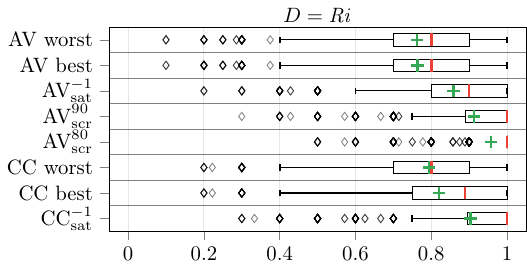}
  \includegraphics{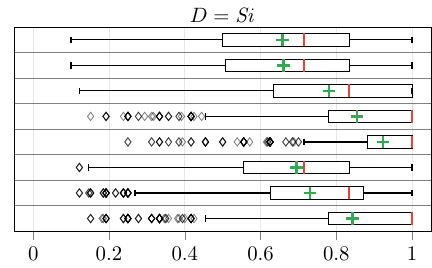}

  \includegraphics{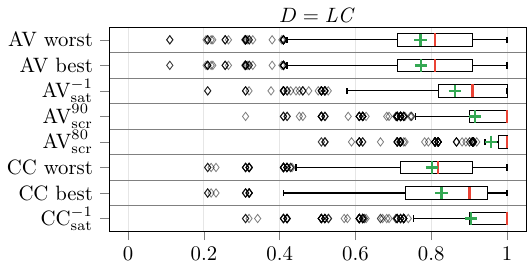}
  \includegraphics{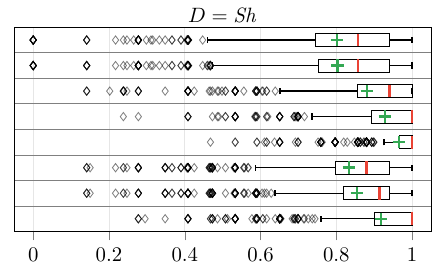}
  \caption{The proportion of the optimal diversity reached on the experimental data when using the specified diversity index $D$.
  ``$\Rule$ best'' (``$\Rule$ worst'') refers to the rule choosing the committees with the highest (lowest) diversity among the winning committees of $\Rule$.
  The red line indicates the median, the green cross the mean.}
  \label{fig:exps}
\end{figure*}%
We want to highlight the following observations, based on \cref{fig:exps}:
On average, both AV and CC (without score or satisfaction constraints) achieve roughly $70\%$ of the optimal diversity when measured with $\simps$ and
$80\%$ when using the other indices.
Thus, the diversity can indeed be improved.
Choosing the winning committee with the highest diversity
rarely makes a difference for AV,
which has multiple winning committees for only around $5\%$ of the instances.
In contrast, the diversity differs between the winning committees of CC---which has multiple winning committees for around $30\%$ of the instances---for around $20\%$ of the instances.
Furthermore, CC achieves, even when choosing the winning committee with the smallest diversity, a slightly higher proportion of the optimal diversity than AV on average.
Similarly, $\satRule{\text{CC}}$ also performs better than $\satRule{\text{AV}}$ on average: $\satRule{\text{CC}}$ reaches a higher diversity for around $38$--$44\%$ and at least the same diversity for around $83$--$87\%$ of the instances (depending on the index).

For AV, we also compare $\satRule{\text{AV}}$ and $\scrRuleX{AV}{90}$: On average, $\scrRuleX{AV}{90}$ reaches $4$--$7\%$ more of the optimal diversity than $\satRule{\text{AV}}$.
Consequently, $\scrRuleX{AV}{90}$ can also lead to a noticeable change compared to AV: The percentage of the optimal diversity achieved on average increases by $12$--$19$.
It is also interesting that the gain in diversity is larger overall for $\scrRuleX{AV}{90}$ compared to AV than for $\scrRuleX{AV}{80}$ compared to~$\scrRuleX{AV}{90}$.
The same applies for the gain from $\scrRuleX{AV}{90}$ to $\scrRuleX{AV}{80}$ compared to that from $\scrRuleX{AV}{80}$ to $\scrRuleX{AV}{70}$ (see \cref{app:sec:exps} for the plots),
which suggests that there are diminishing returns when weakening the score constraint.
In addition, for each diversity index, there are instances for which even allowing a score reduction of $50\%$ does not lead to the optimal diversity.

When comparing the four diversity indices, it seems most challenging to achieve the optimal diversity when using $\simps$ for each rule, as visualized in \cref{fig:exps}.
\toappendix{
The dimension of the experimental data when using $k=8$ and $k=6$ can be seen in \cref{fig:dims2}.
Note, that the number of instances increases for smaller $k$ ($729$ instances for $k=8$, $773$ for $k=6$), because we discard instances with at most $k$ many candidates.

For each diversity index considered, the proportion of the optimal diversity reached over all experimental data for the different rules (including the rules that represent our approaches of incorporating diversity, i.e.~$\satRule{\Rule}$ and $\scrRule{p}$) are visualized using box plots in \crefrange{fig:richness10}{fig:n26}, which also include results for $k\in\brcs{6,8,10}$ and seq-Phragmén, Rule X, PAV, and SAV.

\cref{tab:exp:ndo} shows the average proportion of the optimal diversity that a rule reaches for each $k\in\brcs{6,8,10}$, diversity index, and rule considered.
Similarly, the number of instances for which the rule achieves the optimal diversity can be seen in \cref{tab:exp:optim:ndo}.
These results support the qualitative results mentioned in the main body of the paper.
Here, we want to highlight the following, additional observations based on these tables and plots:

When decreasing $k$, it holds that, for the $k$ investigated and for each rule and diversity index considered, the proportion of instances for which the optimal diversity is reached increases, with the only exceptions occurring when looking at $\scrRuleX{AV}{40}, \scrRuleX{AV}{30}, \scrRuleX{SAV}{50}, \scrRuleX{AV}{40}$, and $\scrRuleX{AV}{30}$ (i.e.~when the scoring constraints are relatively weak), for which the proportions stay the same and the optimal diversity is already reached for at least $97\%$ of the instances with $k=10$.
The average percentage of the optimal diversity reached also increases in most cases
when decreasing $k$ and $\satRule{\Rule}$ is considered, with only a few exceptions when using $\shann$.
However, there are far more exceptions in which this percentage stays the same when looking at $\scrRule{p}$.

In contrast, the number of instances for which the winning committees of CC do not all reach the same diversity (so that maximizing the diversity could be used as a tiebreaker) decreases with decreasing $k$.
With decreasing $k$, the proportion of the instances for which $\satRule{\text{CC}}$ achieves better results than $\satRule{\text{AV}}$ decreases as well, but the proportion of the instances for which $\satRule{\text{CC}}$ achieves at least the same diversity increases.

Lastly, we want to highlight that the optimal diversity is reached for at least $73\%$ of the data when reducing the score to be achieved by $20\%$ of the optimal score (i.e.~$\scrRule{80}$) for each combination of diversity index, $k$, and rule considered,
and for at least $84\%$ of the data when allowing a reduction of the score by $30\%$. 

\begin{figure*}
  \centering%
  \begin{subfigure}[t]{.47\textwidth}
		\centering
		\includegraphics{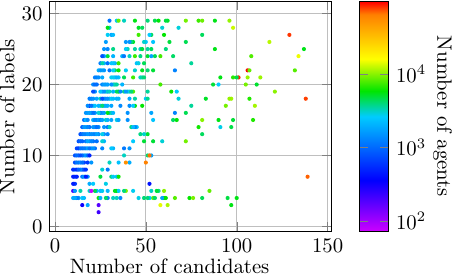}
		\caption{$k=8$.}
	\end{subfigure}\quad
	\begin{subfigure}[t]{.47\textwidth}
		\centering
		\includegraphics{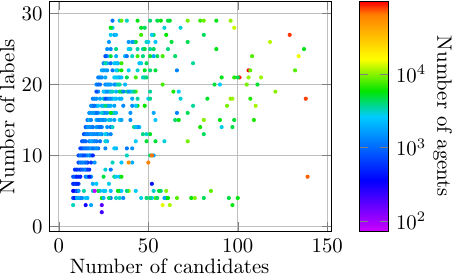}
		\caption{$k=6$.}
	\end{subfigure}%
  \caption{The dimensions of the experimental data, where the color of each point represents the average number of agents of all instances with the given number of labels and candidates.}
  \label{fig:dims2}
\end{figure*}

\begin{figure*}\centering%
  \makebox{\includegraphics{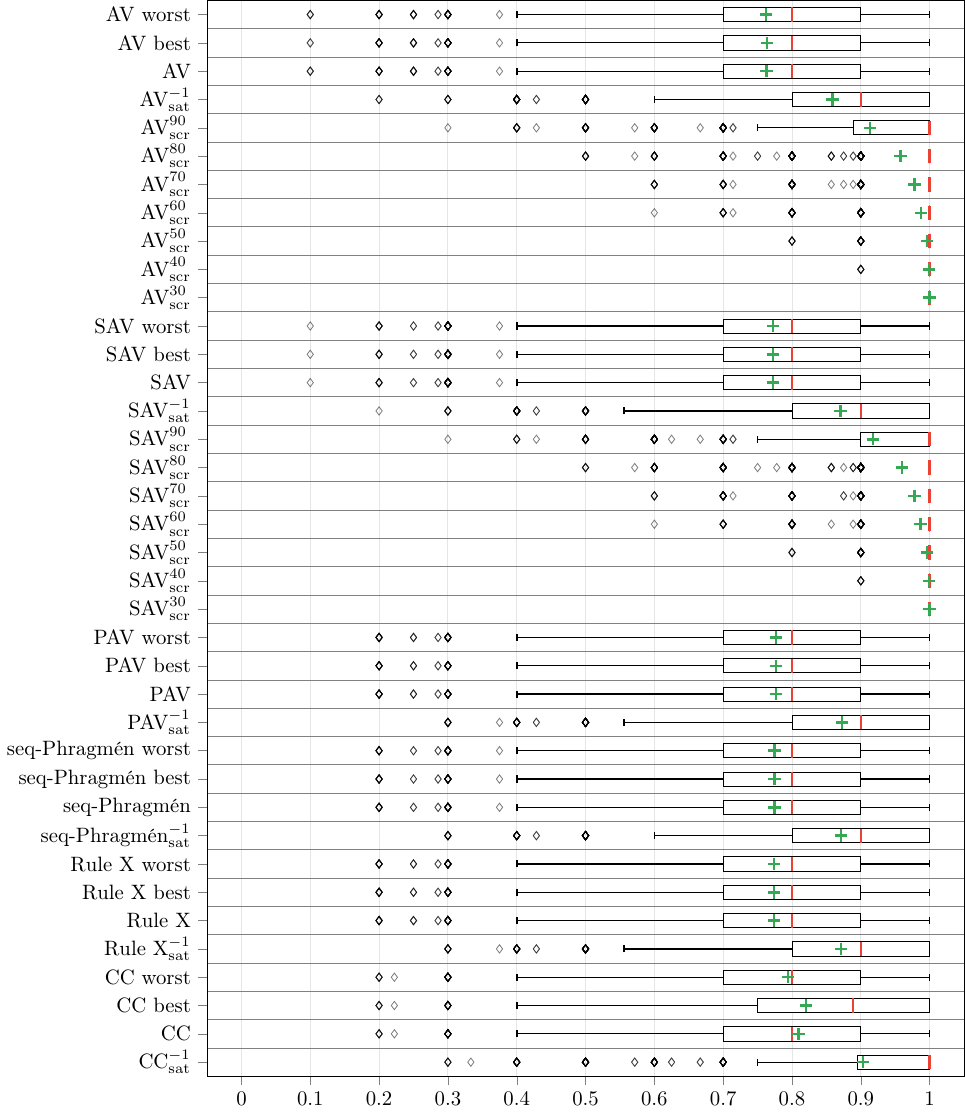}}%
  %\quad%
  \caption{\boldmath The proportion of the optimal diversity reached over all experimental data with $k=10$ when using $\rich$ for the different rules and approaches we consider. ``$\Rule$ best'' (``$\Rule$ worst'') refers to the rule that chooses the committees with the highest (lowest) diversities among the winning committees of $\Rule$. If only $\Rule$ is written, the diversity of the winning committee that \emph{abcvoting} returns for $\Rule$ is considered. The red line indicates the median, the green cross the mean.}
  \label{fig:richness10}
\end{figure*}

\begin{figure*}\centering%
  \makebox{\includegraphics{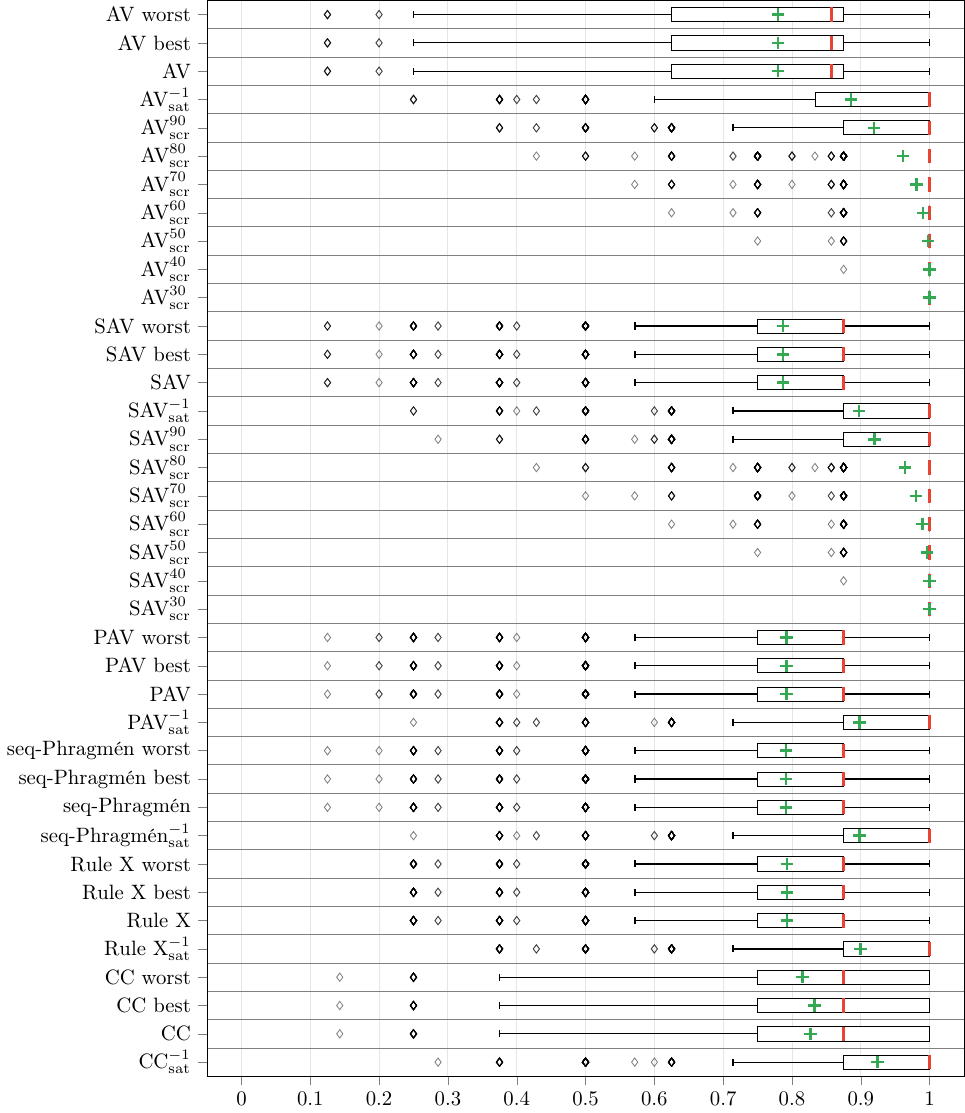}}%
  %\quad%
  \caption{\boldmath The proportion of the optimal diversity reached over all experimental data with $k=8$ when using $\rich$ for the different rules and approaches we consider. ``$\Rule$ best'' (``$\Rule$ worst'') refers to the rule that chooses the committees with the highest (lowest) diversities among the winning committees of $\Rule$. If only $\Rule$ is written, the diversity of the winning committee that \emph{abcvoting} returns for $\Rule$ is considered. The red line indicates the median, the green cross the mean.}
  \label{fig:richness8}
\end{figure*}

\begin{figure*}\centering%
  \makebox{\includegraphics{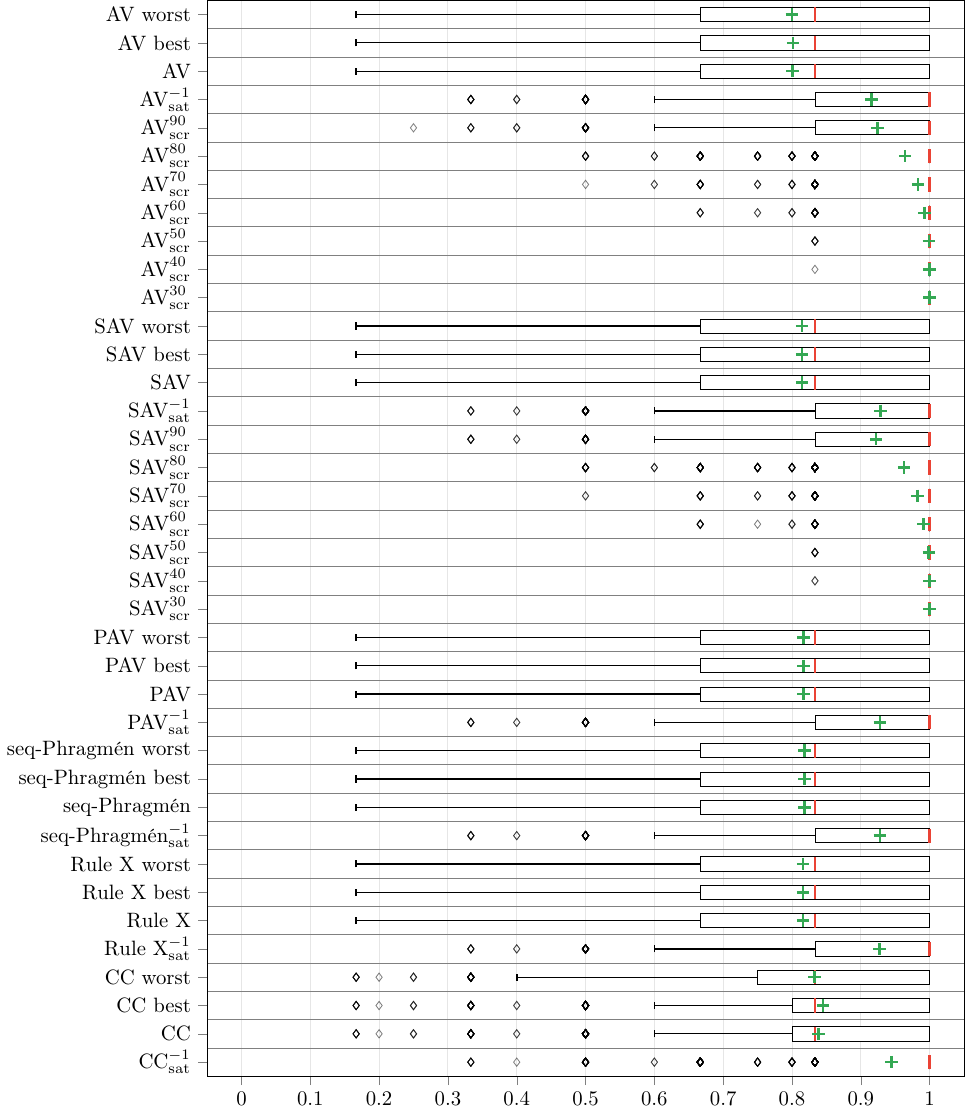}}%
  %\quad%
  \caption{\boldmath The proportion of the optimal diversity reached over all experimental data with $k=6$ when using $\rich$ for the different rules and approaches we consider. ``$\Rule$ best'' (``$\Rule$ worst'') refers to the rule that chooses the committees with the highest (lowest) diversities among the winning committees of $\Rule$. If only $\Rule$ is written, the diversity of the winning committee that \emph{abcvoting} returns for $\Rule$ is considered. The red line indicates the median, the green cross the mean.}
  \label{fig:richness6}
\end{figure*}

\begin{figure*}\centering%
  \makebox{\includegraphics{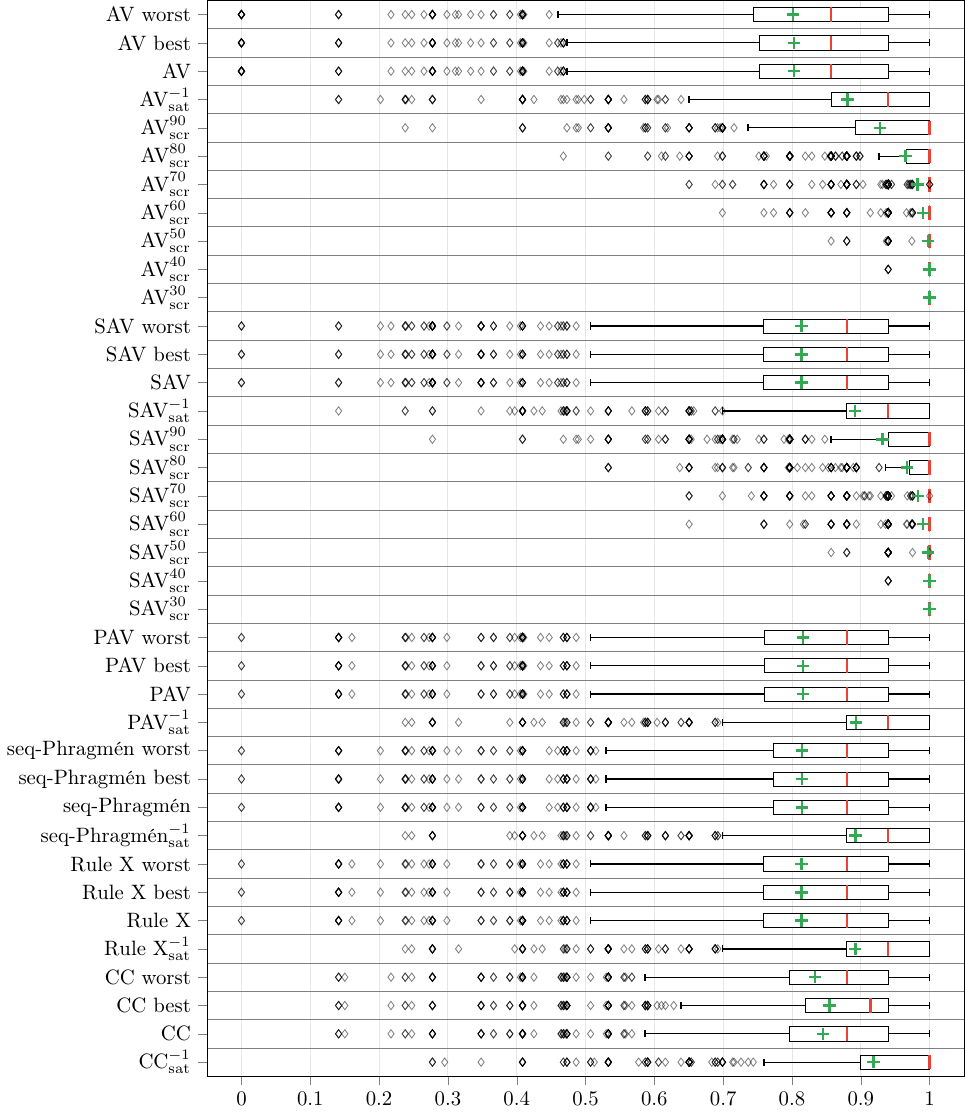}}%
  %\quad%
  \caption{\boldmath The proportion of the optimal diversity reached over all experimental data with $k=10$ when using $\shann$ for the different rules and approaches we consider. ``$\Rule$ best'' (``$\Rule$ worst'') refers to the rule that chooses the committees with the highest (lowest) diversities among the winning committees of $\Rule$. If only $\Rule$ is written, the diversity of the winning committee that \emph{abcvoting} returns for $\Rule$ is considered. The red line indicates the median, the green cross the mean.}
  \label{fig:shann10}
\end{figure*}

\begin{figure*}\centering%
  \makebox{\includegraphics{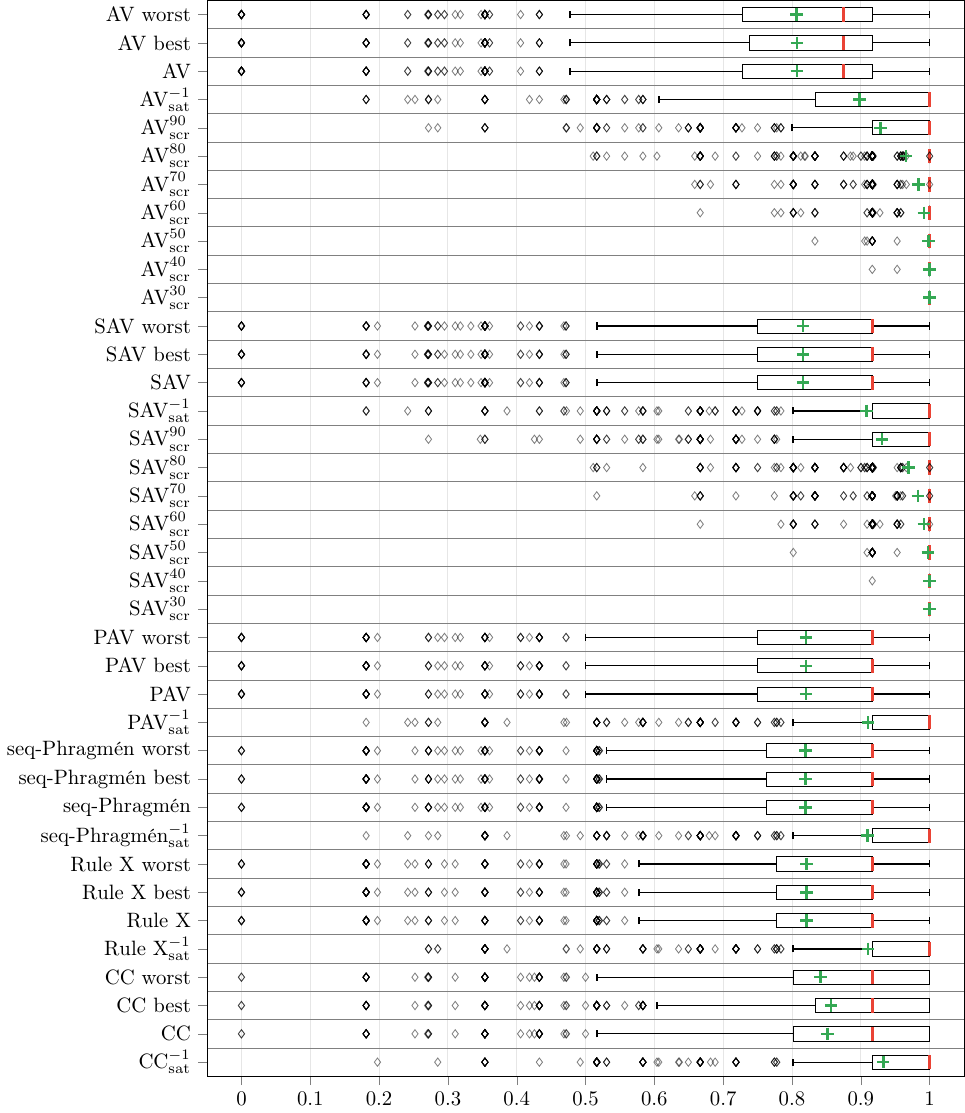}}%
  %\quad%
  \caption{\boldmath The proportion of the optimal diversity reached over all experimental data with $k=8$ when using $\shann$ for the different rules and approaches we consider. ``$\Rule$ best'' (``$\Rule$ worst'') refers to the rule that chooses the committees with the highest (lowest) diversities among the winning committees of $\Rule$. If only $\Rule$ is written, the diversity of the winning committee that \emph{abcvoting} returns for $\Rule$ is considered. The red line indicates the median, the green cross the mean.}
  \label{fig:shann8}
\end{figure*}

\begin{figure*}\centering%
  \makebox{\includegraphics{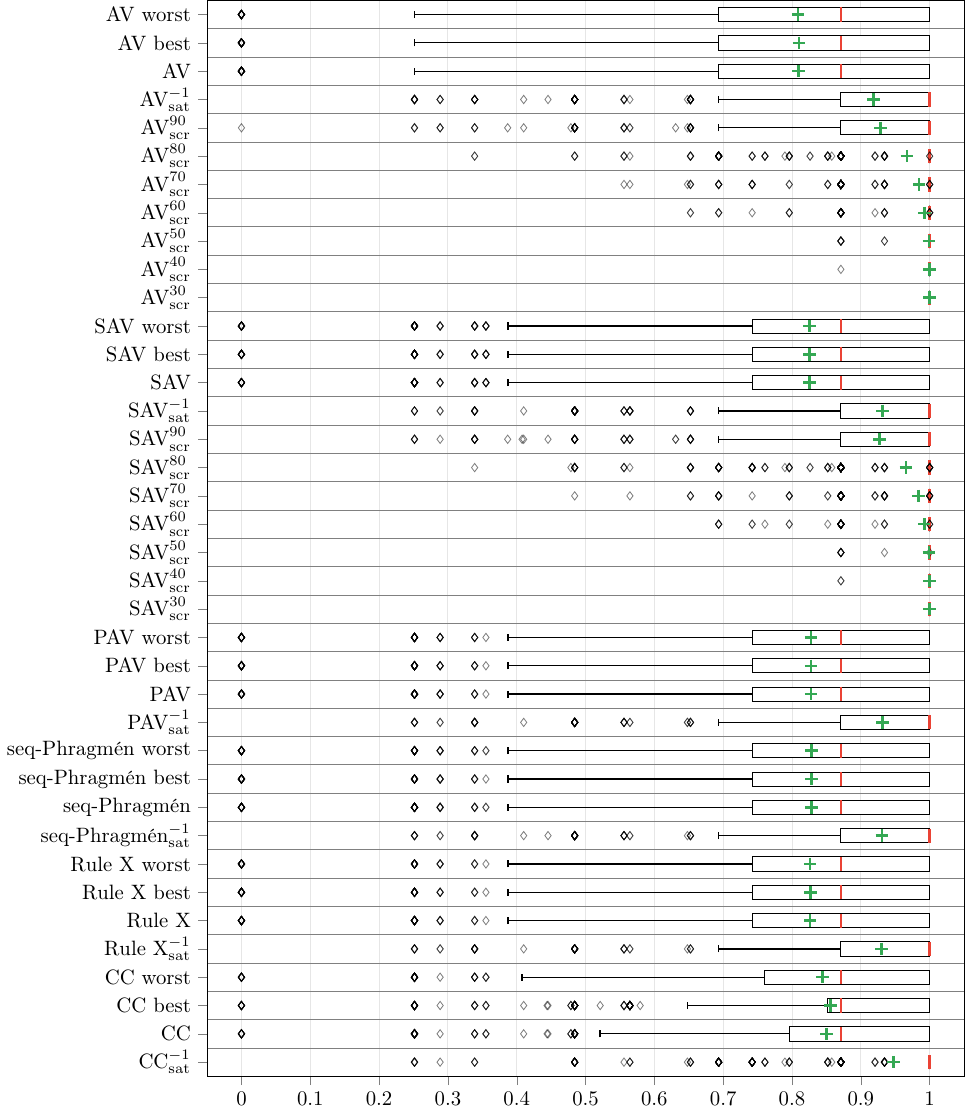}}%
  %\quad%
  \caption{\boldmath The proportion of the optimal diversity reached over all experimental data with $k=6$ when using $\shann$ for the different rules and approaches we consider. ``$\Rule$ best'' (``$\Rule$ worst'') refers to the rule that chooses the committees with the highest (lowest) diversities among the winning committees of $\Rule$. If only $\Rule$ is written, the diversity of the winning committee that \emph{abcvoting} returns for $\Rule$ is considered. The red line indicates the median, the green cross the mean.}
  \label{fig:shann6}
\end{figure*}

\begin{figure*}\centering%
  \makebox{\includegraphics{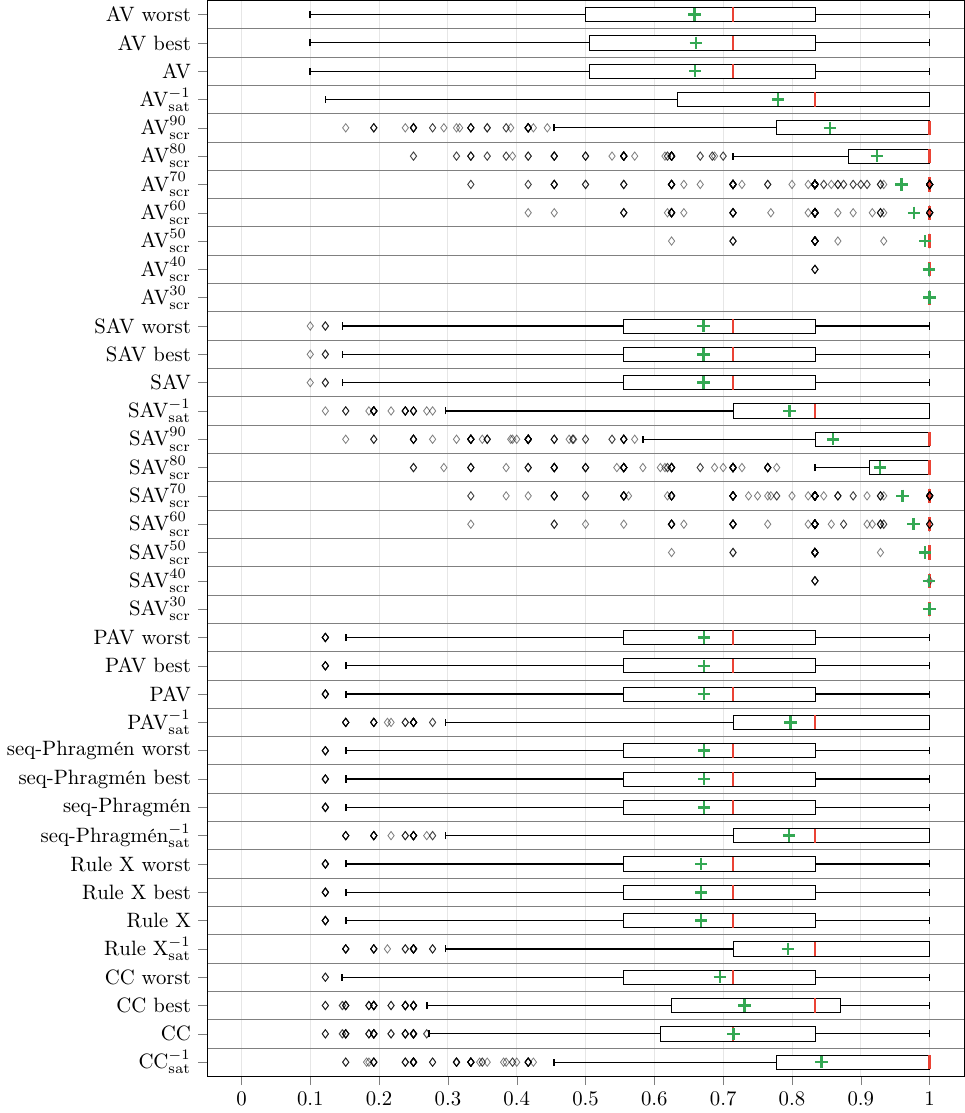}}%
  %\quad%
  \caption{\boldmath The proportion of the optimal diversity reached over all experimental data with $k=10$ when using $\simps$ for the different rules and approaches we consider. ``$\Rule$ best'' (``$\Rule$ worst'') refers to the rule that chooses the committees with the highest (lowest) diversities among the winning committees of $\Rule$. If only $\Rule$ is written, the diversity of the winning committee that \emph{abcvoting} returns for $\Rule$ is considered. The red line indicates the median, the green cross the mean.}
  \label{fig:simps10}
\end{figure*}

\begin{figure*}\centering%
  \makebox{\includegraphics{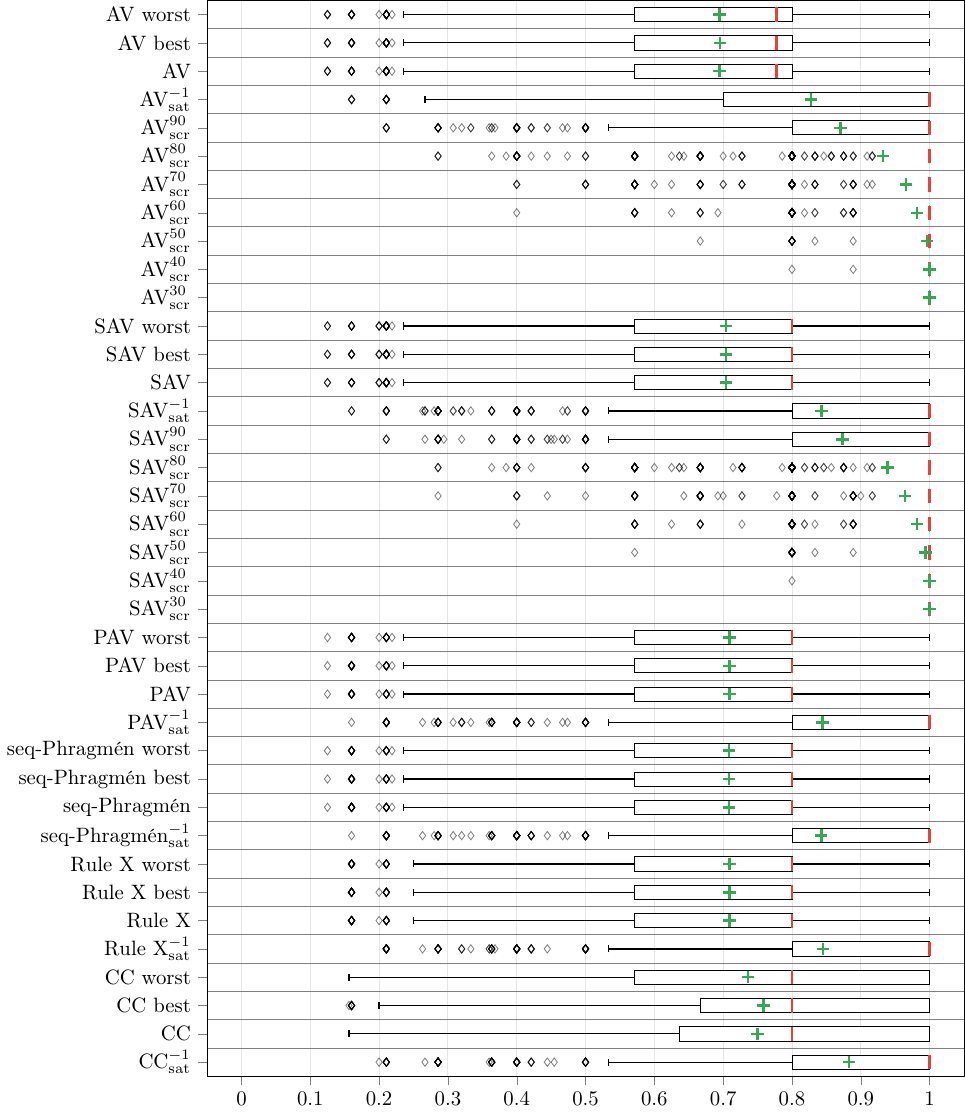}}%
  %\quad%
  \caption{\boldmath The proportion of the optimal diversity reached over all experimental data with $k=8$ when using $\simps$ for the different rules and approaches we consider. ``$\Rule$ best'' (``$\Rule$ worst'') refers to the rule that chooses the committees with the highest (lowest) diversities among the winning committees of $\Rule$. If only $\Rule$ is written, the diversity of the winning committee that \emph{abcvoting} returns for $\Rule$ is considered. The red line indicates the median, the green cross the mean.}
  \label{fig:simps8}
\end{figure*}

\begin{figure*}\centering%
  \makebox{\includegraphics{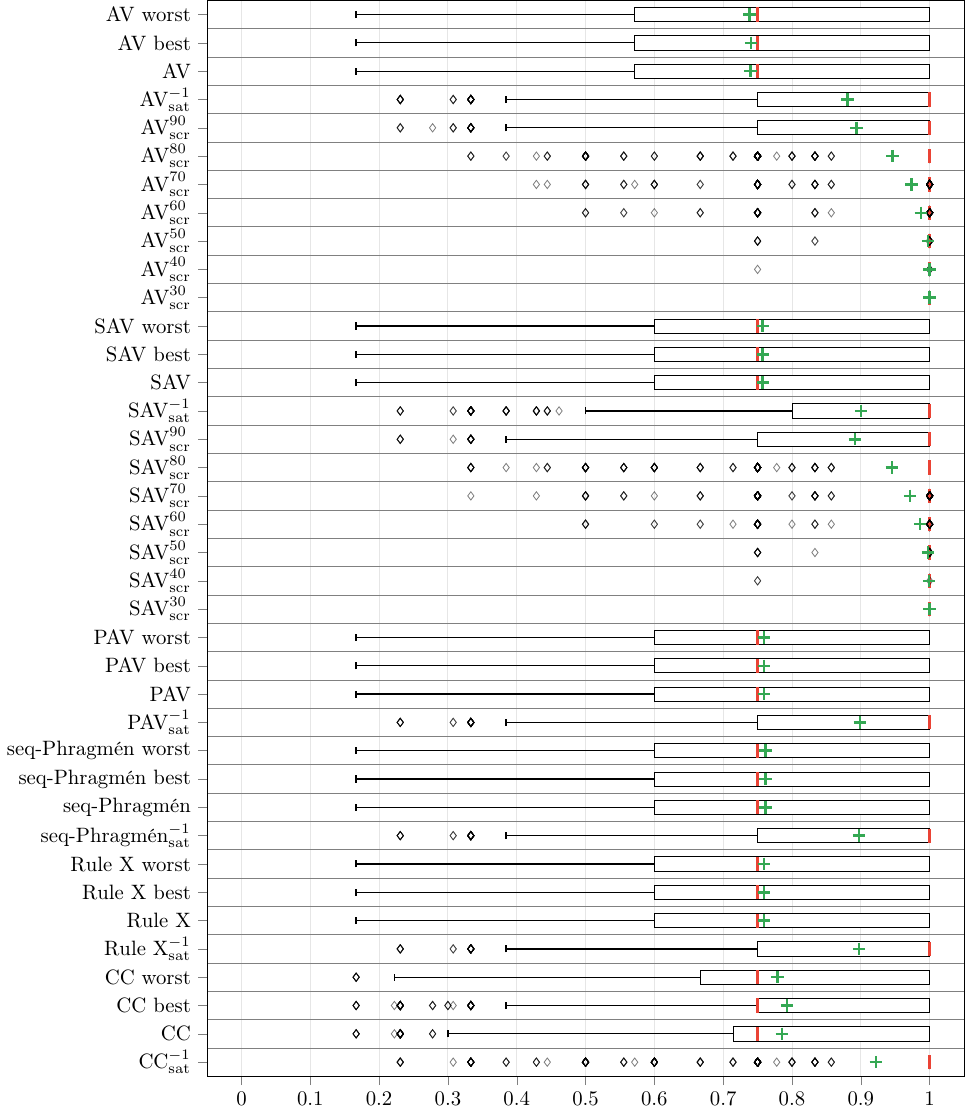}}%
  %\quad%
  \caption{\boldmath The proportion of the optimal diversity reached over all experimental data with $k=6$ when using $\simps$ for the different rules and approaches we consider. ``$\Rule$ best'' (``$\Rule$ worst'') refers to the rule that chooses the committees with the highest (lowest) diversities among the winning committees of $\Rule$. If only $\Rule$ is written, the diversity of the winning committee that \emph{abcvoting} returns for $\Rule$ is considered. The red line indicates the median, the green cross the mean.}
  \label{fig:simps6}
\end{figure*}

\begin{figure*}\centering%
  \makebox{\includegraphics{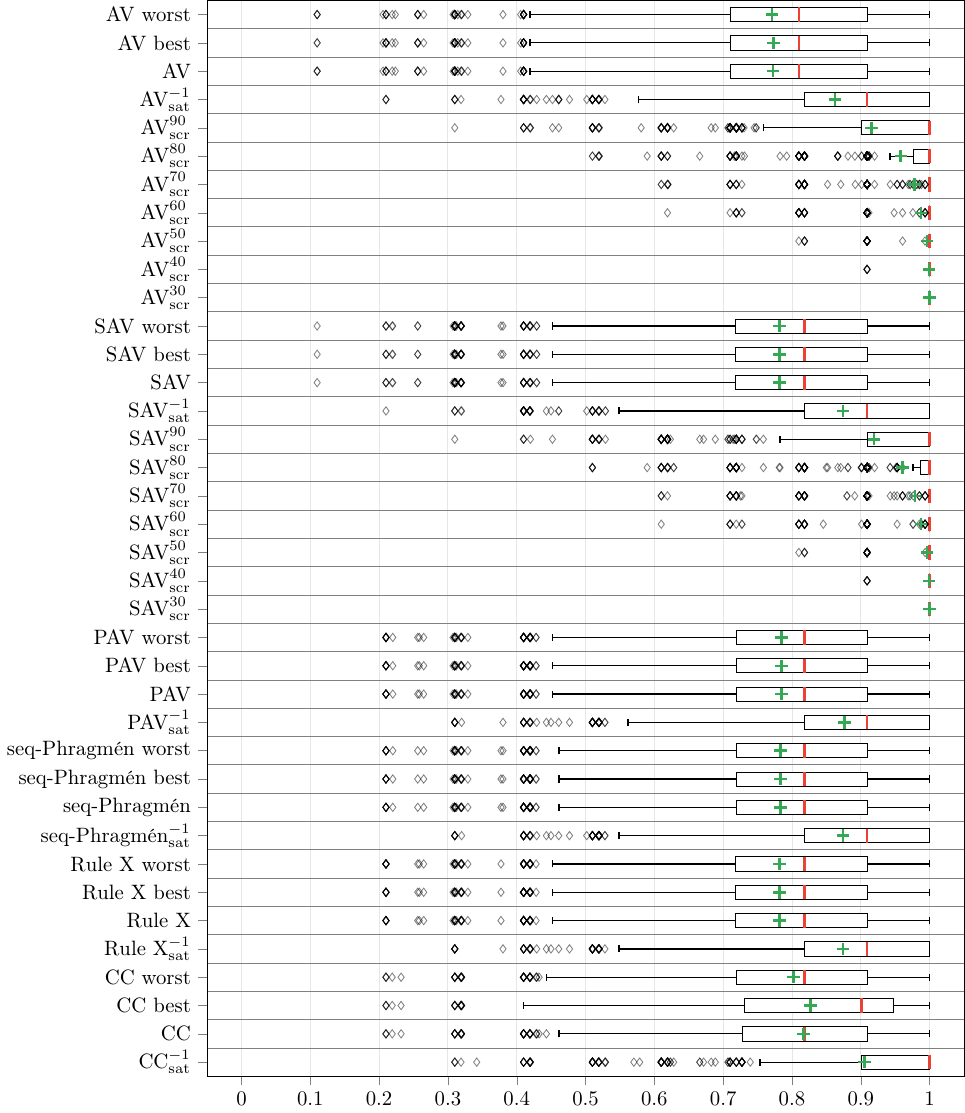}}%
  %\quad%
  \caption{\boldmath The proportion of the optimal diversity reached over all experimental data with $k=10$ when using $\dns$ for the different rules and approaches we consider. ``$\Rule$ best'' (``$\Rule$ worst'') refers to the rule that chooses the committees with the highest (lowest) diversities among the winning committees of $\Rule$. If only $\Rule$ is written, the diversity of the winning committee that \emph{abcvoting} returns for $\Rule$ is considered. The red line indicates the median, the green cross the mean.}
  \label{fig:n210}
\end{figure*}

\begin{figure*}\centering%
  \makebox{\includegraphics{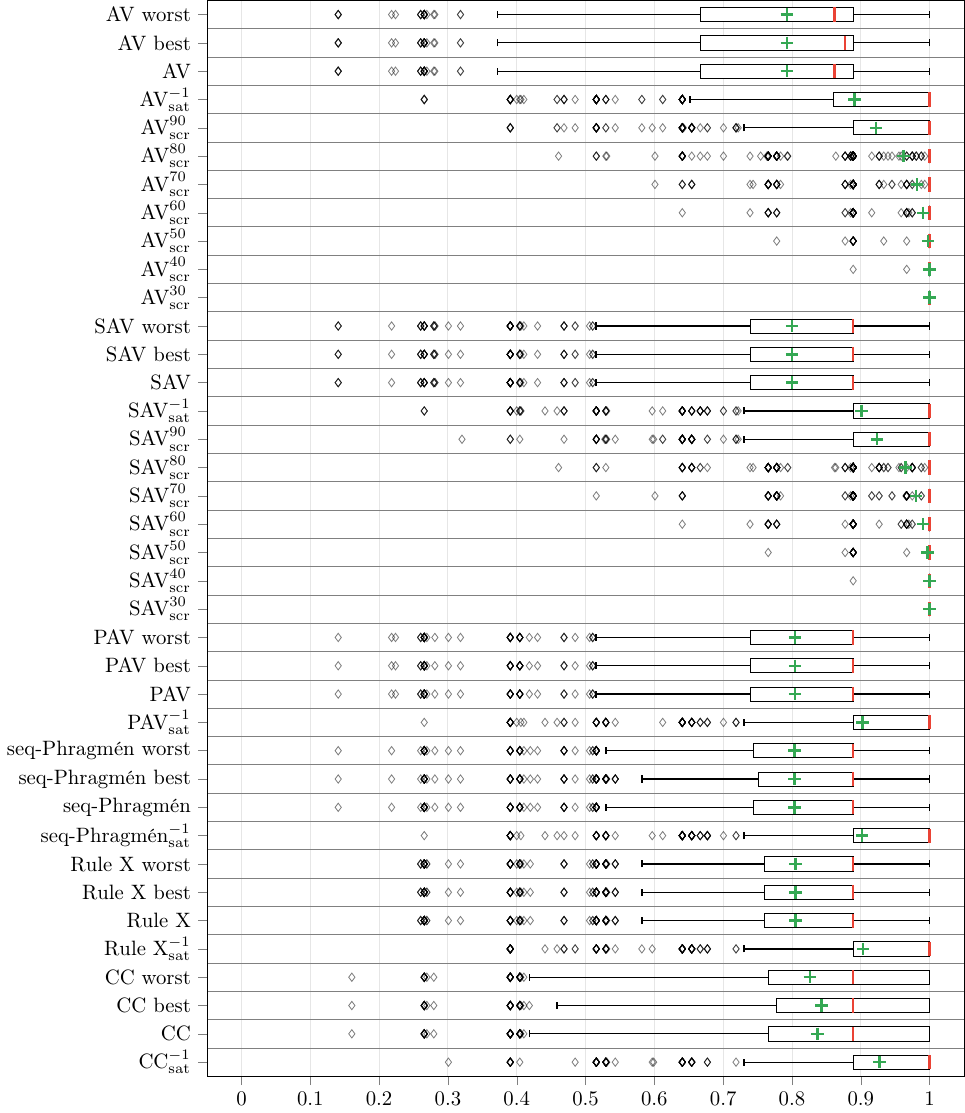}}%
  %\quad%
  \caption{\boldmath The proportion of the optimal diversity reached over all experimental data with $k=8$ when using $\dns$ for the different rules and approaches we consider. ``$\Rule$ best'' (``$\Rule$ worst'') refers to the rule that chooses the committees with the highest (lowest) diversities among the winning committees of $\Rule$. If only $\Rule$ is written, the diversity of the winning committee that \emph{abcvoting} returns for $\Rule$ is considered. The red line indicates the median, the green cross the mean.}
  \label{fig:n28}
\end{figure*}

\begin{figure*}\centering%
  \makebox{\includegraphics{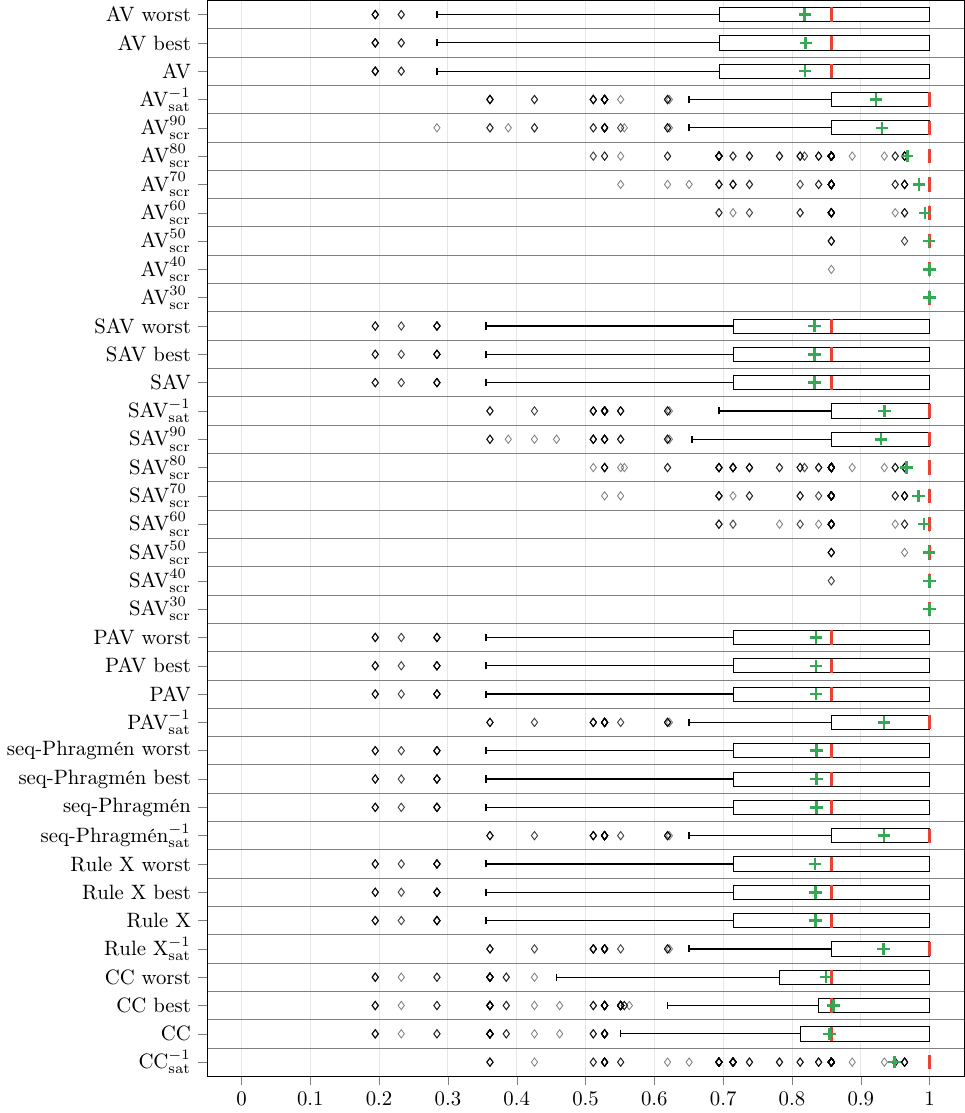}}%
  %\quad%
  \caption{\boldmath The proportion of the optimal diversity reached over all experimental data with $k=6$ when using $\dns$ for the different rules and approaches we consider. ``$\Rule$ best'' (``$\Rule$ worst'') refers to the rule that chooses the committees with the highest (lowest) diversities among the winning committees of $\Rule$. If only $\Rule$ is written, the diversity of the winning committee that \emph{abcvoting} returns for $\Rule$ is considered. The red line indicates the median, the green cross the mean.}
  \label{fig:n26}
\end{figure*}

\begin{table*}\centering
  \begin{tblr}{
    colspec={lrrrrrrrrrrrr},
    column{4} = {rightsep=12pt},
    column{5} = {leftsep=12pt},
    column{7} = {rightsep=12pt},
    column{8} = {leftsep=12pt},
    column{10} = {rightsep=12pt},
    column{11} = {leftsep=12pt},
  }
    Index & \SetCell[c=3]{c}{$\rich$}&&& \SetCell[c=3]{c}{$\shann$}&&& \SetCell[c=3]{c}{$\simps$}&&& \SetCell[c=3]{c}{$\dns$} \\%
    \cmidrule[lr]{2-4}\cmidrule[lr]{5-7}\cmidrule[lr]{8-10}\cmidrule[lr]{11-13}
    $k$                             & $10$ & $8$ & $6$ & $10$ & $8$ & $6$ & $10$ & $8$ & $6$ & $10$ & $8$ & $6$ \\ \midrule
    AV                              &   76 &  78 &  80 &   80 &  81 &  81 &   66 &  69 &  74 &   77 &  79 &  82 \\
    $\satRule{\text{AV}}$           &   86 &  89 &  92 &   88 &  90 &  92 &   78 &  83 &  88 &   86 &  89 &  92 \\
    $\scrRuleX{AV}{90}$             &   91 &  92 &  92 &   93 &  93 &  93 &   85 &  87 &  89 &   92 &  92 &  93 \\
    $\scrRuleX{AV}{80}$             &   96 &  96 &  96 &   96 &  97 &  97 &   92 &  93 &  95 &   96 &  96 &  97 \\
    $\scrRuleX{AV}{70}$             &   98 &  98 &  98 &   98 &  98 &  98 &   96 &  97 &  97 &   98 &  98 &  98 \\
    $\scrRuleX{AV}{60}$             &   99 &  99 &  99 &   99 &  99 &  99 &   98 &  98 &  99 &   99 &  99 &  99 \\
    $\scrRuleX{AV}{50}$             &  100 & 100 & 100 &  100 & 100 & 100 &   99 & 100 & 100 &  100 & 100 & 100 \\
    $\scrRuleX{AV}{40}$             &  100 & 100 & 100 &  100 & 100 & 100 &  100 & 100 & 100 &  100 & 100 & 100 \\
    $\scrRuleX{AV}{30}$             &  100 & 100 & 100 &  100 & 100 & 100 &  100 & 100 & 100 &  100 & 100 & 100 \\
    SAV                             &   77 &  79 &  81 &   81 &  82 &  83 &   67 &  70 &  76 &   78 &  80 &  83 \\
    $\satRule{\text{SAV}}$          &   87 &  90 &  93 &   89 &  91 &  93 &   80 &  84 &  90 &   87 &  90 &  93 \\
    $\scrRuleX{SAV}{90}$            &   92 &  92 &  92 &   93 &  93 &  93 &   86 &  87 &  89 &   92 &  92 &  93 \\
    $\scrRuleX{SAV}{80}$            &   96 &  96 &  96 &   97 &  97 &  97 &   93 &  94 &  95 &   96 &  97 &  97 \\
    $\scrRuleX{SAV}{70}$            &   98 &  98 &  98 &   98 &  98 &  98 &   96 &  96 &  97 &   98 &  98 &  98 \\
    $\scrRuleX{SAV}{60}$            &   99 &  99 &  99 &   99 &  99 &  99 &   98 &  98 &  99 &   99 &  99 &  99 \\
    $\scrRuleX{SAV}{50}$            &  100 & 100 & 100 &  100 & 100 & 100 &   99 &  99 & 100 &  100 & 100 & 100 \\
    $\scrRuleX{SAV}{40}$            &  100 & 100 & 100 &  100 & 100 & 100 &  100 & 100 & 100 &  100 & 100 & 100 \\
    $\scrRuleX{SAV}{30}$            &  100 & 100 & 100 &  100 & 100 & 100 &  100 & 100 & 100 &  100 & 100 & 100 \\
    PAV                             &   78 &  79 &  82 &   82 &  82 &  83 &   67 &  71 &  76 &   78 &  80 &  83 \\
    $\satRule{\text{PAV}}$          &   87 &  90 &  93 &   89 &  91 &  93 &   80 &  84 &  90 &   88 &  90 &  93 \\
    seq-Phragmén                    &   77 &  79 &  82 &   81 &  82 &  83 &   67 &  71 &  76 &   78 &  80 &  84 \\
    $\satRule{\text{seq-Phragmén}}$ &   87 &  90 &  93 &   89 &  91 &  93 &   80 &  84 &  90 &   87 &  90 &  93 \\
    Rule X                          &   77 &  79 &  82 &   81 &  82 &  83 &   67 &  71 &  76 &   78 &  81 &  83 \\
    $\satRule{\text{Rule X}}$       &   87 &  90 &  93 &   89 &  91 &  93 &   79 &  85 &  90 &   87 &  90 &  93 \\
    CC                              &   81 &  83 &  84 &   84 &  85 &  85 &   71 &  75 &  79 &   82 &  84 &  85 \\
    $\satRule{\text{Rule CC}}$      &   90 &  92 &  94 &   92 &  93 &  95 &   84 &  88 &  92 &   91 &  93 &  95 \\
  \end{tblr}
  \caption{For each diversity index in $\brcs{\rich, \shann, \simps}$, each rule considered, and each $k\in\brcs{6,8,10}$, the average percentage of the optimal diversity that the rule reaches is stated.}\label{tab:exp:ndo}
\end{table*}
  
\begin{table*}\centering
  \begin{tblr}{
    colspec={lrrrrrrrrrrrr},
    column{4} = {rightsep=12pt},
    column{5} = {leftsep=12pt},
    column{7} = {rightsep=12pt},
    column{8} = {leftsep=12pt},
    column{10} = {rightsep=12pt},
    column{11} = {leftsep=12pt},
  }
    Index & \SetCell[c=3]{c}{$\rich$}&&& \SetCell[c=3]{c}{$\shann$}&&& \SetCell[c=3]{c}{$\simps$}&&& \SetCell[c=3]{c}{$\dns$} \\%
    \cmidrule[lr]{2-4}\cmidrule[lr]{5-7}\cmidrule[lr]{8-10}\cmidrule[lr]{11-13}
    $k$                             & $10$ & $8$ & $6$ & $10$ & $8$ & $6$ & $10$ & $8$ & $6$ & $10$ & $8$ & $6$ \\ \midrule
    AV                              &   17 &  23 &  35 &   14 &  22 &  35 &   14 &  22 &  35 &   14 &  22 &  35 \\
    $\satRule{\text{AV}}$           &   42 &  55 &  69 &   36 &  52 &  68 &   36 &  52 &  68 &   36 &  52 &  68 \\
    $\scrRuleX{AV}{90}$             &   59 &  64 &  71 &   54 &  60 &  70 &   54 &  60 &  70 &   54 &  60 &  70 \\
    $\scrRuleX{AV}{80}$             &   78 &  81 &  84 &   73 &  77 &  83 &   73 &  77 &  83 &   73 &  77 &  83 \\
    $\scrRuleX{AV}{70}$             &   88 &  90 &  92 &   84 &  87 &  91 &   84 &  87 &  91 &   84 &  87 &  91 \\
    $\scrRuleX{AV}{60}$             &   92 &  94 &  96 &   91 &  92 &  95 &   91 &  92 &  95 &   91 &  92 &  95 \\
    $\scrRuleX{AV}{50}$             &   97 &  98 &  99 &   97 &  98 &  99 &   97 &  98 &  99 &   97 &  98 &  99 \\
    $\scrRuleX{AV}{40}$             &   99 & 100 & 100 &   99 & 100 & 100 &   99 & 100 & 100 &   99 & 100 & 100 \\
    $\scrRuleX{AV}{30}$             &  100 & 100 & 100 &  100 & 100 & 100 &  100 & 100 & 100 &  100 & 100 & 100 \\
    SAV                             &   17 &  23 &  36 &   13 &  21 &  36 &   13 &  21 &  36 &   13 &  21 &  36 \\
    $\satRule{\text{SAV}}$          &   43 &  58 &  75 &   38 &  54 &  73 &   38 &  54 &  73 &   38 &  54 &  73 \\
    $\scrRuleX{SAV}{90}$            &   57 &  63 &  70 &   51 &  59 &  69 &   51 &  59 &  69 &   51 &  59 &  69 \\
    $\scrRuleX{SAV}{80}$            &   79 &  83 &  84 &   74 &  79 &  83 &   74 &  79 &  83 &   74 &  79 &  83 \\
    $\scrRuleX{SAV}{70}$            &   88 &  89 &  91 &   85 &  86 &  90 &   85 &  86 &  90 &   85 &  86 &  90 \\
    $\scrRuleX{SAV}{60}$            &   92 &  94 &  95 &   90 &  92 &  95 &   90 &  92 &  95 &   90 &  92 &  95 \\
    $\scrRuleX{SAV}{50}$            &   97 &  97 &  99 &   97 &  97 &  99 &   97 &  97 &  99 &   97 &  97 &  99 \\
    $\scrRuleX{SAV}{40}$            &   99 & 100 & 100 &   99 & 100 & 100 &   99 & 100 & 100 &   99 & 100 & 100 \\
    $\scrRuleX{SAV}{30}$            &  100 & 100 & 100 &  100 & 100 & 100 &  100 & 100 & 100 &  100 & 100 & 100 \\
    PAV                             &   17 &  24 &  38 &   14 &  22 &  37 &   14 &  22 &  37 &   14 &  22 &  37 \\
    $\satRule{\text{PAV}}$          &   44 &  59 &  73 &   38 &  55 &  72 &   38 &  55 &  72 &   38 &  55 &  72 \\
    seq-Phragmén                    &   18 &  24 &  39 &   14 &  22 &  38 &   14 &  22 &  38 &   14 &  22 &  38 \\
    $\satRule{\text{seq-Phragmén}}$ &   43 &  58 &  74 &   38 &  54 &  72 &   38 &  54 &  72 &   38 &  54 &  72 \\
    Rule X                          &   17 &  24 &  38 &   13 &  22 &  38 &   13 &  22 &  38 &   13 &  22 &  38 \\
    $\satRule{\text{Rule X}}$       &   44 &  59 &  74 &   38 &  54 &  72 &   38 &  54 &  72 &   38 &  54 &  72 \\
    CC                              &   23 &  32 &  42 &   18 &  28 &  41 &   18 &  28 &  41 &   18 &  28 &  41 \\
    $\satRule{\text{Rule CC}}$      &   56 &  69 &  81 &   52 &  66 &  79 &   52 &  66 &  79 &   52 &  66 &  79 \\
  \end{tblr}
  \caption{For each diversity index in $\brcs{\rich, \shann, \simps}$, each rule considered, and each $k\in\brcs{6,8,10}$, the percentage of instances for which the rule achieved the optimal diversity is stated.}\label{tab:exp:optim:ndo}
\end{table*}
}

\toappendix{
  \section{Outlook: Weighted Labels}
  \label{sec:weights}
  There are scenarios in which different labels have a different importance,
e.g.~in a participatory budgeting scenario
in which community projects need to be picked in a city
that conducted such elections in previous years.
The goal may not be to pick a set of projects that itself is diverse
(with each project being equally important), but
one that is diverse when taking into account the 
projects chosen in the past. %---e.g., there could have been already many projects about education,
%but fewer about environmental protection.
For this, different labels could get different weights. %, representing their importance.
While this is not the focus of this paper---in which we wanted to focus on the base case with equally important labels, especially as the diversity indices from ecology also make this assumption---, we want to present some preliminary results, showing that a straight-forward adaption of the diversity indices leads to all indices apart from the adapted $\simps$ index not satisfying a desirable property for the weighted scenario.
For the adapted $\simps$ index, we can show an analogous result to \cref{corl:dssP} using the same technique with only small adjustments, showing that our results for the non-weighted scenario can be useful for the weighted-scenario:

For the latter, each label could be assigned a weight
$\omega \colon \bcks[m] \to \N$ with the following idea: For a label $l_i$,
adding a candidate with label $l_i$ increases the number of $l_i$'s occurrences not by
$1$, but by $\frac{1}{\fun{\omega}(l)}$.
Thus, if e.g.~$\fun{\omega}(i) = 1$ and $\fun{\omega}(j) = 3$, i.e.~$l_j$ is
three times as important as $l_i$,
$l_j$ should be represented three times as many times
as $l_i$,
i.e.~$\frac{n_i}{\fun{\omega}(i)}=\frac{n_j}{\fun{\omega}(j)}$.
We express this through the following property,
demanding that increasing the frequency of a label by one by
decreasing that of a different label should only improve the diversity
if this will bring the committee closer to the goal of
$\frac{n_i}{\fun{\omega}(i)}=\frac{n_j}{\fun{\omega}(j)} \iff n_i\cdot \fun{\omega}(j) = n_j \cdot \fun{\omega}(i)$.
\begin{property}[Prioritization of More Important Labels]\label{propty:PrioWeights}
  A diversity index $D$ satisfies \emph{Prioritization of More Important Labels} if, for all elections $\Elec$, $\Com_1, \Com_2 \in \fun{\Rulebase}(\Elec)$ for which there are $i,j\in\bcks[m]$ so that $n_i(\Elec,\Com_2) = n_i(\Elec,\Com_1) + 1$, $n_j(\Elec,\Com_2) = n_j(\Elec,\Com_1) - 1\geq 0$, and $\forall o\in\bcks[m]\setminus\brcs{i,j}: n_o(\Elec,\Com_2) = n_o(\Elec,\Com_1)$, it holds that $D(\Elec, \Com_2) > D(\Elec, \Com_1)$ if and only if
    $|n_j(\Elec,\Com_1) \omega(i) - n_i(\Elec,\Com_1) \omega(j)|
    >  |n_j(\Elec,\Com_2) \omega(i) - n_i(\Elec,\Com_2) \omega(j)|$.
\end{property}

We incorporate the weights into a given diversity index $D$ considered in the following way, later referred to as $D'$:
\begin{align*}
  \rich'(\Elec,\Com) ={}& \sum_{i\in \bcks[m]: \fun{n_i}(\Elec,S) > 0} \omega(i),\\
  \simps'(\Elec,\Com) ={}& -\sum_{i\in \bcks[m]} \frac{1}{\omega(i)} \cdot p_i(\Elec,\Com)^2\\
  \shann'(\Elec,\Com) ={}& -\sum_{i \in \bcks[m]: \fun{p_i}(\Elec,S) > 0} \frac{1}{\omega(i)} \cdot p_i(\Elec,\Com)\cdot \tlog{p_i(\Elec,\Com)}\\
  \dns'(\Elec,\Com) = {}& \sum_{i=1}^{Z} \paren(\funbc{\min}{m, k} + 1)^{Z+1-i} \cdot \abs{\sigma'_i(\Elec,\Com)},\\
  \text{ with } c ={}& \fun{\operatorname{lcm}}(\omega(1),\dots,\omega(m)), Z = c\cdot k,\\ \sigma'_i(\Elec,\Com) = {}& \brcs{\ell\in\bcks[m]: c\cdot \frac{n_\ell(\Elec,\Com)}{\omega(\ell)} \geq i}.
\end{align*}
Thus, $\rich'(\Elec,\Com)$ still has $\rich$'s aim of minimizing the number of labels occurring zero times as the most important goal, but if only one of two labels can be chosen at least once, the one with the higher weight will be chosen.
Both $\simps$ and $\shann$ have a negation before the respective sum so that the index value becomes larger if the summands become smaller. Therefore, we weight each summand with $\frac{1}{\omega(i)}$, not with $\omega(i)$.
Due to $\dns$'s lexicographic nature, its primary goal is to maximize the number of labels occurring at least once, the secondary goal is to maximize the number of labels occurring at least twice, and so on.
Thus, for a label to become more important, its number of occurrences must be scaled to become smaller compared to the occurrences of a less important label.
We take this into account by dividing $n_\ell(\Elec,\Com)$ by $\omega(\ell)$ in $\sigma'_i$ and adjusting the rest correspondingly.

For these adapted indices, we have the following:
\begin{observation}
  \label{obs:prop:PrioWeights}
  $\rich'$, $\shann'$, and $\dns'$ do not satisfy \pref{propty:PrioWeights}, $\simps'$ does.
\end{observation}
\begin{proof}
  \textbf{Diversity index $\rich'$:} See the proof as to why $\rich$ does not satisfy \pref{propty:balance}, by setting $\omega(l) = 1$ for all $l\in \bcks[m]$.

  \textbf{Diversity index $\shann'$:} Consider an election with $k=7, \omega(i)=6, \omega(j)=1$ and committees $\Com_1, \Com_2$ so that $n_i(\Com_1) = 4, n_j(\Com_1) = 2$.
  It holds that
  \begin{align*}
    & \shann'(\Com_2)- \shann'(\Com_1)\\
    ={}& -\frac{n_i(\Com_1)+1}{\omega(i)\cdot k}\cdot \tlog{\frac{n_i(\Com_1)+1}{k}}\\
    {}&+ \frac{n_j(\Com_1)-1}{\omega(j)\cdot k}\cdot \tlog{\frac{n_j(\Com_1)-1}{k}}\\
    {}&+ \frac{n_i(\Com_1)}{\omega(i)\cdot k}\cdot \tlog{\frac{n_i(\Com_1)}{k}} - \frac{n_j(\Com_1)}{\omega(j)\cdot k}\cdot \tlog{\frac{n_j(\Com_1)}{k}}\\
    \approx {}& -0.09 < 0,
  \end{align*}
  although
  \begin{gather*}
    \abs{n_j(\Com_1) \omega(i) - n_i(\Com_1) \omega(j)} = 8-4=4\\
    >\abs{\paren(n_j(\Com_1)-1) \omega(i) - \paren(n_i(\Com_1)+1)\omega(j)}=5-4=1.
  \end{gather*}

  \textbf{Diversity index $\dns'$:} Consider an election with $k=9,m=3,i=1,j=2, \omega(1)=7, \omega(2)=2, \omega(3)=1$ and committees $\Com_1, \Com_2$ so that $n_1(\Com_1) = 4, n_2(\Com_1) = 2, n_3(\Com_1)=3$.
  It holds that $c=14, Z=14\cdot 9$ and
  \begin{gather*}
    \dns'(\Com_2)- \dns'(\Com_1)\\
    = 4^{Z+1-9} + 4^{Z+1-10} - 4^{Z+1-8}- 4^{Z+1-9}- 4^{Z+1-10}\\- 4^{Z+1-11}- 4^{Z+1-12}- 4^{Z+1-13}- 4^{Z+1-14} < 0,
  \end{gather*}
  although
  \begin{gather*}
    \abs{n_j(\Com_1) \omega(i) - n_i(\Com_1) \omega(j)} = 14-8=6\\
    >\abs{\paren(n_j(\Com_1)-1) \omega(i) - \paren(n_i(\Com_1)+1)\omega(j)}=10-7=3.
  \end{gather*}

  \textbf{Diversity index $\simps'$:} In the following, if the committee is missing as an argument, $\Com_1$ is being referred to.
  It holds that
  \begin{align*}
    &\simps'(\Com_2)-\simps'(\Com_1)'\\
    ={}& 1/k^2 \paren(-\frac{n_i^2 + 2n_i+1}{\omega(i)}-\frac{n_j^2-2n_j+1}{\omega(j)} + \frac{n_i^2}{\omega(i)}+\frac{n_j^2}{\omega(j)})\\
    ={}& 1/k^2 \paren(-\frac{1}{\omega(i)}\paren(2n_i+1) + \frac{1}{\omega(j)}\paren(2n_j -1)).
  \end{align*}
  We continue with a case distinction:
  \begin{itemize}[nosep]
    \item Assume that the following holds:
    \begin{gather*}
    \abs{n_j \omega(i) - n_i \omega(j)}> \abs{\paren(n_j-1) \omega(i) - \paren(n_i+1)\omega(j)},\\
    n_j\cdot \omega(i) > n_i\cdot \omega(j),\\
    \paren(n_j-1)\omega(i) > \paren(n_i+1)\omega(j).
  \end{gather*}
  Then it holds that $\simps'(\Com_2) > \simps'(\Com_1)$ as
  \begin{gather*}
    \simps'(\Com_2) > \simps'(\Com_1)
    \iff \simps'(\Com_2) - \simps'(\Com_1) > 0 \\
    \iff -\frac{1}{\omega(i)}\paren(2n_i+1) + \frac{1}{\omega(j)}\paren(2n_j -1) > 0\\
    \iff -\omega(j)\paren(2n_i+1) + \omega(i)\paren(2n_j -1) > 0\\
    \iff \omega(i)\paren(2n_j -1) > \omega(j)\paren(2n_i+1)\\
    \iff \omega(i)\paren(n_j -1)  + \omega(i)n_j > \omega(j)\paren(n_i+1) + \omega(j)n_i.
  \end{gather*}
  The last inequality holds because of the assumption.
  \item Assume that the following holds:
    \begin{gather*}
    \abs{n_j \omega(i) - n_i \omega(j)}
    > \abs{\paren(n_j-1) \omega(i) - \paren(n_i+1)\omega(j)},\\
    n_j\cdot \omega(i) > n_i\cdot \omega(j),\\
    \paren(n_j-1)\omega(i) \leq \paren(n_i+1)\omega(j).
  \end{gather*}
  Clearly, this implies that
  \begin{gather*}
    n_j\omega(i) - n_i \omega(j)
    > \paren(n_i+1)\omega(j)- \paren(n_j-1) \omega(i).
  \end{gather*}
  Then it holds that $\simps'(\Com_2) > \simps'(\Com_1)$ as
  \begin{gather*}
    \simps'(\Com_2) > \simps'(\Com_1)
    \iff \simps'(\Com_2) - \simps'(\Com_1) > 0 \\
    \iff -\frac{1}{\omega(i)}\paren(2n_i+1) + \frac{1}{\omega(j)}\paren(2n_j -1) > 0\\
    \iff -\omega(j)\paren(2n_i+1) + \omega(i)\paren(2n_j -1) > 0\\
    \iff \omega(i)\paren(2n_j -1) > \omega(j)\paren(2n_i+1)\\
    \iff \omega(i)\paren(n_j -1)  + \omega(i)n_j > \omega(j)\paren(n_i+1) + \omega(j)n_i\\
    \iff n_j\omega(i) - n_i \omega(j) > \paren(n_i+1)\omega(j) - \paren(n_j-1) \omega(i).
  \end{gather*}
  The last inequality holds because of the assumption.
  \item Assume that the following holds:
    \begin{gather*}
    \abs{n_j \omega(i) - n_i \omega(j)}
    < \abs{\paren(n_j-1) \omega(i) - \paren(n_i+1)\omega(j)},\\
    n_j\cdot \omega(i) > n_i\cdot \omega(j)
  \end{gather*}
  This implies that
  \begin{gather*}
    \paren(n_j-1)\omega(i) \leq \paren(n_i+1)\omega(j),
  \end{gather*}
  as otherwise it would hold that
  \begin{gather*}
    0\leq \abs{\paren(n_j-1) \omega(i) - \paren(n_i+1)\omega(j)}\\
    = \paren(n_j-1)\omega(i) - \paren(n_i+1)\omega(j) < n_j \omega(i) - n_i \omega(j)\\
    = \abs{\omega(i) - n_i \omega(j)},
  \end{gather*}
  a contradiction.
  Based on this, it holds that $\simps'(\Com_2) < \simps'(\Com_1)$ as
  \begin{gather*}
    \simps'(\Com_2) < \simps'(\Com_1)
    \iff \simps'(\Com_2) - \simps'(\Com_1) < 0 \\
    \iff -\frac{1}{\omega(i)}\paren(2n_i+1) + \frac{1}{\omega(j)}\paren(2n_j -1) < 0\\
    \iff -\omega(j)\paren(2n_i+1) + \omega(i)\paren(2n_j -1) < 0\\
    \iff \omega(i)\paren(2n_j -1) < \omega(j)\paren(2n_i+1)\\
    \iff \omega(i)\paren(n_j -1)  + \omega(i)n_j < \omega(j)\paren(n_i+1) + \omega(j)n_i\\
    \iff n_j\omega(i) - n_i \omega(j) < \paren(n_i+1)\omega(j) - \paren(n_j-1) \omega(i).
  \end{gather*}
  The last inequality holds because of the assumption.
  \item Assume that $n_j \omega(i) \leq n_i \omega(j)$ holds. Then it follows that
    \begin{gather*}
    \paren(n_j - 1) \omega(i) < n_j \omega(i) \leq n_i \omega(j) < \paren(n_i+1)\omega(j),\\
    \implies \paren(n_i+1)\omega(j) > \paren(n_j - 1)\omega(i)\\
    \implies \abs{n_j \omega(i) - n_i \omega(j)}\\
    < \abs{\paren(n_j-1) \omega(i) - \paren(n_i+1)\omega(j)}.
  \end{gather*}
  Then it holds that $\simps'(\Com_2) < \simps'(\Com_1)$ as
  \begin{gather*}
    \simps'(\Com_2) < \simps'(\Com_1)
    \iff \simps'(\Com_2) - \simps'(\Com_1) < 0 \\
    \iff -\frac{1}{\omega(i)}\paren(2n_i+1) + \frac{1}{\omega(j)}\paren(2n_j -1) < 0\\
    \iff -\omega(j)\paren(2n_i+1) + \omega(i)\paren(2n_j -1) < 0\\
    \iff \omega(i)\paren(2n_j -1) < \omega(j)\paren(2n_i+1)\\
    \iff \omega(i)\paren(n_j -1)  + \omega(i)n_j < \omega(j)\paren(n_i+1) + \omega(j)n_i\\
    \iff n_j\omega(i) - n_i \omega(j) < \paren(n_i+1)\omega(j) - \paren(n_j-1) \omega(i).
  \end{gather*}
  The last inequality holds because it follows from the assumption that $n_j\omega(i) - n_i \omega(j) \leq 0$ and $\paren(n_i+1)\omega(j) - \paren(n_j-1) \omega(i)$ > 0.
  \end{itemize}
\end{proof}

Next, we show a theorem
analogous to \cref{theorem:DSSGP} (whose proof only needs small adjustments for the weighted scenario):
\begin{theorem}
  \label{theorem:DSSGP:weights}
  \prob{MAX-$\paren(D, s)$-DSCR} is in \classP{} if 
  \begin{itemize}[nosep]
    \item $s$ is a separable function and there is an $\alpha \in\mathbb{N}_0$ polynomial in the input size so that, for each $c\in C$, $\wC(\Elec,c)\leq\alpha$,
    \item the index $D$ can be expressed as $\sum_{l\in \bcks[m]}\sum_{i=1}^{n_l}\fun{t_l}(i)$ and, for all $i \in\bcks[k]$, it holds that $0<\fun{t_l}(i)$ and $\fun{t_l}(i)$ can be computed in time polynomial in the input size, and $t_l$ is strictly monotonically decreasing, for each $l\in\bcks[m]$.
  \end{itemize}
\end{theorem}
\begin{proof}
Given an instance $I_D$ of MAX-$\paren(D, s)$-DSCR with~$n$ candidates,
we construct an instance~$I_K$ of the 0-1 Knapsack problem in polynomial time.
We assume that~$\beta \leq k\cdot \alpha$ 
($\beta$ being the lower bound for the score, see \cref{def:maxdss}), 
since there is no solution for~$I_D$ otherwise.
For each candidate~$c_i$, we add an item~$x_i$
with the weight~$w_K(x_i) \ceq n\alpha + 1 - w(c_i)$
and the value~$v(x_i) = t_j(\pi(c_i)) + \eta$,
where~$\fun{\lambda}(c_i) = l_j, \eta\ceq 1+k\cdot \max_{l\in\bcks[m]} t_l(1)$ and~$\pi$ outputs~$c_i$'s position in a descending ordering of the candidates with the same label as $c_i$ based on $w$.
The knapsack's bound is~$B\ceq k(n\alpha + 1) - \beta$.
Let, for a solution $X$ of $I_K$, $S(X) = \{c_i \mid x_i\in X\}$, $v(X)$ and $w_K(X)$ the value and weight of $X$, and, for a solution $S$ of $I_D$, $X(S) = \{x_i: c_i \in S\}$. 
We claim that $I_K$ has a solution~$X$ with value at least~$k\cdot \eta$
if and only if 
$S(X)$ is a solution to~$I_D$
and that $I_D$ is infeasible otherwise:

Let~$X$ be a solution to~$I_K$ with value at least~$k\cdot \eta$.
It follows that~$|X|\geq k$.
Assume that~$|X|>k$,
then~$w_K(X)\geq (k+1)(n\alpha + 1) - w(S(X)) \geq (k+1)(n\alpha + 1) - n\alpha = k(n\alpha+1)+1 > B$,
a contradiction.
Thus, $|X|=k$ and $w_K(X)\leq B\iff k(n\alpha + 1) - w(S(X)) \leq k(n\alpha + 1) - \beta \iff \beta \leq w(S(X))$.
Therefore,
$S(X)$ is of size exactly~$k$ and respects the score constraint.
If $X$ is a solution to~$I_K$ with value less than~$k\cdot \eta$,
then~$|X|<k$. 
Since~$X$ is maximal,
there is no size-$k$ solution respecting the capacity constraints
and hence no solution to~$I_D$ that respects the score constraint.
Analogously,
$I_K$ has no solution with value at least~$k\cdot \eta$ if~$I_D$ has none.

Finally, we show that, for an optimal solution $\OSK$ of $I_K$ with value at least~$k\cdot \eta$, $S(\OSK)$ is an optimal solution of $I_D$.
Note that, if, for a label $l_j$, $n_j$ many items $x_i$ with $\fun{\labof}(c_i) = l_j$ are chosen in $\OSK$, the $n_j$ items which come first in the order $\pi$ are always chosen and thus $v(\OSK)=\fun{D}(\Elec, S(\OSK)) + k\cdot \eta$.
Next, assume that $I_D$ has an optimal solution $\OSD\neq S(\OSK)$ with $\fun{D}(\OSD) > \fun{D}(S(\OSK))$.
Consider the committee $\OSB$ with $\fun{n_l}(\Elec, \OSB) = \fun{n_l}(\Elec, \OSD)$ for $l\in \bcks[m]$ and $c\in \OSB \iff \fun{\pi}(c)\leq \fun{n_j}(\Elec, \OSD)$ with $l_j=\fun{\labof}(c)$.
Thus, $\fun{D}(\Elec,\OSB) = \fun{D}(\Elec,\OSD) = \fun{v}(X(\OSB)) - k \cdot \eta > \fun{D}(\Elec, S(\OSK)) = \fun{v}(\OSK) - k \cdot \eta$, a contradiction.

Therefore, the problem can be solved e.g.~by the solving the 0-1 Knapsack instance with dynamic programming in $\fun{\mathcal{O}}(nB)=\fun{\mathcal{O}}(n \paren(k(n\alpha + 1) - \beta))$ time.
\end{proof}

$\simps'$ fulfills the condition in \cref{theorem:DSSGP} by using $\fun{t_l}(i)=2k+\frac{1-2i}{\omega(l)}$. Thus:
\begin{corollary}
  \label{corl:dssP:weights}
  \prob{MAX-$\paren(\simps', \mathrm{score}_{\mathrm{AV}})$-DSCR} is in \classP{}.
\end{corollary}
\begin{proof}
  For $\simps'$, it holds that maximizing $-\sum_{j\in \bcks[m]} \frac{1}{\omega(j)}\cdot p_j^2$ yields the
  same solutions (with a different objective value) as maximizing $\paren(\sum_{j\in \bcks[m]} -\frac{n_j^2}{\omega(j)})+2k^2=\sum_{j\in \bcks[m]} -\frac{n_j^2}{\omega(j)}+2\cdot n_j \cdot k$.
  Thus, $\fun{t_j}(i)=\paren(-2i+1)/\omega(j)+2k$ can be chosen (which is strictly monotonically decreasing),
  as $\sum_{i=1}^n 2k -\paren(2i -1)/\omega(j) = 2kn-n^2/\omega(j)$.
\end{proof}

}

\section{Epilogue}
\label{sec:epi}
We adapted several diversity indices used in ecology to the context of committee elections
and introduced a new diversity index.
We also introduced properties of diversity indices which allow us to differentiate between any pair of indices and to characterize the new index via two of our properties.
The underlying model assumes that each candidate has one label:
While this allows to define a label as a set of ``sub-labels'' (e.g.~$\brcs{\text{urban greenery}, \text{adults}}$ as one label),
all indices we consider do not take the (dis)similarity of labels into account.
Further research could investigate diversity indices that incorporate such distances, which also requires thinking about how such distances are determined.

Furthermore, we treat all labels as equally important, which is in line with how the diversity indices considered treat different species in the literature.
However, there are scenarios in which labels have different importance/weights.
Our preliminary results for weighted labels, presented in \cref{sec:weights} in the appendix, show that, after adapting the indices straight-forwardly to weighted labels, only the adapted $\simps$ index satisfies a desirable property that we define for the weighted scenario.
In addition, we show an analogous result to \cref{corl:dssP}, using the same technique as in \cref{theorem:DSSGP} with only small adjustments. This shows that our results for the non-weighted scenario could be useful for a weighted scenario.

From an algorithmic point of view,
we proved that \prob{$\paren(D, s)$-DSCR} is in \classP{} in some cases if $s$ is a separable scoring function.
This includes the score of AV and SAV for each diversity index considered
apart from $\shann$ in case of SAV---%
we left open whether \prob{$\paren(\shann, \mathrm{score}_{\mathrm{SAV}})$-DSCR} is polynomial-time solvable.
However, there are other $s$ for which we prove that \prob{$\paren(D, s)$-DSCR} is \NP-hard, which is also the case for \prob{$D$-DSAT}.
Further work may study parameterized complexity or approximation algorithms
for these problems.

Our experiments revealed interesting trade-offs between satisfaction/scoring guarantees and diversity, showing, among other results, that the diversity of committees can indeed be improved.
It would also be interesting to investigate how much the diversity indices differ on real world data or to evaluate past elections regarding their scoring-diversity performance.
This could be particularly interesting in the context of participatory budgeting, 
which calls for an extension of our model in which the costs of the projects and the respective
budget becomes the third objective (next to voter satisfaction and diversity).

Finally note that, while ecological indices provide rigorous measures of diversity,
their practical interpretability varies, particularly in settings with high label
cardinality.
Our axiomatic characterization aims to enhance transparency,
though systematic guidance for index selection in specific application domains
remains an important direction for future work.
\clearpage

%%%%%%%%%%%%%%%%%%%%%%%%%%%%%%%%%%%%%%%%%%%%%%%%%%%%%%%%%%%%%%%%%%%%%%%%

\section*{Acknowledgements}
% \begin{acks}
TF is supported by \thxPACSs{}
% \end{acks}

\bibliography{jd}

@book{leinster2021entropy,
  title     = {Entropy and Diversity: The Axiomatic Approach},
  isbnx      = {9781108965576},
  urlx       = {http://dx.doi.org/10.1017/9781108963558},
  doix       = {10.1017/9781108963558},
  publisher = {Cambridge University Press},
  author    = {Leinster, Tom},
  year      = {2021}
}

@article{pielou1975ecological,
  title   = {Ecological Diversity},
  author  = {Pielou, E.C.},
  journal = {John Wiley \& Sons},
  year    = {1975}
}

@article{whittaker1972evolution,
  title     = {Evolution and Measurement of Species Diversity},
  author    = {Whittaker, Robert H.},
  journal   = {Taxon},
  volume    = {21},
  number    = {2-3},
  pages     = {213--251},
  year      = {1972},
  publisher = {Wiley},
  urlx       = {http://dx.doi.org/10.2307/1218190},
  doix       = {10.2307/1218190}
}

@article{Shannon48,
  title     = {A mathematical theory of communication},
  author    = {Shannon, Claude E.},
  journal   = {The Bell System Technical Journal},
  volume    = {27},
  number    = {3},
  pages     = {379--423},
  year      = {1948},
  publisher = {Institute of Electrical and Electronics Engineers (IEEE)},
  urlx       = {http://dx.doi.org/10.1002/j.1538-7305.1948.tb01338.x},
  doix       = {10.1002/j.1538-7305.1948.tb01338.x}
}

@article{Simpson49,
  title     = {Measurement of Diversity},
  volume    = {163},
  issnx      = {1476-4687},
  urlx       = {http://dx.doi.org/10.1038/163688a0},
  doix       = {10.1038/163688a0},
  number    = {4148},
  journal   = {Nature},
  publisher = {Springer},
  author    = {Simpson, Edward H.},
  year      = {1949},
  pages     = {688--688}
}

@book{lackner2023multi,
  title     = {Multi-Winner Voting with Approval Preferences},
  isbnx      = {9783031090165},
  urlx       = {http://dx.doi.org/10.1007/978-3-031-09016-5},
  doix       = {10.1007/978-3-031-09016-5},
  series       = {Springer Briefs in Intelligent Systems},
  publisher    = {Springer},
  author    = {Lackner, Martin and Skowron, Piotr},
  year      = {2023},
}

@inproceedings{ijcai2023p297,
  title     = {Participatory Budgeting: Data, Tools and Analysis},
  author    = {Faliszewski, Piotr and Flis, Jarosław and Peters, Dominik and Pierczyński, Grzegorz and Skowron, Piotr and Stolicki, Dariusz and Szufa, Stanisław and Talmon, Nimrod},
  booktitle = {Proceedings of the 32nd International Joint Conference on
               Artificial Intelligence},
  publisher = {International Joint Conferences on Artificial Intelligence Organization},
  editorx    = {Edith Elkind},
  pages     = {2667--2674},
  year      = {2023},
  month     = {8},
  doix       = {10.24963/ijcai.2023/297},
  urlx       = {https://doi.org/10.24963/ijcai.2023/297},
  series    = {IJCAI-2023}
}

@inproceedings{celisijcaiFairness,
  title     = {Multiwinner Voting with Fairness Constraints},
  author    = {L. Elisa Celis and Lingxiao Huang and Nisheeth K. Vishnoi},
  booktitle = {Proceedings of the 27th International Joint Conference on Artificial Intelligence (IJCAI'18)},
  publisher = {International Joint Conferences on Artificial Intelligence Organization},
  pages     = {144--151},
  year      = {2018},
  xmonth    = {7},
  xdoi      = {10.24963/ijcai.2018/20},
  xurl      = {https://doi.org/10.24963/ijcai.2018/20}
}

@inproceedings{bredereck2018divconstr,
  title     = {Multiwinner elections with diversity constraints},
  author    = {Bredereck, Robert and Faliszewski, Piotr and Igarashi, Ayumi and Lackner, Martin and Skowron, Piotr},
  booktitle = {Proceedings of the 32nd AAAI Conference on Artificial Intelligence (AAAI'18)},
  publisher = {AAAI Press},
  year      = {2018},
  pages        = {933--940},
}

@article{ianovski2022electingDomConstraints,
  title     = {Electing a committee with dominance constraints},
  author    = {Ianovski, Egor},
  journal   = {Annals of Operations Research},
  volume    = {318},
  number    = {2},
  pages     = {985--1000},
  year      = {2022},
  publisher = {Springer}
}

@article{aziz2019softconstr,
  title     = {A rule for committee selection with soft diversity constraints},
  author    = {Aziz, Haris},
  journal   = {Group Decision and Negotiation},
  volume    = {28},
  number    = {6},
  pages     = {1193--1200},
  year      = {2019},
  publisher = {Springer}
}

@inproceedings{2022Switzerland,
  title     = {Diverse representation via computational participatory elections-lessons from a case study},
  author    = {Evequoz, Florian and Rochel, Johan and Keswani, Vijay and Celis, L Elisa},
  booktitle = {Proceedings of the 2nd ACM Conference on Equity and Access in Algorithms, Mechanisms, and Optimization (EAAMO'22)},
  publisher = {ACM},
  pages     = {1--11},
  year      = {2022}
}

@article{straszak1993computer,
  title     = {Computer-assisted constrained approval voting},
  author    = {Straszak, Andrzej and Libura, Marek and Sikorski, Jarostaw and Wagner, Dariusz},
  journal   = {Group Decision and Negotiation},
  volume    = {2},
  pages     = {375--385},
  year      = {1993},
  publisher = {Springer}
}

@inproceedings{gawron2022movies,
  title        = {Using multiwinner voting to search for movies},
  author       = {Gawron, Grzegorz and Faliszewski, Piotr},
  booktitle    = {Proceedings of the 19th European Conference on Multi-Agent Systems ({EUMAS}'22)},
  series       = {LNCS},
  volume       = {13442},
  pages        = {134--151},
  year         = {2022},
  organization = {Springer}
}

@inproceedings{ijcai2022dire,
  title     = {DiRe Committee : Diversity and Representation Constraints in Multiwinner Elections},
  author    = {Relia, Kunal},
  booktitle = {Proceedings of the 31st International Joint Conference on Artificial Intelligence ({IJCAI}'22)},
  publisher = {International Joint Conferences on Artificial Intelligence Organization},
  pages     = {5143--5149},
  year      = {2022},
  month     = {7}
}

@inproceedings{izsak2018synergies,
  title     = {Committee selection with intraclass and interclass synergies},
  author    = {Izsak, Rani and Talmon, Nimrod and Woeginger, Gerhard},
  booktitle = {Proceedings of the 32nd AAAI Conference on Artificial Intelligence (AAAI'18)},
  publisher = {AAAI Press},
  year      = {2018},
  pages        = {1071--1078}
}

@article{joss-abcvoting,
  title     = {abcvoting: {A} {P}ython package for approval-based multi-winner voting rules},
  author    = {Martin Lackner and Peter Regner and Benjamin Krenn},
  doix       = {10.21105/joss.04880},
  year      = {2023},
  publisher = {The Open Journal},
  volume    = {8},
  number    = {81},
  pages     = {4880},
  journal   = {Journal of Open Source Software}
}

@article{mattei2017apreflib,
  title   = {A Preflib.org Retrospective: Lessons Learned and New Directions},
  author  = {Mattei, Nicholas and Walsh, Toby},
  journal = {Trends in Computational Social Choice. AI Access Foundation},
  pages   = {289--309},
  year    = {2017}
}

@article{LaSt08,
  title   = {A live experiment on approval voting},
  author  = {Laslier, Jean-Fran{\c{c}}ois and Van der Straeten, Karine},
  journal = {Experimental Economics},
  volume  = {11},
  number  = {1},
  pages   = {97--105},
  year    = {2008},
  publisher={Cambridge University Press}
}

@article{Laslier_2004,
  title     = {Une expérience de vote par assentiment lors de l’élection présidentielle française de 2002},
  volume    = {54},
  issnx      = {1950-6686},
  urlx       = {http://dx.doi.org/10.3917/rfsp.541.0099},
  doix       = {10.3917/rfsp.541.0099},
  number    = {1},
  journal   = {Revue française de science politique},
  publisher = {CAIRN},
  author    = {Laslier, Jean-François and Van der Straeten, Karine},
  year      = {2004},
  pages     = {99}
}

@article{AzizBCEFW17,
  author       = {Haris Aziz and Markus Brill and Vincent Conitzer and Edith Elkind and Rupert Freeman and Toby Walsh},
  title        = {Justified representation in approval-based committee voting},
  journal      = {Social Choice and Welfare},
  volume       = {48},
  number       = {2},
  pages        = {461--485},
  year         = {2017},
  publisher    ={Springer},
}

@article{leinster2012measuring,
  title     = {Measuring diversity: the importance of species similarity},
  author    = {Leinster, Tom and Cobbold, Christina A.},
  journal   = {Ecology},
  volume    = {93},
  issnx      = {1939-9170},
  urlx       = {http://dx.doi.org/10.1890/10-2402.1},
  doix       = {10.1890/10-2402.1},
  number    = {3},
  pages     = {477--489},
  year      = {2012},
  publisher = {Wiley}
}

@article{IZSAK2000151,
  title    = {A link between ecological diversity indices and measures of biodiversity},
  journal  = {Ecological Modelling},
  volume   = {130},
  number   = {1},
  pages    = {151-156},
  year     = {2000},
  issnx     = {0304-3800},
  doix      = {https://doi.org/10.1016/S0304-3800(00)00203-9},
  urlx      = {https://www.sciencedirect.com/science/article/pii/S0304380000002039},
  author   = {János Izsák and László Papp},
  publisher={Elsevier}
}

@article{baczkowski1997properties,
  title     = {Properties of a generalized diversity index},
  author    = {Baczkowski, AJ and Joanes, DN and Shamia, GM},
  journal   = {Journal of Theoretical Biology},
  volume    = {188},
  number    = {2},
  pages     = {207--213},
  year      = {1997},
  publisher = {Elsevier}
}

@article{routledgeDivIdxAdmissable,
  title     = {Diversity indices: which ones are admissible?},
  author    = {Routledge, RD},
  journal   = {Journal of Theoretical Biology},
  volume    = {76},
  number    = {4},
  pages     = {503--515},
  year      = {1979},
  publisher = {Elsevier}
}

@article{izsak1996sensitivity,
  title     = {Sensitivity profiles of diversity indices},
  author    = {Izs{\'a}k, J{\'a}nos},
  journal   = {Biometrical journal},
  volume    = {38},
  number    = {8},
  pages     = {921--930},
  year      = {1996},
  publisher = {Wiley}
}

@article{jost2006entropy,
  title     = {Entropy and diversity},
  author    = {Jost, Lou},
  journal   = {Oikos},
  volume    = {113},
  number    = {2},
  pages     = {363--375},
  year      = {2006},
  publisher = {Wiley}
}

@article{procaccia2008CChard,
  author       = {Ariel D. Procaccia and
                  Jeffrey S. Rosenschein and
                  Aviv Zohar},
  title        = {On the complexity of achieving proportional representation},
  journal      = {Social Choice and Welfare},
  volume       = {30},
  number       = {3},
  pages        = {353--362},
  year         = {2008},
  xurl          = {https://doi.org/10.1007/s00355-007-0235-2},
  xdoi          = {10.1007/S00355-007-0235-2}
}

@inproceedings{aziz2015PAChard,
  author       = {Haris Aziz and Serge Gaspers and Joachim Gudmundsson and Simon Mackenzie and Nicholas Mattei and Toby Walsh},
  title        = {Computational Aspects of Multi-Winner Approval Voting},
  booktitle    = {Proceedings of the 14th International Conference on Autonomous Agents and Multiagent Systems ({AAMAS}~'15)},
  pages        = {107--115},
  publisher    = {{ACM}},
  year         = {2015},
  urlx          = {http://dl.acm.org/citation.cfm?id=2772896}
}

@inproceedings{BenabbouCHSZ18diversityhousing,
  author       = {Nawal Benabbou and Mithun Chakraborty and Xuan{-}Vinh Ho and Jakub Sliwinski and Yair Zick},
  title        = {Diversity Constraints in Public Housing Allocation},
  booktitle    = {Proceedings of the 17th International Conference on Autonomous Agents and Multiagent Systems ({AAMAS}~'18)},
  pages        = {973--981},
  publisher    = {IFAAMAS / {ACM}},
  year         = {2018},
  xurl          = {http://dl.acm.org/citation.cfm?id=3237843}
}

@article{AB21,
Author = {Ayg{\"u}n, Orhan and Bó, Inácio},
Title = {College Admission with Multidimensional Privileges: The {B}razilian Affirmative Action Case},
Journal = {American Economic Journal: Microeconomics},
Volume = {13},
Number = {3},
Year = {2021},
Pages = {1–28},
publisher = {American Economic Association}
}

@article{AZ25,
title = {Multi-rank smart reserves: A general framework for selection and matching diversity goals},
journal = {Artificial Intelligence},
volume = {339},
pages = {104274},
year = {2025},
author = {Haris Aziz and Zhaohong Sun},
publisher = {Elsevier}
}

@article{Biro2010collegequotas,
title = {The College Admissions problem with lower and common quotas},
journal = {Theoretical Computer Science},
volume = {411},
number = {34},
pages = {3136-3153},
year = {2010},
author = {Péter Biró and Tamás Fleiner and Robert W. Irving and David F. Manlove},
publisher={Elsevier}
}

@article{GassiDissMoyouwou2023,
  author    = {Diss, Mostapha and Gassi, Clinton Gubong and Moyouwou, Issofa},
  title     = {Combining Diversity and Excellence in Multiwinner Elections},
  volume    = {34},
  issn      = {1572-9907},
  number    = {4},
  journal   = {Group Decision and Negotiation},
  publisher = {Springer},
  year      = {2025},
  month     = may,
  pages     = {683--713} 
}

@misc{Our_Repo,
  title     = {diversity-indices-in-elections},
  publisher = {GitHub},
  author    = {Böhm, Paula and Bredereck, Robert and Fluschnik, Till},
  year      = {2026},
  month     = {02},
  note      = {Published on GitHub: \url{https://github.com/paubo14/diversity-indices-in-elections}}
}

%%%%%%%%%%%%%%%%%%%%%%%%%%%%%%%%%%%%%%%%%%%%%%%%%%%%%%%%%%%%%%%%%%%%%%%%
\clearpage
\appendix
\section*{Appendix}
\appendixProofText

\end{document}